\documentclass[3p]{els2article}

\usepackage[utf8]{inputenc}

\usepackage{graphicx}
\usepackage{graphics}
\usepackage{latexsym}
\usepackage{textcomp}
\usepackage{amsmath}
\usepackage{amsfonts}
\usepackage{amssymb}
\usepackage{amsthm}
\usepackage{bussproofs}
\usepackage{stmaryrd}
\usepackage[all]{xy}
\usepackage{cmll}

\usepackage{boxedminipage}
\usepackage{tabularx}
\usepackage{url}
\usepackage{proof}
\usepackage{bussproofs}
\usepackage{latexsym}
\usepackage{stmaryrd}
\usepackage{enumerate}
\usepackage{pdflscape}

\usepackage{multicol}

\def\bs{\boldsymbol}

\def\ob{\langle}
\def\cb{\rangle}
\def\ra{\rightarrow}

\def\m{\mathcal}

\def\cf{\boxright}
\def\cp{\preccurlyeq}

\DeclareSymbolFont{symbolsC}{U}{txsyc}{m}{n}
\DeclareMathSymbol{\strictif}{\mathrel}{symbolsC}{74}
\DeclareMathSymbol{\boxright}{\mathrel}{symbolsC}{128}
\DeclareMathSymbol{\Diamondright}{\mathrel}{symbolsC}{132}
\DeclareMathSymbol{\boxRight}{\mathrel}{symbolsC}{136}
\DeclareMathSymbol{\DiamondRight}{\mathrel}{symbolsC}{140}
\DeclareMathSymbol{\Diamonddot}{\mathrel}{symbolsC}{144}

\newenvironment{dedsystem}
{\centering
 \newcolumntype{Y}{>{\small\centering\arraybackslash}X}
 \setlength{\extrarowheight}{4ex}
 \tabularx{\textwidth}{|YYY|}}
{\endtabularx}

\newdefinition{definition}{Definition}

\newtheorem{theorem}{Theorem}
\newtheorem{lemma}{Lemma}
\newtheorem{corollary}{Corollary}

\title{PUC-Logic}

\author[cds]{Ricardo Queiroz de Araujo Fernandes\fnref{fn1}}                                                                            
\ead{ricardo@cds.eb.mil.br}

\author[di]{Edward Hermann Haeusler}                   
\ead{hermann@inf.puc-rio.br}     

\author[df]{Luiz Carlos Pereira}                   
\ead{luiz@inf.puc-rio.br}                                       
                                                                                            
\fntext[fn1]{Thanks to PUC-Rio for the VRac sponsor. Thanks to DAAD (Germany) for the Specialist Literature Programme.}

\address[cds]{Centro de Desenvolvimento de Sistemas\\QG do Exército, Bloco G, 2º Andar\\Setor Militar Urbano, CEP 70630-901\\Brasília, DF, Brasil}

\address[di]{Departamento de Inform\'atica}

\address[df]{Departamento de Filosofia\\Pontif\'icia Universidade Cat\'olica do Rio de Janeiro\\Rua Marquês de São Vicente, 225, Gávea, CEP 22453-900\\Rio de Janeiro, RJ, Brasil}
                                                             
\begin{document}

\begin{abstract} 
We present a logic for Proximity-based Understanding of Conditionals (PUC-Logic) that unifies the Counterfactual and Deontic logics proposed by David Lewis. We also propose a natural deduction system (PUC-ND) associated to this new logic. This inference system is proven to be sound, complete, normalizing and decidable. The relative completeness for the $\bs{V}$ and $\bs{CO}$ logics is shown to emphasize the unified approach over the work of Lewis.
\end{abstract}


\maketitle

\section{Counterfactuals}

\begin{itemize}
\item If Oswald did not kill Kennedy, then someone else did.
\item If Oswald had not killed Kennedy, then someone else would have.\cite{Lewis}
\end{itemize}

The phrases above are respectively instances of the indicative and the subjunctive conditionals. The indicative conditional is associated to the material implication, whereas the subjunctive construction of the language is traditionally studied by the philosophy as the counterfactual  conditional\cite{Lewis,Goodman} or the counterfactual for short.

Conditional propositions involve two components, the antecedent and the consequent. Counterfactual conditionals differ from material implication in a subtle way. The truth of a material implication is based on the actual state-of-affairs. From the knowledge that Kennedy was killed, we can accept the truth of the phrase. On the other hand, a counterfactual conditional should take into account the truth of the antecedent, even if it is not the case. The truth of the antecedent is mandatory in this analysis.

Some approaches to counterfactuals entail belief revision, particularly those based on {\em Ramsey} test evaluation \cite{Ramsey}. In this analysis, the truth value of a counterfactual is considered within a minimal change generated by admitting the antecedent true\cite{Goodman}.

A possible way to circumvent belief revision mechanisms is to consider alternative (possible) state-of-affairs, considered here as worlds, and, based on some accessibility notion, choose the closest one among the worlds that satisfy the antecedent. If the consequent is true at this considered world, then the counterfactual is also true\cite{Lewis}.

Both conditionals have false antecedents and false consequents in the current state-of-affairs. However, the second conditional is clearly false, since we found no reason to accept that, in the closest worlds in which Kennedy is not killed by Oswald, Kennedy is killed by someone else.

We choose the approach of Lewis\cite{Lewis} in our attempt to formalize an inference system for counterfactuals because his accessibility relation leaves out the discussion for a general definition of similitude among worlds, which is considered as given in his analysis.

It also opened the possibility for a contribution in the other way. If we found some general properties in his accessibility relation, considering the evaluation of the formulas in the counterfactual reasoning, we could sketch some details of the concepts of similitude.

\section{Lewis analysis}

Lewis, in the very first page of his book\cite{Lewis}:
\begin{quote}
"\textit{If kangaroos had no tails, they would topple over}, seems to me to mean something like this: in any possible state of affairs in which kangaroos have no tails, and which resembles our actual state of affairs as much as kangaroos having no tails permits it to, the kangaroos topple over."
\end{quote}

We can observe that the word ``resemble'' may be seen as a reference to the concept of similarity between some possible state of affairs in relation to the actual state of affairs. The expression ``as much as'' here may be understood as a relative comparison of similarities among the possible states of affairs in relation to the actual state of affairs. But Lewis gave no formal definition of similarity in his book\cite{Lewis}.

He defined two basic counterfactual conditional operators:
\begin{itemize}
\item $A \cf B$: If it were the case that A, then it would be the case that B;
\item $A \Diamondright B$: If it were the case that A, then it might be the case that B.
\end{itemize}
\noindent And provided also the definition of other counterfactual operators. But, since they are interdefinable, he took $\cf$ as the primitive for the construction of formulas.

In the middle of his book, he introduced the comparative possibility operators and showed that they can serve as the primitive notion for counterfactuals.
\begin{itemize}
\item $A \cp B$: It is as possible that A as it is that B.
\end{itemize}
\noindent This operator gave us simpler proofs during this work.

He used possible-world semantics for intentional logic. For that reason the state of affairs are treated as worlds. To express similarity, he used proximity notions: a world is closer to the actual world in comparison to other worlds if it is more similar to the actual world than other considered worlds.

Lewis called the set of worlds to be considered for an evaluation as the strictness of the conditional. He pointed out that the strictness of the counterfactual conditional is based on the similarity of worlds. He showed that the counterfactual could not be treated by strict conditionals, necessity operators or possibility operators given by modal logics. To do so, he argued that strictness of the conditional can not be given before all evaluations. He constructed sequences of connected counterfactuals in a single English sentence for which the strictness cannot be given for the evaluation:

\begin{quote}
"If Otto had come, it would have been a lively party; but if both Otto and Anna had come it would have been a dreary party; but if Waldo had come as well, it would have been lively; but..."
\end{quote}

\noindent to show that the strictness of the counterfactuals cannot be defined by the context, because the sentence provides a single context for the evaluation of all counterfactuals. If we try to fix a strictness that makes a counterfactual true, then the next counterfactual is made false.

Lewis proposed a variably strict conditional, in which different degrees of strictness is given for every world before the evaluation of any counterfactual. To express this concept, the accessibility relation is defined by a system of spheres, which is given for every world by a nesting function $\$$ that applies over a set of worlds $\m{W}$. The nested function attributes a set of non-empty sets of worlds for each world and this set of sets is in total order for the inclusion relation.

\begin{figure}[htb]
\begin{center}
\includegraphics[height=4cm,width=4cm]{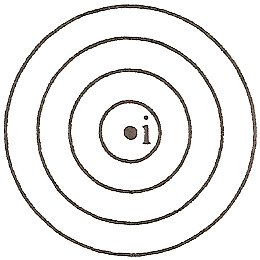}
\caption{A system of spheres around some world $i$}
\end{center}
\end{figure}

A systems of spheres, of any kind, is central in the most traditional analysis of counterfactuals. But the idea behind it is also available to many different logics. So, if we manage to handle them in a satisfactory manner, we will be able to use it in a broader class of logics. The system of neighbourhoods facilitates the development of the model, by leaving open the choice for a proper definition of similarity. And that concept can be used for a broader class of logics, not only the counterfactuals.

From Lewis definitions, the nesting function is a primitive notion:
\begin{quote}
$\phi \cf \psi$ is true at a world $i$ (according to a system of spheres $\$$) if and only if either: no $\phi$-world belongs to any sphere $S$ in $\$_{i}$\footnote{$\$_{i}$ gives the neighbourhoods around the world $i$. They are the available strictness to evaluate counterfactuals at $i$.}, or some sphere $S$ in $\$_{i}$ does contains at least one $\phi$-world, and $\phi \ra \psi$ holds at every world in $S$.
\end{quote}

\begin{quote}
$\phi \cp \psi$ is true at a world $i$ (according to a system of spheres $\$$) if and only if, for every sphere $S$ in $\$_{i}$, if $S$ contains any $\psi$-world then $S$ contains a $\phi$-world.
\end{quote}

Lewis\cite{Lewis} also provided conditions that may be applied to the nesting function $\$$. To every condition corresponds a different counterfactual logic:\begin{itemize}
\item Normality (N): $\$$ is normal iff $\forall w \in \m{W} : \$(w) \neq \emptyset$;
\item Total reflexivity (T): $\$$ is totally reflexive iff $\forall w \in \m{W} : w \in \bigcup\$(w)$;
\item Weak centering (W): $\$$ is weakly centered iff $\forall w \in \m{W} : \$(w) \neq \emptyset \mbox{ and }\forall N \in \$(w) : w \in N$ ;
\item Centering (C): $\$$ is centered iff $\forall w \in \m{W} : \{w\} \in \$(w)$;
\item Limit Assumption (L): $\$$ satisfies the Limit Assumption iff, for any world $w$ and any formula $\phi$, if there is some $\phi$-world\footnote{A $\phi$-world is a world in which $\phi$ holds.} in $\bigcup\$(w)$, then there is some smallest sphere of $\$(w)$ that contains a $\phi$-world;
\item Stalnaker's Assumption (A): $\$$ satisfies Stalnaker's Assumption iff, for any world $w$ and any formula $\phi$, if there is some $\phi$-world in $\bigcup\$(w)$, then there is some sphere of $\$(w)$ that contains exactly one $\phi$-world;
\item Local Uniformity (U-): $\$$ is locally uniform iff for any world $w$ and any $v \in \bigcup\$(w)$, $\bigcup\$(w)$ and $\bigcup\$(v)$ are the same;
\item Uniformity (U): $\$$ is uniform iff for any worlds $w$ and $v$, $\bigcup\$(w)$ and $\bigcup\$(v)$ are the same;
\item Local absoluteness (A-): $\$$ is locally absolute iff for any world $w$ and any $v \in \bigcup\$(w)$, $\$(w)$ and $\$(v)$ are the same;
\item Absoluteness (A): $\$$ is absolute iff for any worlds $w$ and $v$, $\$(w)$ and $\$(v)$ are the same.
\end{itemize}

The $\bs{V}$-logic is the most basic counterfactual logic presented by Lewis\cite{Lewis}, where ``V'' stands for variably strict conditional. If, for example, we accept the centering condition (C), then we have the $\bs{VC}$-logic. Lewis showed in his book a chart of 26 non-equivalent $\bs{V}$-logics that arises from the combinations of the conditions.

We prefer to call the spheres as neighbourhoods, because they represent better the concept of proximity, which Lewis used to express similarity. The neighbourhoods provide a relative way to compare distance. The world that is contained in a neighbourhood is closer to the actual world than other world that is not contained in that same neighbourhood.

As far as we know, there is only one natural deduction system for the counterfactuals, which is given by Bonevac \cite{Bonevac}. But his system is designed to deals with the $\bs{VW}$-logic, since it contains the rule of counterfactual exploitation ($\cf$E), which encapsulates the weak centering condition. His approach to define rules for the counterfactual operators provides a better intuition of the counterfactual logic. His systems is expressive enough to deal with modalities and strict conditionals. The labelling of world shifts using formulas make easier to capture the counterfactual mechanics.

We also found the work of Sano \cite{Sano} which pointed out the advantages of using the hybrid formalism for the counterfactual logic. He presented some axioms and rules for the $\bs{V_{\m{HC}(@)}}$-logic that extends the $\bs{V}$-logic of Lewis.

Another interesting reference is the article of Gent\cite{Gent}, which presents a new sequent- or tableaux-style proof system for V C. His work depends on the operator $\llbracket \rrbracket$ and the definition of signed formulas.

We recently found a sequent calculus, provided by Lellmann\cite{Jelia2012}, that treats the $\bs{V}$-logic of Lewis and its extensions. Its language depends on modal operators, specially the counterfactual operators $\cf$ and $\boxRight$ and the comparative possibility operator $\cp$.

As far as we know, our deduction system is the only one dealing with Lewis systems in a general form, that is, without using modalities in the syntax and treating the most basic counterfactual $\bs{V}$-logic.

\section{Proximity-based Understanding of Conditionals}

In \cite{LSFA09}, we presented a sequent calculus for counterfactual logic based on a Local Set Theory \cite{Bell}. In this article, we defined the satisfaction relation for worlds, for sets of the worlds and for neighbourhoods, where we encapsulated some quantifications that made it easier to express the operators with fewer quantifiers. But the encapsulation made the the inference system to have no control of the quantifications. Here we propose a logic for Proximity-based Understanding of Conditionals, PUC-Logic for short, that take control of the quantifications with labels.

\begin{definition}
Given a non-empty set $\m{W}$ (considered the set of worlds), we define a nesting function $\$$ that assigns to each world of $\m{W}$ a set of nested sets of $\m{W}$. A set of nested sets is a set of sets in which the inclusion relation among sets is a total order.
\end{definition}

\begin{definition}
A \textit{frame} is a tuple $\m{F} = \ob \m{W} , \$ , \m{V} \cb$, in which $\m{V}$ is a truth assignment function for each atomic formula with image on the subsets of $\m{W}$. A \textit{model} is a pair $\m{M} = \ob \m{F} , \chi \cb$, $\m{F}$  a frame and $\chi$ a world of $\m{W}$, called the \textit{reference world} of the model. A \textit{template} is a pair $\m{T} = \ob \m{M} , N \cb$, $N \in \$(\chi)$ and $N$ is called the \textit{reference neighbourhood} of the template.
\end{definition}

We use the term \textit{structure} to refer a model or a template.

\begin{definition}
A structure is \textit{finite} if its set of worlds is finite.
\end{definition}

We now define a relation between structures to represent the pertinence of neighbourhoods in a neighbourhood system of a world and the pertinence of worlds in a given neighbourhood.

\begin{definition}
Given a model $\m{M} = \ob \m{W} , \$ , \m{V} , \chi \cb$, then, for any $N \in \$(\chi)$, the template $\m{T} = \ob \m{W} , \$ , \m{V} , \chi , N \cb$ is in \textit{perspective} relation to $\m{M}$. We represent this by $\m{M} \multimap \m{T}$. Given a template $\m{T} = \ob \m{W} , \$ , \m{V} , \chi , N \cb$, then, for any $w \in N$, the model $\m{M} = \ob \m{W} , \$ , \m{V} , w \cb$ is in perspective relation to $\m{T}$. We represent this by $\m{T} \multimap \m{M}$.
\end{definition}

\begin{definition}
The concatenation of $n$ tuples of the perspective relation is called a path of size $n$ and is represented by the symbol $\multimap_{n}$.
\end{definition}

One remark: if the size of a path is even, then a model is related to another model or a template is related to another template.

\begin{definition}
The transitive closure of the perspective relation is called the \textit{projective} relation, which is represented by the symbol $\leadsto$.
\end{definition}

\begin{definition}
Given a world $\chi$ and the nested neighbourhood function $\$$ we can build a sequence of sets of worlds:
\begin{enumerate}
\item $\bigtriangleup^{\$}_{0}(\chi) = \{\chi\}$;
\item $\bigtriangleup^{\$}_{k+1}(\chi) = \bigcup_{w \in \bigtriangleup^{\$}_{k}(\chi)} (\bigcup \$(w))$, $k \geq 0$.
\end{enumerate}
Let $\bigtriangleup^{\$}(\chi) = \bigcup_{n \in \mathbb{N}} \bigtriangleup^{\$}_{n}(\chi)$ and $\bigtriangleup^{\$}_{\vec{n}}(\chi) = \bigcup_{0 \leq m \leq n} \bigtriangleup^{\$}_{m}(\chi)$.
\end{definition}

We introduce labels in our language, in order to syntactically represent quantifications over two specific domains: neighbourhoods and worlds. So, for that reason, a label may be a neighbourhood label or a world label:
\begin{itemize}
\item Neighbourhood labels:
\begin{itemize}
\item[($\circledast$)] Universal quantifier over neighbourhoods of some neighbourhood system;
\item[($\circledcirc$)] Existential quantifier over neighbourhoods of some neighbourhood system;
\item[($N$)] Variables (capital letters) that may denote some neighbourhood of some neighbourhood system.
\end{itemize}
\item World labels:
\begin{itemize}
\item[($\ast$)] Universal quantifier over worlds of some neighbourhood;
\item[($\bullet$)] Existential quantifier over worlds of some neighbourhood;
\item[($u$)] Variables (lower case letters) that denote some world of some neighbourhood.
\end{itemize}
\end{itemize}

We denote the set of neighbourhood labels by $\bs{L}_{n}$ and the set of world labels by $\bs{L}_{w}$.

\begin{definition}
The language of PUC-Logic consists of:
\begin{itemize}
\item countably neighbourhood variables: $N,M,L,\ldots$;
\item countably world variables: $w,z,\ldots$;
\item countably proposition symbols: $p_{0},p_{1},\ldots$;
\item countably proposition constants: $\top_{n},\bot_{n},\top_{w},\bot_{w},\shneg N, \shpos N, \shneg M, \shpos M,\ldots$;
\item connectives: $\wedge,\vee,\ra,\neg$;
\item neighbourhood labels: $\circledast, \circledcirc$;
\item world labels: $\ast,\bullet$;
\item auxiliary symbols: $(,)$.
\end{itemize}
\end{definition}

As in the case of labels, we want to separate the sets of well-formed formulas into two disjoint sets, according to sort of label that labels the formula. We denote the set of neighbourhood formulas by $\bs{F}_{n}$ and the set of world formulas by $\bs{F}_{w}$.

\begin{definition} \label{wffDefinition}
The sets $\bs{F}_{n}$ and $\bs{F}_{w}$ of well-formed formulas\footnote{We use the term wff to denote both the singular and the plural form of the expression well-formed formula.} are constructed the following rules:
\begin{enumerate}
\item $\top_{n},\bot_{n} \in \bs{F}_{n}$;
\item $\top_{w},\bot_{w} \in \bs{F}_{w}$;
\item $\shneg N,\shpos N \in \bs{F}_{w}$, for every neighbourhood variable $N$;
\item $\alpha \in \bs{F}_{n}$, for every atomic formula $\alpha$, except $\top$ and $\bot$;
\item if $\alpha \in \bs{F}_{n}$, then $\neg \alpha \in \bs{F}_{n}$;
\item if $\alpha \in \bs{F}_{w}$, then $\neg \alpha \in \bs{F}_{w}$;
\item if $\alpha,\beta \in \bs{F}_{n}$, then $\alpha \wedge \beta,\alpha \vee \beta,\alpha \ra \beta \in \bs{F}_{n}$;
\item if $\alpha,\beta \in \bs{F}_{w}$, then $\alpha \wedge \beta,\alpha \vee \beta,\alpha \ra \beta \in \bs{F}_{w}$;
\item if $\alpha \in \bs{F}_{n}$ and $\phi \in \bs{L}_{w}$, then $\alpha^{\phi} \in \bs{F}_{w}$;
\item if $\alpha \in \bs{F}_{w}$ and $\phi \in \bs{L}_{n}$, then $\alpha^{\phi} \in \bs{F}_{n}$.\end{enumerate}\end{definition}

We introduced the two formulas for true and false, in order to make the sets of formulas disjoint. The formula $\shneg N$ is introduced to represent that a neighbourhood contains the neighbourhood $N$ and the formula $\shpos N$ represent a neighbourhood is contained in $N$.

The last two rules of definition \ref{wffDefinition} introduces the labelling the formulas. Moreover, since we can label a labelled formula, every formula has a stack of labels that represent nested labels. We call it the \textit{attribute} of the formula. The top label of the stack is the \textit{index} of the formula. We represent the attribute of a formula as a letter that appear to the right of the formula. If the attribute is empty, we may omit it and the formula has no index. The attribute of some formula will always be empty if the last rule, used to build the formula, is not one of the labelling rules, as in the case of $((\alpha \ra \alpha)^{\circledast,\bullet})\vee(\gamma^{\circledcirc,\ast})$.

To read a labelled formula, it is necessary to read its index first and then the rest of the formula. For example, $(\alpha \ra \alpha)^{\circledast,\bullet}$ should be read as: there is some world, in all neighbourhoods of the considered neighbourhood system, in which it is the case that $\alpha \ra \alpha$.

We may concatenate stacks of labels and labels, using commas, to produce a stack of labels that is obtained by respecting the order of the labels in the stacks and the order of the concatenation, like $\alpha^{\Sigma,\Delta}$, where $\alpha$ is a formula and $\Sigma$ and $\Delta$ are stacks of labels. But we admit no nesting of attributes, which means that $(\alpha^{\Sigma})^{\Delta}$ is the same as $\alpha^{\Sigma,\Delta}$.

\begin{definition}
Given a stack of labels $\Sigma$, we define $\overline{\Sigma}$ as the stack of labels that is obtained from $\Sigma$ by reversing the order of the labels in the stack.\end{definition}\begin{definition}
Given a stack of labels $\Sigma$, the size $s(\Sigma)$ is its number of labels.\end{definition}\begin{definition}
Given a set of worlds $\m{W}$, a set of world variables and a set of neighbourhood variables, we define a variable assignment function $\sigma$, that assigns a world of $\m{W}$ to each world variable and a non-empty set of $\m{W}$ to each neighbourhood variable.\end{definition}\begin{definition} \label{satisfactionDefinition}
Given a variable assignment function $\sigma$, the relation $\models$ of \textit{satisfaction} between formulas, models and templates is given by:
\begin{enumerate}
\item $\ob \m{W} , \$ , \m{V} , \chi \cb \models \alpha$, $\alpha$ atomic, iff: $\chi \in \m{V}(\alpha)$. For every world $w \in \m{W}$, $w \in \m{V}(\top_{n})$ and $w \not\in \m{V}(\bot_{n})$;
\item $\ob \m{W} , \$ , \m{V} , \chi \cb \models \neg \; (\alpha^{\Sigma})$ iff: $\neg \; (\alpha^{\Sigma}) \in \bs{F}_{n}$ and $\ob \m{W} , \$ , \m{V} , \chi \cb \not\models \alpha^{\Sigma}$;
\item $\ob \m{W} , \$ , \m{V} , \chi \cb \models \alpha^{\Sigma} \wedge \beta^{\Omega}$ iff: $\alpha^{\Sigma} \wedge \beta^{\Omega} \in \bs{F}_{n}$ and\\$(\;\ob \m{W} , \$ , \m{V} , \chi \cb \models \alpha^{\Sigma} \mbox{ and } \ob \m{W} , \$ , \m{V} , \chi \cb \models \beta^{\Omega}$;
\item $\ob \m{W} , \$ , \m{V} , \chi \cb \models \alpha^{\Sigma} \vee \beta^{\Omega}\;)$ iff: $\alpha^{\Sigma} \vee \beta^{\Omega} \in \bs{F}_{n}$ and\\$(\; \ob \m{W} , \$ , \m{V} , \chi \cb \models \alpha^{\Sigma} \mbox{ or } \ob \m{W} , \$ , \m{V} , \chi \cb \models \beta^{\Omega} \;)$;
\item $\ob \m{W} , \$ , \m{V} , \chi \cb \models \alpha^{\Sigma} \ra \beta^{\Omega}$ iff: $\alpha^{\Sigma} \ra \beta^{\Omega} \in \bs{F}_{n}$ and\\$(\; \ob \m{W} , \$ , \m{V} , \chi \cb \models \neg (\alpha^{\Sigma}) \mbox{ or } \ob \m{W} , \$ , \m{V} , \chi \cb \models \beta^{\Omega} \;)$;
\item $\ob \m{W} , \$ , \m{V} , \chi \cb \models \alpha^{\Sigma,\circledast}$ iff: $\forall N \in \$(\chi) : \ob \m{W} , \$ , \m{V} , \chi , N \cb \models \alpha^{\Sigma}$;
\item $\ob \m{W} , \$ , \m{V} , \chi \cb \models \alpha^{\Sigma,\circledcirc}$ iff: $\exists N \in \$(\chi) : \ob \m{W} , \$ , \m{V} , \chi , N \cb \models \alpha^{\Sigma}$;
\item $\ob \m{W} , \$ , \m{V} , \chi \cb \models \alpha^{\Sigma,N}$ iff: $\ob \m{W} , \$ , \m{V} , \chi , \sigma(N) \cb \models \alpha^{\Sigma}$;
\item $\ob \m{W} , \$ , \m{V} , \chi , N \cb \models \shneg M$ iff: $\sigma(M) \in \$(\chi)$ and $\sigma(M) \subset N$;
\item $\ob \m{W} , \$ , \m{V} , \chi , N \cb \models \shpos M$ iff: $\sigma(M) \in \$(\chi)$ and $N \subset \sigma(M)$;
\item $\ob \m{W} , \$ , \m{V} , \chi , N \cb \models \alpha^{\Sigma,\ast}$ iff: $\forall w \in N : \ob \m{W} , \$ , \m{V} , w \cb \models \alpha^{\Sigma}$;
\item $\ob \m{W} , \$ , \m{V} , \chi , N \cb \models \alpha^{\Sigma,\bullet}$ iff: $\exists w \in N : \ob \m{W} , \$ , \m{V} , w \cb \models \alpha^{\Sigma}$;
\item $\ob \m{W} , \$ , \m{V} , \chi , N \cb \models \alpha^{\Sigma,u}$ iff: $\sigma(u) \in N$ and $\ob \m{W} , \$ , \m{V} , \sigma(u) \cb \models \alpha^{\Sigma}$;
\item $\ob \m{W} , \$ , \m{V} , \chi , N \cb \models \neg \; (\alpha^{\Sigma})$ iff: $\neg \; (\alpha^{\Sigma}) \in \bs{F}_{w}$ and $\ob \m{W} , \$ , \m{V} , \chi , N \cb \not\models \alpha^{\Sigma}$;
\item $\ob \m{W} , \$ , \m{V} , \chi , N \cb \models \alpha^{\Sigma} \wedge \beta^{\Omega}$ iff: $\alpha^{\Sigma} \wedge \beta^{\Omega} \in \bs{F}_{w}$ and\\$(\;\ob \m{W} , \$ , \m{V} , \chi , N \cb \models \alpha^{\Sigma} \mbox{ and }\ob \m{W} , \$ , \m{V} , \chi , N \cb \models \beta^{\Omega}\;)$;
\item $\ob \m{W} , \$ , \m{V} , \chi , N \cb \models \alpha^{\Sigma} \vee \beta^{\Omega}$ iff: $\alpha^{\Sigma} \vee \beta^{\Omega} \in \bs{F}_{w}$ and\\$(\; \ob \m{W} , \$ , \m{V} , \chi , N \cb \models \alpha^{\Sigma} \mbox{ or } \ob \m{W} , \$ , \m{V} , \chi , N \cb \models \beta^{\Omega} \;)$;
\item $\ob \m{W} , \$ , \m{V} , \chi , N \cb \models \alpha^{\Sigma} \ra \beta^{\Omega}$ iff: $\alpha^{\Sigma} \ra \beta^{\Omega} \in \bs{F}_{w}$ and\\$(\; \ob \m{W} , \$ , \m{V} , \chi , N \cb \models \neg (\alpha^{\Sigma})  \mbox{ or } \ob \m{W} , \$ , \m{V} , \chi , N \cb \models \beta^{\Omega} \;)$;
\item $\ob \m{W} , \$ , \m{V} , \chi , N \cb \models \top_{w}$ and $\ob \m{W} , \$ , \m{V} , \chi , N \cb \not\models \bot_{w}$, for every template.\end{enumerate}\end{definition}

\begin{definition} \label{logicalConsequenceDf}
The relation $\alpha^{\Sigma} \models \beta^{\Omega}$ of \textit{logical consequence} is defined iff $\alpha^{\Sigma},\beta^{\Omega} \in \bs{F}_{n}$ and for all model $\m{M} \models \alpha^{\Sigma}$, we have $\m{M} \models \beta^{\Omega}$. The relation is also defined iff $\alpha^{\Sigma},\beta^{\Omega} \in \bs{F}_{w}$ and for all template $\m{T} \models \alpha^{\Sigma}$, we have $\m{T} \models \beta^{\Omega}$. Given $\Gamma \cup \{\alpha^{\Sigma}\} \subset \bs{F}_{n}$, the relation $\Gamma \models \alpha^{\Sigma}$ of logical consequence is defined iff for all model $\m{M}$ that satisfies every formula of $\Gamma$, $\m{M} \models \alpha^{\Sigma}$. Given $\Gamma \cup \{\alpha^{\Sigma}\} \subset \bs{F}_{w}$, the relation $\Gamma \models \alpha^{\Sigma}$ is defined iff for all templates $\m{T}$ that satisfies every formula of $\Gamma$, $\m{T} \models \alpha^{\Sigma}$.\end{definition}

\begin{definition}
$\alpha^{\Sigma} \in \bs{F}_{n}$ ($\in \bs{F}_{w}$) is a n-tautology (w-tautology) iff for every model (template) $\m{M} \models \alpha^{\Sigma}$ ($\m{T} \models \alpha^{\Sigma}$).\end{definition}

\begin{lemma} \label{taut}
$\alpha^{\Sigma}$ is a n-tautology iff $\alpha^{\Sigma,\ast,\circledast}$ is a n-tautology.
\end{lemma}

\begin{proof}
If $\alpha^{\Sigma}$ is a n-tautology, $\forall z \in \m{W}$, $\ob \m{W} , \$ , \m{V} , z \cb \models \alpha^{\Sigma}$. In particular, given a world $\chi \in \m{W}$, $\forall N \in \$(\chi) : \forall w \in N : \ob \m{W} , \$ , \m{V} , w \cb \models \alpha^{\Sigma}$ and, by definition, $\ob \m{W} , \$ , \m{V} , \chi \cb \models \alpha^{\Sigma,\ast,\circledast}$ for every world of $\m{W}$ and $\alpha^{\Sigma,\ast,\circledast}$ is also a n-tautology. Conversely, if $\alpha^{\Sigma,\ast,\circledast}$ is a n-tautology, then $\forall N \in \$(\chi) : \forall w \in N : \ob \m{W} , \$ , \m{V} , w \cb \models \alpha^{\Sigma}$ for every choice of $\m{W}$ , $\$$ , $\m{V}$ and $w$. So, given $\m{W}$ , $\m{V}$ and $w$, we can choose $\$$ to be the constant function $\{\m{W}\}$. So, $\forall z \in \m{W}$, $\ob \m{W} , \$ , \m{V} , z \cb \models \alpha^{\Sigma}$ and $\alpha^{\Sigma}$ must also be a n-tautology.\end{proof}

The relation defined below is motivated by the fact that, if a model $\m{M}$ satisfies a formula like $\alpha^{\circledast,\ast}$, then for every template $\m{T}$, such that $\m{M} \multimap \m{T}$, $\m{T}$ satisfies $\alpha^{\circledast}$ by definition. And also for every model $\m{H}$, such that $\m{M} \multimap_{2} \m{H}$, $\m{H}$ satisfies $\alpha$ by definition.

\begin{definition} \label{referentialConsequence}
Given a model $\m{M}$, called \textit{the reference model}, the relation $\alpha^{\Sigma} \models_{\m{M}:n} \beta^{\Omega}$ of \textit{referential consequence} is defined iff:
\begin{itemize}
\item $n > 0$ and ($\m{M} \models \alpha^{\Sigma}$ implies $\m{H} \models \beta^{\Omega}$, for any structure $\m{M} \multimap_{n} \m{H}$);
\item $n = 0$ and (if $\m{M} \models \alpha^{\Sigma}$ implies $\m{M} \models \beta^{\Omega}$).
\end{itemize} Given $\Gamma \cup \{\alpha^{\Sigma}\} \subset \bs{F}_{n}$, $\Gamma \models_{\m{M}:n} \alpha^{\Sigma}$ iff:
\begin{itemize}
\item $n > 0$ and ($\m{H} \models \alpha^{\Sigma}$, for any structure $\m{M} \multimap_{n} \m{H}$ that satisfies every formula of $\Gamma$);
\item $n = 0$ and ($\m{M} \models \alpha^{\Sigma}$ if $\m{M}$ satisfies every formula of $\Gamma$).
\end{itemize}\end{definition}

\begin{figure}[htbp] \label{theSystem}
\begin{dedsystem}
\hline

\AxiomC{$\Pi$} \RightLabel{$\Delta$}
\UnaryInfC{$\alpha^{\Sigma} \wedge \beta^{\Omega}$} \LeftLabel{1:} \RightLabel{$\Delta$}
\UnaryInfC{$\alpha^{\Sigma}$}  \DisplayProof &

\AxiomC{$\Pi$} \RightLabel{$\Delta$}
\UnaryInfC{$\alpha^{\Sigma} \wedge \beta^{\Omega}$} \LeftLabel{2:} \RightLabel{$\Delta$}
\UnaryInfC{$\beta^{\Omega}$}  \DisplayProof &

\AxiomC{$\Pi_{1}$} \RightLabel{$\Delta$}
\UnaryInfC{$\alpha^{\Sigma}$}
\AxiomC{$\Pi_{2}$} \RightLabel{$\Delta$}
\UnaryInfC{$\beta^{\Omega}$} \LeftLabel{3:} \RightLabel{$\Delta$}
\BinaryInfC{$\alpha^{\Sigma} \wedge \beta^{\Omega}$} \DisplayProof \\

\AxiomC{$\Pi$} \RightLabel{$\Delta$}
\UnaryInfC{$\alpha^{\Sigma}$} \LeftLabel{4:} \RightLabel{$\Delta$}
\UnaryInfC{$\alpha^{\Sigma} \vee \beta^{\Omega}$} \DisplayProof  &

\AxiomC{$\Pi_{1}$} \RightLabel{$\Delta$}
\UnaryInfC{$\alpha^{\Sigma} \vee \beta^{\Omega}$}
\AxiomC{$\phantom{-}$}
\alwaysNoLine
\UnaryInfC{[$\alpha^{\Sigma}$]} \RightLabel{$\Delta$}
\alwaysSingleLine
\UnaryInfC{$\Pi_{2}$} \RightLabel{$\Theta$}
\UnaryInfC{$\gamma^{\Lambda}$}
\AxiomC{$\phantom{-}$}
\alwaysNoLine
\UnaryInfC{[$\beta^{\Omega}$]} \RightLabel{$\Delta$}
\alwaysSingleLine
\UnaryInfC{$\Pi_{3}$} \RightLabel{$\Theta$}
\UnaryInfC{$\gamma^{\Lambda}$} \LeftLabel{5:} \RightLabel{$\Theta$}
\alwaysSingleLine
\TrinaryInfC{$\gamma^{\Lambda}$} \DisplayProof &

\AxiomC{$\Pi$} \RightLabel{$\Delta$}
\UnaryInfC{$\beta^{\Omega}$} \LeftLabel{6:} \RightLabel{$\Delta$}
\UnaryInfC{$\alpha^{\Sigma} \vee \beta^{\Omega}$} \DisplayProof \\

\AxiomC{$\phantom{-}$}
\alwaysNoLine
\UnaryInfC{[$\neg (\alpha^{\Sigma})$]} \RightLabel{$\Delta$}
\alwaysSingleLine
\UnaryInfC{$\Pi$} \RightLabel{$\Delta$}
\UnaryInfC{$\bot$} \LeftLabel{7:} \RightLabel{$\Delta$}
\UnaryInfC{$\alpha^{\Sigma}$} \DisplayProof &

\AxiomC{$\Pi$} \RightLabel{$\Delta$}
\UnaryInfC{$\bot$} \LeftLabel{8:} \RightLabel{$\Delta$}
\UnaryInfC{$\alpha^{\Sigma}$} \DisplayProof &

\AxiomC{$\Pi$} \RightLabel{$\Delta$}
\UnaryInfC{$\bot$} \LeftLabel{9:}
\UnaryInfC{$\bot_{n}$} \DisplayProof \\

\AxiomC{$\alpha^{\Sigma}$} \LeftLabel{10:} \RightLabel{$\Delta$}
\UnaryInfC{$\alpha^{\Sigma}$} \DisplayProof &

\AxiomC{$\phantom{-}$}
\alwaysNoLine
\UnaryInfC{[$\alpha^{\Sigma}$]} \RightLabel{$\Delta$}
\alwaysSingleLine
\UnaryInfC{$\Pi$} \RightLabel{$\Delta$}
\UnaryInfC{$\beta^{\Omega}$} 
\alwaysSingleLine \LeftLabel{11:} \RightLabel{$\Delta$}
\UnaryInfC{$\alpha^{\Sigma} \ra \beta^{\Omega}$}
\alwaysNoLine \UnaryInfC{$\phantom{-}$} \DisplayProof &

\AxiomC{$\Pi_{1}$} \RightLabel{$\Delta$}
\UnaryInfC{$\alpha^{\Sigma}$}
\AxiomC{$\Pi_{2}$} \RightLabel{$\Delta$}
\UnaryInfC{$\alpha^{\Sigma} \ra \beta^{\Omega}$} \LeftLabel{12:} \RightLabel{$\Delta$}
\BinaryInfC{$\beta^{\Omega}$} \DisplayProof \\

\AxiomC{$\Pi$} \RightLabel{$\Delta$}
\UnaryInfC{$\alpha^{\Sigma,\phi}$} \LeftLabel{13:} \RightLabel{$\Delta,\phi$}
\UnaryInfC{$\alpha^{\Sigma}$}
\alwaysNoLine \UnaryInfC{$\phantom{-}$} \DisplayProof &

\AxiomC{$\Pi$} \RightLabel{$\Delta,\phi$}
\UnaryInfC{$\alpha^{\Sigma}$} \LeftLabel{14:} \RightLabel{$\Delta$}
\UnaryInfC{$\alpha^{\Sigma,\phi}$}
\alwaysNoLine \UnaryInfC{$\phantom{-}$} \DisplayProof &

\AxiomC{$\Pi$} \RightLabel{$\Delta,u$}
\UnaryInfC{$\alpha^{\Sigma}$} \LeftLabel{15:} \RightLabel{$\Delta,\ast$}
\UnaryInfC{$\alpha^{\Sigma}$}
\alwaysNoLine \UnaryInfC{$\phantom{-}$} \DisplayProof \\

\AxiomC{$\Pi$} \RightLabel{$\Delta,\ast$}
\UnaryInfC{$\alpha^{\Sigma}$} \LeftLabel{16:} \RightLabel{$\Delta,u$}
\UnaryInfC{$\alpha^{\Sigma}$}
\alwaysNoLine \UnaryInfC{$\phantom{-}$} \DisplayProof &

\AxiomC{$\Pi$} \RightLabel{$\Delta,u$}
\UnaryInfC{$\alpha^{\Sigma}$} \LeftLabel{17:} \RightLabel{$\Delta,\bullet$}
\UnaryInfC{$\alpha^{\Sigma}$}
\alwaysNoLine \UnaryInfC{$\phantom{-}$} \DisplayProof &

\AxiomC{$\Pi_{1}$} \RightLabel{$\Delta,\bullet$}
\UnaryInfC{$\alpha^{\Sigma}$}
\AxiomC{[$\alpha^{\Sigma}$]} \RightLabel{$\Delta,u$}
\UnaryInfC{$\Pi_{2}$} \RightLabel{$\Theta$}
\UnaryInfC{$\beta^{\Omega}$} \LeftLabel{18:} \RightLabel{$\Theta$}
\BinaryInfC{$\beta^{\Omega}$} \DisplayProof \\

\AxiomC{$\Pi_{1}$} \RightLabel{$\Delta,N$}
\UnaryInfC{$\alpha^{\Sigma}$}
\AxiomC{$\Pi_{2}$} \RightLabel{$\Delta,\circledcirc$}
\UnaryInfC{$\beta^{\Omega}$} \LeftLabel{19:} \RightLabel{$\Delta,\circledcirc$}
\BinaryInfC{$\alpha^{\Sigma}$}
\alwaysNoLine \UnaryInfC{$\phantom{-}$} \DisplayProof &

\AxiomC{$\Pi_{1}$} \RightLabel{$\Delta,\circledcirc$}
\UnaryInfC{$\alpha^{\Sigma}$}
\AxiomC{[$\alpha^{\Sigma}$]} \RightLabel{$\Delta,N$}
\UnaryInfC{$\Pi_{2}$} \RightLabel{$\Theta$}
\UnaryInfC{$\beta^{\Omega}$} \LeftLabel{20:} \RightLabel{$\Theta$}
\BinaryInfC{$\beta^{\Omega}$} \DisplayProof &

\AxiomC{$\Pi$} \RightLabel{$\Delta,N$}
\UnaryInfC{$\alpha^{\Sigma}$} \LeftLabel{21:} \RightLabel{$\Delta,\circledast$}
\UnaryInfC{$\alpha^{\Sigma}$} \DisplayProof \\

\AxiomC{$\Pi_{1}$} \RightLabel{$\Delta,\circledast$}
\UnaryInfC{$\alpha^{\Sigma}$}
\AxiomC{$\Pi_{2}$} \RightLabel{$\Delta,N$}
\UnaryInfC{$\beta^{\Omega}$} \LeftLabel{22:} \RightLabel{$\Delta,N$}
\BinaryInfC{$\alpha^{\Sigma}$} \DisplayProof &

\AxiomC{$\Pi_{1}$} \RightLabel{$\Delta,N$}
\UnaryInfC{$\alpha^{\Sigma,\bullet}$}
\AxiomC{$\Pi_{2}$} \RightLabel{$\Delta,M$}
\UnaryInfC{$\shneg N$} \LeftLabel{23:} \RightLabel{$\Delta,M$}
\BinaryInfC{$\alpha^{\Sigma,\bullet}$} \DisplayProof &

\AxiomC{$\Pi_{1}$} \RightLabel{$\Delta,N$}
\UnaryInfC{$\alpha^{\Sigma,\ast}$}
\AxiomC{$\Pi_{2}$} \RightLabel{$\Delta,M$}
\UnaryInfC{$\shpos N$} \LeftLabel{24:} \RightLabel{$\Delta,M$}
\BinaryInfC{$\alpha^{\Sigma,\ast}$}
\alwaysNoLine \UnaryInfC{$\phantom{-}$} \DisplayProof \\

\AxiomC{$\Pi_{1}$} \RightLabel{$\Delta,N$}
\UnaryInfC{$\shneg M$}
\AxiomC{$\Pi_{2}$} \RightLabel{$\Delta,M$}
\UnaryInfC{$\shneg P$} \LeftLabel{25:} \RightLabel{$\Delta,N$}
\BinaryInfC{$\shneg P$}
\alwaysNoLine \UnaryInfC{$\phantom{-}$} \DisplayProof &

\AxiomC{$\Pi_{1}$} \RightLabel{$\Delta,N$}
\UnaryInfC{$\shpos M$}
\AxiomC{$\Pi_{2}$} \RightLabel{$\Delta,M$}
\UnaryInfC{$\shpos P$} \LeftLabel{26:} \RightLabel{$\Delta,N$}
\BinaryInfC{$\shpos P$}
\alwaysNoLine \UnaryInfC{$\phantom{-}$} \DisplayProof &

\AxiomC{[$\shneg M$]} \RightLabel{$\Delta,N$}
\UnaryInfC{$\Pi_{1}$} \RightLabel{$\Theta$}
\UnaryInfC{$\alpha^{\Sigma}$} 
\AxiomC{[$\shneg N$]} \RightLabel{$\Delta,M$}
\UnaryInfC{$\Pi_{2}$} \RightLabel{$\Theta$}
\UnaryInfC{$\alpha^{\Sigma}$} \LeftLabel{27:} \RightLabel{$\Theta$}
\BinaryInfC{$\alpha^{\Sigma}$}
\alwaysNoLine \UnaryInfC{$\phantom{-}$} \DisplayProof \\

\AxiomC{[$\shpos M$]} \RightLabel{$\Delta,N$}
\UnaryInfC{$\Pi_{1}$} \RightLabel{$\Theta$}
\UnaryInfC{$\alpha^{\Sigma}$} 
\AxiomC{[$\shpos N$]} \RightLabel{$\Delta,M$}
\UnaryInfC{$\Pi_{2}$} \RightLabel{$\Theta$}
\UnaryInfC{$\alpha^{\Sigma}$} \LeftLabel{28:} \RightLabel{$\Theta$}
\BinaryInfC{$\alpha^{\Sigma}$}
\alwaysNoLine \UnaryInfC{$\phantom{-}$} \DisplayProof &

\AxiomC{[$\shneg N$]} \RightLabel{$\Delta,M$}
\UnaryInfC{$\Pi_{1}$} \RightLabel{$\Theta$}
\UnaryInfC{$\alpha^{\Sigma}$} 
\AxiomC{[$\shpos N$]} \RightLabel{$\Delta,M$}
\UnaryInfC{$\Pi_{2}$} \RightLabel{$\Theta$}
\UnaryInfC{$\alpha^{\Sigma}$} \LeftLabel{29:} \RightLabel{$\Theta$}
\BinaryInfC{$\alpha^{\Sigma}$}
\alwaysNoLine \UnaryInfC{$\phantom{-}$} \DisplayProof &

\AxiomC{$\phantom{-}$} \LeftLabel{30:} \RightLabel{$\Delta$}
\UnaryInfC{$\top$} \alwaysNoLine \UnaryInfC{$\phantom{-}$} \DisplayProof \\
 
\hline
\end{dedsystem}
\caption{Natural Deduction System for PUC-Logic (PUC-ND)}
\end{figure}

Every rule of PUC-ND has a stack of labels, called its \textit{context}. The scope is represented by a capital Greek letter at the right of each rule. The \textit{scope} of a rule is the top label of its context. Given a context $\Delta$, we denote its scope by $\oc \Delta$. If the context is empty, then there is no scope. As in the case of labels and formulas, we want to separate the contexts into two disjoint sets: $\Delta \in \bs{C}_{n}$ if $\oc \Delta \in \bs{L}_{n}$; $\Delta \in \bs{C}_{w}$ if $\Delta$ is empty or $\oc \Delta \in \bs{L}_{w}$.

\begin{definition} \label{fitDf}
We say that a wff $\alpha^{\Sigma}$ \textit{fits} into a context $\Delta$ iff $\alpha^{\Sigma,\overline{\Delta}} \in \bs{F}_{n}$.
\end{definition}

The wff $\alpha^{\bullet} \ra \beta^{\bullet}$ and $\gamma^{u,\circledast,\ast}$ fit into the context $\{\circledcirc\}$, because $(\alpha^{\bullet} \ra \beta^{\bullet})^{\circledcirc} \in \bs{F}_{n}$ and $\gamma^{u,\circledast,\ast,\circledcirc} \in \bs{F}_{n}$. The wff $\alpha^{\bullet} \vee \beta^{\ast}$ and $\gamma^{\ast,N,u}$ do not fit into the context $\{\circledcirc,\ast\}$, because $(\alpha^{\bullet} \vee \beta^{\ast})^{\ast,\circledcirc}$ and $\gamma^{\ast,N,u,\ast,\circledcirc}$ are not wff and, therefore, cannot be in $\bs{F}_{n}$. There is no wff that fits into the context $\{\ast\}$, because the label $\ast \in \bs{L}_{w}$ and the rule of labelling can only include the resulting formula into $\bs{F}_{w}$. We can conclude that if a wff is in $\bs{F}_{n}$, then the context must be in $\bs{C}_{w}$ and the same for $\bs{F}_{w}$ and $\bs{C}_{n}$. The fitting restriction ensures that the conclusion of a rule is always a wff.

Moreover, the definition of fitting resembles the attribute grammar approach for context free languages \cite{Knuth}. This is the main reason to name the stack of labels of a formula as the attribute of the formula.

\newpage

\noindent Here it follows the names and restrictions of the rules of PUC-ND:
\begin{enumerate}
\item $\bs\wedge$\textbf{-elimination}: (a) $\alpha^{\Sigma}$ and $\beta^{\Omega}$ must fit into the context; (b) $\Delta$ has no existential quantifier;\\
The existential quantifier is excluded to make it possible to distribute the context over the $\wedge$ operator, what is shown in lemma \ref{resolution}.

\item $\bs\wedge$\textbf{-elimination}: (a) $\alpha^{\Sigma}$ and $\beta^{\Omega}$ must fit into the context; (b) $\Delta$ has no existential quantifier;\\
The existential quantifier is excluded to make it possible to distribute the context over the $\wedge$ operator, what is shown in lemma \ref{resolution}.

\item $\bs\wedge$\textbf{-introduction}: (a) $\alpha^{\Sigma}$ and $\beta^{\Omega}$ must fit into the context; (b) $\Delta$ has no existential quantifier;\\
The existential quantifier is excluded because the existence of some world (or neighbourhood) in which some wff $A$ holds and the existence of some world in which $B$ holds do not implies that there is some world in which $A$ and $B$ holds.

\item $\bs\vee$\textbf{-introduction}: (a) $\alpha^{\Sigma}$ and $\beta^{\Omega}$ must fit into the context; (b) $\Delta$ has no universal quantifier;\\
The universal quantifier is excluded to make it possible to distribute the context over the $\vee$ operator, what is shown in lemma \ref{resolution}.

\item $\bs\vee$\textbf{-elimination}: (a) $\alpha^{\Sigma}$ and $\beta^{\Omega}$ must fit into the context $\Delta$; (b) $\Delta$ has no universal quantifier;
The universal quantifier is excluded because the fact that for all worlds (or neighbourhoods) $A \vee B$ holds does not implies that for all worlds $A$ holds or for all worlds $B$ holds.

\item $\bs\vee$\textbf{-introduction}: (a) $\alpha^{\Sigma}$ and $\beta^{\Omega}$ must fit into the context; (b) $\Delta$ has no universal quantifier;\\
The universal quantifier is excluded to make it possible to distribute the context over the $\vee$ operator, what is shown in lemma \ref{resolution}.

\item $\bs\bot$\textbf{-classical}: (a) $\alpha^{\Sigma}$ and $\bot$ must fit into the context;

\item $\bs\bot$\textbf{-intuitionistic}: (a) $\alpha^{\Sigma}$ and $\bot$ must fit into the context;

\item \textbf{absurd expansion}: (a) $\Delta$ must have no occurrence of $\circledast$; (b) $\bot$ must fit into the context; (c) $\Delta$ must be non empty.\\
The symbol $\bot$ is used to denote a formula that may only be $\bot_{n}$ or $\bot_{w}$. In the occurrence of $\circledast$, we admit the possibility of an empty system of neighbourhoods. In that context, the absurd does not mean that we actually reach an absurd in our world. $\Delta$ must be non empty to avoid unnecessary detours, like the conclusion of $\bot_{n}$ from $\bot_{n}$ in the empty context;

\item \textbf{hypothesis-injection}: (a) $\alpha^{\Sigma}$ must fit into the context.\\
This rule permits an scope change before any formula change. It also avoids combinatorial definitions of rules with hypothesis and formulas inside a given context;

\item $\bs\ra$\textbf{-introduction}: (a) $\alpha^{\Sigma}$ and $\beta^{\Omega}$ must fit into the context;

\item $\bs\ra$\textbf{-elimination (modus ponens)}: (a) $\alpha^{\Sigma}$ and $\beta^{\Omega}$ must fit into the context; (b) $\Delta$ has no existential quantifier; (c) the premises may be in reverse order;\\
The existential quantifier is excluded because the existence of some world (or neighbourhood) in which some wff $A$ holds and the existence of some world in which $A \ra B$ holds do not implies that there is some world in which $B$ holds.

\item \textbf{context-introduction}: (a) $\alpha^{\Sigma,\phi}$ and $\alpha^{\Sigma}$ must fit into their contexts;

\item \textbf{context-elimination}: (a) $\alpha^{\Sigma,\phi}$ and $\alpha^{\Sigma}$ must fit into their contexts;

\item \textbf{world universal introduction}: (a) $\alpha^{\Sigma}$ must fit into the context; (b) $u$ must not occur in any hypothesis on which $\alpha^{\Sigma}$ depends; (c) $u$ must not occur in the context of any hypothesis on which $\alpha^{\Sigma}$ depends;

\item \textbf{world universal elimination}: (a) $\alpha^{\Sigma}$ must fit into the context; (b) $u$ must not occur in $\alpha^{\Sigma}$ or $\Delta$;

\item \textbf{world existential introduction}: (a) $\alpha^{\Sigma}$ must fit into the context;

\item \textbf{world existential elimination}: (a) the formula $\alpha^{\Sigma}$ must fit into the context; (b) $u$ must not occur in $\alpha^{\Sigma}$, $\Delta$, $\Theta$ or any open hypothesis on which $\beta^{\Omega}$ depends; (c) $u$ must not occur in the context of any open hypothesis on which $\beta^{\Omega}$ depends; (d) the premises may be in reverse order;

\item \textbf{neighbourhood existential introduction}: (a) $\alpha^{\Sigma}$ must fit into the context; (b) the premises may be in reverse order;

\item \textbf{neighbourhood existential elimination}: (a) the formula $\alpha^{\Sigma}$ must fit into the context; (b) $N$ must not occur in $\alpha^{\Sigma}$, $\Delta$, $\Theta$ or any open hypothesis on which $\beta^{\Omega}$ depends; (c) $N$ must not occur in the context of any open hypothesis on which $\beta^{\Omega}$ depends; (d) the premises may be in reverse order;

\item \textbf{neighbourhood universal introduction}: (a) the formula $\alpha^{\Sigma}$ must fit into the contexts; (b) $N$ must not occur in any open hypothesis on which $\alpha^{\Sigma}$ depends; (c) $N$ must not occur in the context of any open hypothesis on which $\alpha^{\Sigma}$ depends;

\item \textbf{neighbourhood universal wild-card}: (a) the formulas $\alpha^{\Sigma}$ and $\beta^{\Omega}$ must fit into their contexts; (b) the premises may be in reverse order;\\
This rule is necessary, because a system of neighbourhood may be empty and every variable must denote some neighbourhood because of the variable assignment function $\sigma$. The wild-card rule may be seen as a permition to use some available variable as an instantiation, by making explicit the choice of the variable.

\item \textbf{world existential propagation}: (a) $\alpha^{\Sigma,\bullet}$ and $\shneg N$ fit into their contexts; (b) the premises may be in reverse order;

\item \textbf{world universal propagation}: (a) $\alpha^{\Sigma,\ast}$ and $\shpos N$ fit into their contexts; (b) the premises may be in reverse order;

\item \textbf{transitive neighbourhood inclusion}: (a) $\shneg M$ and $\shneg P$ fit into their contexts; (b) the premises may be in reverse order;

\item \textbf{transitive neighbourhood inclusion}: (a) $\shpos M$ and $\shpos P$ fit into their contexts; (b) the premises may be in reverse order;

\item \textbf{neighbourhood total order}: (a) $\shneg M$, $\shneg N$ and $\alpha^{\Sigma}$ fit into their contexts; (b) the premises may be in reverse order;

\item \textbf{neighbourhood total order}: (a) $\shpos M$, $\shpos N$ and $\alpha^{\Sigma}$ fit into their contexts; (b) the premises may be in reverse order;

\item \textbf{neighbourhood total order}: (a) $\shneg N$, $\shpos N$ and $\alpha^{\Sigma}$ fit into their contexts. (b) the premises may be in reverse order;

\item \textbf{truth acceptance}: (a) $\Delta$ must have no occurrence of $\circledcirc$; (b) $\top$ must fit into the context.
The symbol $\top$ is used to denote a formula that may only be $\top_{n}$ or $\top_{w}$. If we accepted the occurrence of $\circledcirc$, the existence of some neighbourhood in every system of neighbourhoods would be necessary and the logic of PUC-ND should be normal according to Lewis classification \cite{Lewis}. $\Delta$ must be non empty to avoid unnecessary detours, like the conclusion of $\top_{n}$ from $\top_{n}$ in the empty context.
\end{enumerate}

We present here, as an example of the PUC-ND inference calculus, a proof of a tautology. Considering Lewis definitions, we understand that if there is some neighbourhood that has some $\beta^{\Omega}$-world but no $\alpha^{\Sigma}$-world, then, for all neighbourhoods, having some $\alpha^{\Sigma}$-world implies having some $\beta^{\Omega}$-world. The reason is the total order for the inclusion relation among neighbourhoods.\begin{center}\AxiomC{$^{4}$[$(\neg(\alpha^{\Sigma}))^{\ast} \wedge \beta^{\Omega,\bullet})^{\circledcirc}$]} 
\UnaryInfC{$(\neg(\alpha^{\Sigma}))^{\ast} \wedge \beta^{\Omega,\bullet})^{\circledcirc}$} \RightLabel{$\circledcirc$}
\UnaryInfC{$(\neg(\alpha^{\Sigma}))^{\ast} \wedge \beta^{\Omega,\bullet}$}
\AxiomC{$^{3}$[$(\neg(\alpha^{\Sigma}))^{\ast} \wedge \beta^{\Omega,\bullet}$]} \RightLabel{$N$}
\UnaryInfC{$\Pi$}
\UnaryInfC{$(\alpha^{\Sigma,\bullet} \ra \beta^{\Omega,\bullet})^{\circledast}$} \LeftLabel{3}
\BinaryInfC{$(\alpha^{\Sigma,\bullet} \ra \beta^{\Omega,\bullet})^{\circledast}$} \LeftLabel{4}
\UnaryInfC{$((\neg(\alpha^{\Sigma}))^{\ast} \wedge \beta^{\Omega,\bullet})^{\circledcirc} \ra (\alpha^{\Sigma,\bullet} \ra \beta^{\Omega,\bullet})^{\circledast}$} 
\alwaysNoLine
\UnaryInfC{$\phantom{.}$} 
\UnaryInfC{$\phantom{.}$} \DisplayProof

\begin{footnotesize}\AxiomC{$(\neg(\alpha^{\Sigma}))^{\ast} \wedge \beta^{\Omega,\bullet}$} \RightLabel{$N$}
\UnaryInfC{$(\neg(\alpha^{\Sigma}))^{\ast} \wedge \beta^{\Omega,\bullet}$} \RightLabel{$N$}
\UnaryInfC{$\beta^{\Omega,\bullet}$}
\AxiomC{$^{2}$[$\shneg N$]} \RightLabel{$M$}
\UnaryInfC{$\shneg N$} \RightLabel{$M$}
\BinaryInfC{$\beta^{\Omega,\bullet}$} \RightLabel{$M$}
\UnaryInfC{$\alpha^{\Sigma,\bullet} \ra \beta^{\Omega,\bullet}$}
\AxiomC{$(\neg(\alpha^{\Sigma}))^{\ast} \wedge \beta^{\Omega,\bullet}$} \RightLabel{$N$}
\UnaryInfC{$(\neg(\alpha^{\Sigma}))^{\ast} \wedge \beta^{\Omega,\bullet}$} \RightLabel{$N$}
\UnaryInfC{$(\neg(\alpha^{\Sigma}))^{\ast}$}
\AxiomC{$^{2}$[$\shpos N$]} \RightLabel{$M$}
\UnaryInfC{$\shpos N$} \RightLabel{$M$}
\BinaryInfC{$(\neg(\alpha^{\Sigma}))^{\ast}$} \RightLabel{$M,\ast$}
\UnaryInfC{$\neg(\alpha^{\Sigma})$} \RightLabel{$M,u$}
\UnaryInfC{$\neg(\alpha^{\Sigma})$}
\AxiomC{$^{1}$[$\alpha^{\Sigma,\bullet}$]} \RightLabel{$M$}
\UnaryInfC{$\alpha^{\Sigma,\bullet}$} \RightLabel{$M,\bullet$}
\UnaryInfC{$\alpha^{\Sigma}$} \RightLabel{$M,u$}
\UnaryInfC{$\alpha^{\Sigma}$} \RightLabel{$M,u$}
\BinaryInfC{$\bot$} \RightLabel{$M,u$}
\UnaryInfC{$\beta^{\Omega}$} \RightLabel{$M,\bullet$}
\UnaryInfC{$\beta^{\Omega}$} \RightLabel{$M$}
\UnaryInfC{$\beta^{\Omega,\bullet}$} \RightLabel{$M$} \LeftLabel{1}
\UnaryInfC{$\alpha^{\Sigma,\bullet} \ra \beta^{\Omega,\bullet}$} \LeftLabel{2} \RightLabel{$M$}
\BinaryInfC{$\alpha^{\Sigma,\bullet} \ra \beta^{\Omega,\bullet}$} \RightLabel{$\circledast$}
\UnaryInfC{$\alpha^{\Sigma,\bullet} \ra \beta^{\Omega,\bullet}$} \LeftLabel{$\bs{\Pi} \hspace*{1cm}$}
\UnaryInfC{$(\alpha^{\Sigma,\bullet} \ra \beta^{\Omega,\bullet})^{\circledast}$} \DisplayProof\end{footnotesize}\end{center}

\begin{lemma} \label{countingLemma}
If $\Delta \in \bs{C}_{n}$, then $s(\Delta)$ is odd. If $\Delta \in \bs{C}_{w}$, then $s(\Delta)$ is even.
\end{lemma}

\begin{proof}
By definition, if $\Delta$ is empty, then $\Delta \in \bs{C}_{w}$ and $s(\Delta)$ is even. According to the rules of the PUC-ND, if $\Delta$ is empty, it can only accept an additional label $\phi \in \bs{L}_{n}$, then $\{\Delta,\phi\} \in \bs{C}_{n}$ and $s(\Delta)$ is odd. We conclude that changing the context from $\bs{C}_{w}$ to $\bs{C}_{n}$ and vice-versa always involves adding one to the size of the label and the even sizes are only and always for contexts in $\bs{C}_{w}$.\end{proof}

\section{PUC Soundness and Completeness}

For the proof of soundness of PUC-Logic, we prove that the PUC-ND derivations preserves the relation of resolution, which is a relation that generalizes the satisfability relation. To do so, we need to prove some lemmas. In many cases we use the definition \ref{referentialConsequence} of the referential consequence relation.

\begin{definition} \label{resolutionDf}
Given a model $\m{M}$, a context $\Delta$ and a wff $\alpha^{\Sigma}$, the relation $\m{M} \models^{\Delta} \alpha^{\Sigma}$ of \textit{resolution} is defined iff $\alpha^{\Sigma}$ fits into the context $\Delta$ and $\m{M} \models \alpha^{\Sigma,\overline{\Delta}}$. If $\Gamma \subset \bs{F}_{n}$ or $\Gamma \subset \bs{F}_{w}$, then $\m{M} \models^{\Delta} \Gamma$ if the resolution relation holds for every formula of $\Gamma$.
\end{definition}

\begin{lemma} \label{lemmaConsequence}
Given a model $\m{M} = \ob \m{W} , \$ , \m{V} , \chi \cb$, if $\m{M} \models^{\Delta} \alpha^{\Sigma}$ and $\alpha^{\Sigma} \models_{\m{M}:s(\Delta)} \beta^{\Omega}$, then $\m{M} \models^{\Delta} \beta^{\Omega}$.
\end{lemma}

\begin{proof}
If $\Delta$ is empty ($s(\Delta)=0$), the resolution gives us $\m{M} \models \alpha^{\Sigma}$. From $\alpha^{\Sigma} \models_{\m{M}:0} \beta^{\Omega}$ we know that $\m{M} \models \beta^{\Omega}$ if $\m{M} \models \alpha^{\Sigma}$ and, by the definition of resolution, $\m{M} \models^{\Delta} \beta^{\Omega}$;\\

\noindent If $\Delta = \{\circledast\}$ ($s(\Delta)=1$), then, by definition, $\m{M} \models^{\{\circledast\}} \alpha^{\Sigma}$ means $\m{M} \models \alpha^{\Sigma,\circledast}$ and for every template $\m{T}$, such that $\m{M} \multimap \m{T}$, $\m{T} \models \alpha^{\Sigma}$. 
$$
\begin{xy}
\xymatrix{
&\ob \m{W} , \$ , \m{V} , \chi \cb \models \alpha^{\Sigma,\circledast} \ar@{-o}[dl] \ar@{-o}[d] \ar@{-o}[dr] &\\
\ob \m{W} , \$ , \m{V} , \chi , N \cb \models \alpha^{\Sigma} & \ldots & \ob \m{W} , \$ , \m{V} , \chi , S \cb \models \alpha^{\Sigma}
}
\end{xy}
$$
$N, \ldots , S$ represent all neighbourhoods of $\$(\chi)$. From $s(\{\circledast\})=1$, we know that $\alpha^{\Sigma} \models_{\m{M}:1} \beta^{\Omega}$ and, by definition, we can change $\alpha^{\Sigma}$ by $\beta^{\Omega}$ in all endpoints of the directed graph and conclude $\ob \m{W} , \$ , \m{V} , \chi \cb \models \beta^{\Omega,\circledast}$ and $\m{M} \models^{\Delta} \beta^{\Omega}$;\\

\noindent If $\Delta = \{\circledcirc\}$ ($s(\Delta)=1$), then $\m{M} \models \alpha^{\Sigma,\circledcirc}$.
$$
\begin{xy}
\xymatrix{
&\ob \m{W} , \$ , \m{V} , \chi \cb \models \alpha^{\Sigma,\circledcirc} \ar@{-o}[dl] \ar@{-o}[d] \ar@{-o}[dr] &\\
\ob \m{W} , \$ , \m{V} , \chi , N \cb \models \alpha^{\Sigma} & \ldots & \ob \m{W} , \$ , \m{V} , \chi , S \cb \models \alpha^{\Sigma}
}
\end{xy}
$$
$N, \ldots , S$ represent all neighbourhoods of $\$(\chi)$ such that $\alpha^{\Sigma}$ holds. We know that there is at least one of such neighbourhoods. From $\alpha^{\Sigma} \models_{\m{M}:1} \beta^{\Omega}$, we can change $\alpha^{\Sigma}$ by $\beta^{\Omega}$ in all endpoints and conclude $\m{M} \models \beta^{\Omega,\circledcirc}$ because we know that there is at least one of such downward paths. By definition, $\m{M} \models^{\Delta} \beta^{\Omega}$;\\

\noindent If $\Delta = \{N\}$ ($s(\Delta)=1$), then $\m{M} \models \alpha^{\Sigma,N}$.
$$
\begin{xy}
\xymatrix{
\ob \m{W} , \$ , \m{V} , \chi \cb \models \alpha^{\Sigma,N} \ar@{-o}[d]\\
\ob \m{W} , \$ , \m{V} , \chi , \sigma(N) \cb \models \alpha^{\Sigma}
}
\end{xy}
$$
From $\alpha^{\Sigma} \models_{\m{M}:1} \beta^{\Omega}$, we change $\alpha^{\Sigma}$ by $\beta^{\Omega}$ in the endpoint and conclude $\m{M} \models \beta^{\Omega,N}$. By definition, $\m{M} \models^{\Delta} \beta^{\Omega}$;\\

\noindent If $\Delta = \{\circledast,\ast\}$ ($s(\Delta)=2$), then $\m{M} \models \alpha^{\Sigma,\ast,\circledast}$.
$$
\begin{scriptsize}
\begin{xy}
\xymatrix{
\ob \m{W} , \$ , \m{V} , \chi \cb \models \alpha^{\Sigma,\ast,\circledast} \ar@{-o}[d] \ar@{-o}[drrrr] &&&\\
\ob \m{W} , \$ , \m{V} , \chi , N \cb \models \alpha^{\Sigma,\ast} \ar@{-o}[d] \ar@{-o}[drr] && \ldots && \ob \m{W} , \$ , \m{V} , \chi , S \cb \models \alpha^{\Sigma,\ast} \ar@{-o}[d]\\
\ob \m{W} , \$ , \m{V} , \lambda_{1} \cb \models \alpha^{\Sigma} & \ldots & \ob \m{W} , \$ , \m{V} , \lambda_{t} \cb \models \alpha^{\Sigma} && \ldots
}
\end{xy}
\end{scriptsize}
$$
$N, \ldots , S$ represent all neighbourhoods of $\$(\chi)$. $\lambda_{1}, \ldots , \lambda_{t}$ represent all worlds of $N$. From $\alpha^{\Sigma} \models_{\m{M}:2} \beta^{\Omega}$, we can change $\alpha^{\Sigma}$ by $\beta^{\Omega}$ in all endpoints and conclude $\m{M} \models \beta^{\Omega,\ast,\circledast}$. By definition, $\m{M} \models^{\Delta} \beta^{\Omega}$;\\

\noindent If $\Delta = \{\circledast,\bullet\}$ ($s(\Delta)=2$), then $\m{M} \models \alpha^{\Sigma,\bullet,\circledast}$.
$$
\begin{scriptsize}
\begin{xy}
\xymatrix{
\ob \m{W} , \$ , \m{V} , \chi \cb \models \alpha^{\Sigma,\bullet,\circledast} \ar@{-o}[d] \ar@{-o}[drrrr] &&&\\
\ob \m{W} , \$ , \m{V} , \chi , N \cb \models \alpha^{\Sigma,\bullet} \ar@{-o}[d] \ar@{-o}[drr] && \ldots && \ob \m{W} , \$ , \m{V} , \chi , S \cb \models \alpha^{\Sigma,\bullet} \ar@{-o}[d]\\
\ob \m{W} , \$ , \m{V} , \lambda_{1} \cb \models \alpha^{\Sigma} & \ldots & \ob \m{W} , \$ , \m{V} , \lambda_{t} \cb \models \alpha^{\Sigma} && \ldots
}
\end{xy}
\end{scriptsize}
$$
$N, \ldots , S$ represent all neighbourhoods of $\$(\chi)$. $\lambda_{1}, \ldots , \lambda_{t}$ represent all worlds of $N$ in which $\alpha^{\Sigma}$ holds. We know that there is at least one of these worlds. From $\alpha^{\Sigma} \models_{\m{M}:2} \beta^{\Omega}$, we can change $\alpha^{\Sigma}$ by $\beta^{\Omega}$ in all endpoints and conclude $\m{M} \models \beta^{\Omega,\bullet,\circledast}$ and $\m{M} \models^{\Delta} \beta^{\Omega}$;\\

\noindent If $\Delta = \{\circledast,u\}$ ($s(\Delta)=2$), then $\m{M} \models \alpha^{\Sigma,u,\circledast}$.
$$
\begin{xy}
\xymatrix{
\ob \m{W} , \$ , \m{V} , \chi \cb \models \alpha^{\Sigma,u,\circledast} \ar@{-o}[d] \ar@{-o}[drrr] &&&\\
\ob \m{W} , \$ , \m{V} , \chi , N_{1} \cb \models \alpha^{\Sigma,u} \ar@{-o}[d] & \ldots && \ob \m{W} , \$ , \m{V} , \chi , N_{s} \cb \models \alpha^{\Sigma,u} \ar@{-o}[dlll]\\
\ob \m{W} , \$ , \m{V} , \sigma(u) \cb \models \alpha^{\Sigma} &&&& 
}
\end{xy}
$$
$N, \ldots , S$ represent all neighbourhoods of $\$(\chi)$. From $\alpha^{\Sigma} \models_{\m{M}:2} \beta^{\Omega}$, we can change $\alpha^{\Sigma}$ by $\beta^{\Omega}$ in the endpoint and conclude $\m{M} \models \beta^{\Omega,u,\circledast}$. So, by definition, $\m{M} \models^{\Delta} \beta^{\Omega}$;\\
\noindent Any combination of labels follows, by analogy, the same arguments for each label presented above.\end{proof}

\begin{lemma} \label{lemmaDirectConsequence}
Given a model $\m{M} = \ob \m{W} , \$ , \m{V} , \chi \cb$, if $\m{M} \models^{\Delta} \alpha^{\Sigma}$ and $\alpha^{\Sigma} \models \beta^{\Omega}$, then $\m{M} \models^{\Delta} \beta^{\Omega}$.
\end{lemma}

\begin{proof}
We follow the argument of lemma \ref{lemmaConsequence}, by changing $\alpha^{\Sigma}$ by $\beta^{\Omega}$ in all endpoints, what is possible by the definition of logical consequence.
\end{proof}

\begin{lemma} \label{lemmaVee}
Given $\Delta$ without universal quantifiers, if $\alpha^{\Sigma,\overline{\Delta}} \vee \beta^{\Omega,\overline{\Delta}}$ is wff, then $\alpha^{\Sigma,\overline{\Delta}} \vee \beta^{\Omega,\overline{\Delta}} \equiv (\alpha^{\Sigma} \vee \beta^{\Omega})^{\overline{\Delta}}$.
\end{lemma}

\begin{proof} We proceed by induction on the size of $\Delta$:\\
\noindent If $\Delta$ is empty, then equivalence is true;\\
\noindent (base) If $\Delta$ contains only one label, it must be a neighbourhood label:
\begin{itemize}
\item[-] $\alpha^{\Sigma,\circledcirc} \vee \beta^{\Omega,\circledcirc}$ may be read as $\exists N \in \$(\chi) : \ob \m{W} , \$ , \m{V} , \chi , N \cb \models \alpha^{\Sigma}$ or $\exists M \in \$(\chi) : \ob \m{W} , \$ , \m{V} , \chi , M \cb \models \beta^{\Omega}$. But $\exists N \in \$(\chi) : \ob \m{W} , \$ , \m{V} , \chi , N \cb \models \alpha^{\Sigma}$ implies, by definition, $\exists N \in \$(\chi) : \ob \m{W} , \$ , \m{V} , \chi , N \cb \models \alpha^{\Sigma} \vee \beta^{\Omega}$. Then we have $\exists N \in \$(\chi) : \ob \m{W} , \$ , \m{V} , \chi , N \cb \models \alpha^{\Sigma} \vee \beta^{\Omega}$ or $\exists M \in \$(\chi) : \ob \m{W} , \$ , \m{V} , \chi , M \cb \models \alpha^{\Sigma} \vee \beta^{\Omega}$. Since the neighbourhood variables are bound, we have $\exists N \in \$(\chi) : \ob \m{W} , \$ , \m{V} , \chi , N \cb \models \alpha^{\Sigma} \vee \beta^{\Omega}$, which is represented whit labels as $(\alpha^{\Sigma} \vee \beta^{\Omega})^{\circledcirc}$. Then $\alpha^{\Sigma,\circledcirc} \vee \beta^{\Omega,\circledcirc}$ implies $(\alpha^{\Sigma} \vee \beta^{\Omega})^{\circledcirc}$. On the other hand, $(\alpha^{\Sigma} \vee \beta^{\Omega})^{\circledcirc}$ may be read as $\exists N \in \$(\chi) : \ob \m{W} , \$ , \m{V} , \chi , N \cb \models \alpha^{\Sigma} \vee \beta^{\Omega}$, which means, by definition, $\exists N \in \$(\chi) : \ob \m{W} , \$ , \m{V} , \chi , N \cb \models \alpha^{\Sigma} \mbox{ or } \ob \m{W} , \$ , \m{V} , \chi , N \cb \models \beta^{\Omega}$. In the first case, $\exists N \in \$(\chi) : \ob \m{W} , \$ , \m{V} , \chi , N \cb \models \alpha^{\Sigma}$, which may be read as $\alpha^{\Sigma,\circledcirc}$. In the second case, $\exists N \in \$(\chi) : \ob \m{W} , \$ , \m{V} , \chi , N \cb \models \beta^{\Omega}$, which may be read as $\beta^{\Omega,\circledcirc}$. Since we have one or the other case, we have $\alpha^{\Sigma,\circledcirc} \vee \beta^{\Omega,\circledcirc}$. So, $(\alpha^{\Sigma} \vee \beta^{\Omega})^{\circledcirc} \equiv \alpha^{\Sigma,\circledcirc} \vee \beta^{\Omega,\circledcirc}$;
\item[-] $\alpha^{\Sigma,N} \vee \beta^{\Omega,N}$ may be read as $\sigma(N) \in \$(\chi)$ and $\ob \m{W} , \$ , \m{V} , \chi , \sigma(N) \cb \models \alpha^{\Sigma}$ or $\sigma(N) \in \$(\chi)$ and $\ob \m{W} , \$ , \m{V} , \chi , \sigma(N) \cb \models \beta^{\Omega}$. Then we have $\sigma(N) \in \$(\chi)$ and ($\ob \m{W} , \$ , \m{V} , \chi , \sigma(N) \cb \models \alpha^{\Sigma}$ or $\ob \m{W} , \$ , \m{V} , \chi , \sigma(N) \cb \models \beta^{\Omega}$), which is, by definition, $\ob \m{W} , \$ , \m{V} , \chi , \sigma(N) \cb \models \alpha^{\Sigma} \vee \beta^{\Omega}$.  Then $\alpha^{\Sigma,N} \vee \beta^{\Omega,N}$ implies $(\alpha^{\Sigma} \vee \beta^{\Omega})^{N}$. On the other hand, $(\alpha^{\Sigma} \vee \beta^{\Omega})^{N}$ may be read as $\sigma(N) \in \$(\chi)$ and $\ob \m{W} , \$ , \m{V} , \chi , \sigma(N) \cb \models \alpha^{\Sigma} \vee \beta^{\Omega}$, which means, by definition, $\sigma(N) \in \$(\chi)$ and $(\; \ob \m{W} , \$ , \m{V} , \chi , \sigma(N) \cb \models \alpha^{\Sigma} \mbox{ or } \ob \m{W} , \$ , \m{V} , \chi , \sigma(N) \cb \models \beta^{\Omega} \;)$. So, we have ($\sigma(N) \in \$(\chi)$ and $\ob \m{W} , \$ , \m{V} , \chi , \sigma(N) \cb \models \alpha^{\Sigma}$) or ($\sigma(N) \in \$(\chi)$ and $\ob \m{W} , \$ , \m{V} , \chi , \sigma(N) \cb \models \beta^{\Omega}$), which may be read as $\alpha^{\Sigma,N} \vee \beta^{\Omega,N}$. So, $(\alpha^{\Sigma} \vee \beta^{\Omega})^{N} \equiv \alpha^{\Sigma,N} \vee \beta^{\Omega,N}$;
\end{itemize}

\noindent (base) If $\Delta$ contains two labels, it may be $\{\circledcirc,\bullet\}$, $\{N,\bullet\}$, $\{\circledcirc,u\}$ or $\{N,u\}$. But we just need to look at the distributivity for the $\bullet$ label and for world variables, because we have already seen the distributivity of the $\vee$ connective for the label $\circledcirc$ and for any neighbourhood variable.
\begin{itemize}
\item[-] $\alpha^{\Sigma,\bullet,\circledcirc} \vee \beta^{\Omega,\bullet,\circledcirc}$ may be read as $\exists N \in \$(\chi) : \ob \m{W} , \$ , \m{V} , \chi , N \cb \models \alpha^{\Sigma,\bullet}$ or $\exists M \in \$(\chi) : \ob \m{W} , \$ , \m{V} , \chi , M \cb \models \beta^{\Omega,\bullet}$. But $\ob \m{W} , \$ , \m{V} , \chi , N \cb \models \alpha^{\Sigma,\bullet}$ implies, by definition, $\exists w \in N : \ob \m{W} , \$ , \m{V} , w \cb \models \alpha^{\Sigma}$, which implies $\exists w \in N : \ob \m{W} , \$ , \m{V} , w \cb \models \alpha^{\Sigma} \vee \beta^{\Omega}$. So, we have $\exists N \in \$(\chi) : \exists w \in N : \ob \m{W} , \$ , \m{V} , w \cb \models \alpha^{\Sigma} \vee \beta^{\Omega}$ or $\exists M \in \$(\chi) : \exists z \in N : \ob \m{W} , \$ , \m{V} , z \cb \models \alpha^{\Sigma} \vee \beta^{\Omega}$. Since every variable is bound, we have $\exists N \in \$(\chi) : \exists w \in N : \ob \m{W} , \$ , \m{V} , w \cb \models \alpha^{\Sigma} \vee \beta^{\Omega}$, which is, by definition, equivalent to $\exists N \in \$(\chi) : \ob \m{W} , \$ , \m{V} , \chi , N \cb \models (\alpha^{\Sigma} \vee \beta^{\Omega})^{\bullet}$, which is equivalent, by definition, to $(\alpha^{\Sigma} \vee \beta^{\Omega})^{\bullet,\circledcirc}$. On the other hand, $(\alpha^{\Sigma} \vee \beta^{\Omega})^{\bullet,\circledcirc}$ may be read as $\exists N \in \$(\chi) : \exists w \in N : \ob \m{W} , \$ , \m{V} , w \cb \models \alpha^{\Sigma} \vee \beta^{\Omega}$, which is, by definition, $\exists N \in \$(\chi) : \exists w \in N : \ob \m{W} , \$ , \m{V} , w \cb \models \alpha^{\Sigma} \mbox{ or } \ob \m{W} , \$ , \m{V} , w \cb \models \beta^{\Omega}$, which implies $\exists N \in \$(\chi) : \exists w \in N : \ob \m{W} , \$ , \m{V} , w \cb \models \alpha^{\Sigma} \mbox{ or } \exists z \in N : \ob \m{W} , \$ , \m{V} , z \cb \models \beta^{\Omega}$, which implies $\exists N \in \$(\chi) : \exists w \in N : \ob \m{W} , \$ , \m{V} , w \cb \models \alpha^{\Sigma} \mbox{ or } \exists M \in \$(\chi) : \exists z \in M : \ob \m{W} , \$ , \m{V} , z \cb \models \beta^{\Omega}$, which may be represented with labels as $\alpha^{\Sigma,\bullet,\circledcirc} \vee \beta^{\Omega,\bullet,\circledcirc}$. So, $\alpha^{\Sigma,\bullet,\circledcirc} \vee \beta^{\Omega,\bullet,\circledcirc} \equiv (\alpha^{\Sigma} \vee \beta^{\Omega})^{\bullet,\circledcirc}$;
\item[-] The proofs of $\alpha^{\Sigma,\bullet,N} \vee \beta^{\Omega,\bullet,N} \equiv (\alpha^{\Sigma} \vee \beta^{\Omega})^{\bullet,N}$, $\alpha^{\Sigma,u,\circledcirc} \vee \beta^{\Omega,u,\circledcirc} \equiv (\alpha^{\Sigma} \vee \beta^{\Omega})^{u,\circledcirc}$ and $\alpha^{\Sigma,u,N} \vee \beta^{\Omega,u,N} \equiv (\alpha^{\Sigma} \vee \beta^{\Omega})^{u,N}$ are analogous.
\end{itemize}

\noindent (induction) If $\alpha^{\Sigma} \vee \beta^{\Omega} \in \bs{F}_{w}$, $\Delta = \{\Delta',\phi\}$ and $s(\Delta) = n + 1$, then $\oc\Delta \in \bs{L}_{n}$ and $\alpha^{\Sigma,\overline{\Delta}} \vee \beta^{\Omega,\overline{\Delta}}$ may be written as $\alpha^{\Sigma,\phi,\overline{\Delta'}} \vee \beta^{\Omega,\phi,\overline{\Delta'}}$, where $s(\Delta') = n$. Then, by the induction hypothesis, $\alpha^{\Sigma,\phi,\overline{\Delta'}} \vee \beta^{\Omega,\phi,\overline{\Delta'}} = (\alpha^{\Sigma,\phi} \vee \beta^{\Omega,\phi})^{\overline{\Delta'}}$. From the base assertions, $(\alpha^{\Sigma,\phi} \vee \beta^{\Omega,\phi})^{\overline{\Delta'}} = ((\alpha^{\Sigma} \vee \beta^{\Omega})^{\phi})^{\overline{\Delta'}} = (\alpha^{\Sigma} \vee \beta^{\Omega})^{\phi,\overline{\Delta'}} = (\alpha^{\Sigma} \vee \beta^{\Omega})^{\overline{\Delta}}$;\\
\noindent (induction) If $\alpha^{\Sigma} \vee \beta^{\Omega} \in \bs{F}_{n}$ and $s(\Delta) = n + 2$, then $\oc\Delta \in \bs{L}_{w}$ and $\alpha^{\Sigma,\overline{\Delta}} \vee \beta^{\Omega,\overline{\Delta}}$ may be written as $\alpha^{\Sigma,\phi,\Theta,\overline{\Delta'}} \vee \beta^{\Omega,\phi,\Theta,\overline{\Delta'}}$, where $s(\Delta') = n$. Then, by induction hypothesis, $\alpha^{\Sigma,\phi,\Theta,\overline{\Delta'}} \vee \beta^{\Omega,\phi,\Theta,\overline{\Delta'}} = (\alpha^{\Sigma,\phi,\Theta} \vee \beta^{\Omega,\phi,\Theta})^{\overline{\Delta'}}$. By base, $(\alpha^{\Sigma,\phi,\Theta} \vee \beta^{\Omega,\phi,\Theta})^{\overline{\Delta'}} = ((\alpha^{\Sigma} \vee \beta^{\Omega})^{\phi,\Theta})^{\overline{\Delta'}} = (\alpha^{\Sigma} \vee \beta^{\Omega})^{\phi,\Theta,\overline{\Delta'}} = (\alpha^{\Sigma} \vee \beta^{\Omega})^{\overline{\Delta}}$.\end{proof}

\begin{lemma} \label{lemmaWedge}
Given $\Delta$ without existential quantifiers, if $\alpha^{\Sigma,\overline{\Delta}} \wedge \beta^{\Omega,\overline{\Delta}}$ is wff, then $\alpha^{\Sigma,\overline{\Delta}} \wedge \beta^{\Omega,\overline{\Delta}} \equiv (\alpha^{\Sigma} \wedge \beta^{\Omega})^{\overline{\Delta}}$.
\end{lemma}

\begin{proof} We proceed by induction on the size of $\Delta$:\\
\noindent If $\Delta$ is empty, then equivalence is true;\\
\noindent (base) If $\Delta$ contains only one label, it must be a neighbourhood label:
\begin{itemize}
\item[-] $\alpha^{\Sigma,\circledast} \wedge \beta^{\Omega,\circledast}$ may be read as $\forall N \in \$(\chi) : \ob \m{W} , \$ , \m{V} , \chi , N \cb \models \alpha^{\Sigma}$ and $\forall M \in \$(\chi) : \ob \m{W} , \$ , \m{V} , \chi , M \cb \models \beta^{\Omega}$. But then, we may conclude that, for every neighbourhood $L \in \$(\chi)$, $\ob \m{W} , \$ , \m{V} , \chi , L \cb \models \alpha^{\Sigma}$ and $\ob \m{W} , \$ , \m{V} , \chi , L \cb \models \beta^{\Omega}$, which can be represented with labels, since $L$ is arbitrary, as $(\alpha^{\Sigma} \wedge \beta^{\Omega})^{\circledast}$. On the other hand, $(\alpha^{\Sigma} \wedge \beta^{\Omega})^{\circledast}$ can be read as $\forall N \in \$(\chi) : \ob \m{W} , \$ , \m{V} , \chi , N \cb \models \alpha^{\Sigma} \wedge \beta^{\Omega}$, which is equivalent, by definition, to $\forall N \in \$(\chi) : \ob \m{W} , \$ , \m{V} , \chi , N \cb \models \alpha^{\Sigma} \mbox{ and } \ob \m{W} , \$ , \m{V} , \chi , N \cb \models \beta^{\Omega}$. So we have $\forall N \in \$(\chi) : \ob \m{W} , \$ , \m{V} , \chi , N \cb \models \alpha^{\Sigma}$ and $\forall N \in \$(\chi) : \ob \m{W} , \$ , \m{V} , \chi , N \cb \models \beta^{\Omega}$, that is equivalent to $\alpha^{\Sigma,\circledast} \wedge \beta^{\Omega,\circledast}$;

\item[-] $\alpha^{\Sigma,N} \wedge \beta^{\Omega,N}$ may be read as ($\sigma(N) \in \$(\chi)$ and $\ob \m{W} , \$ , \m{V} , \chi , \sigma(N) \cb \models \alpha^{\Sigma}$) and ($\sigma(N) \in \$(\chi)$ and $\ob \m{W} , \$ , \m{V} , \chi , \sigma(N) \cb \models \beta^{\Omega}$). But then, we may conclude, by definition, that $\ob \m{W} , \$ , \m{V} , \chi , \sigma(N) \cb \models \alpha^{\Sigma}$ and $\ob \m{W} , \$ , \m{V} , \chi , \sigma(N) \cb \models \beta^{\Omega}$, which can be represented with labels as $(\alpha^{\Sigma} \wedge \beta^{\Omega})^{N}$. On the other hand, $(\alpha^{\Sigma} \wedge \beta^{\Omega})^{N}$ can be read as $\sigma(N) \in \$(\chi)$ and $\ob \m{W} , \$ , \m{V} , \chi , \sigma(N) \cb \models \alpha^{\Sigma} \wedge \beta^{\Omega}$, which is equivalent, by definition, to $\sigma(N) \in \$(\chi)$ and $\ob \m{W} , \$ , \m{V} , \chi , \sigma(N) \cb \models \alpha^{\Sigma} \mbox{ and } \ob \m{W} , \$ , \m{V} , \chi , \sigma(N) \cb \models \beta^{\Omega}$. So we have ($\sigma(N) \in \$(\chi)$ and $\ob \m{W} , \$ , \m{V} , \chi , \sigma(N) \cb \models \alpha^{\Sigma}$ and ($\sigma(N) \in \$(\chi)$ and $\ob \m{W} , \$ , \m{V} , \chi , \sigma(N) \cb \models \beta^{\Omega}$), that is equivalent to $\alpha^{\Sigma,N} \wedge \beta^{\Omega,N}$;
\end{itemize}

\noindent (base) If $\Delta$ contains two labels, it may be $\{\circledast,\ast\}$, $\{N,\ast\}$, $\{\circledast,u\}$ or $\{N,u\}$. But we just need to look at the distributivity for the $\ast$ label and for world variables, because we have already seen the distributivity of the $\wedge$ connective for the label $\circledast$ and for any neighbourhood variable.
\begin{itemize}
\item[-] $\alpha^{\Sigma,\ast,\circledast} \wedge \beta^{\Omega,\ast,\circledast}$ may be read as $\forall N \in \$(\chi) : \ob \m{W} , \$ , \m{V} , \chi , N \cb \models \alpha^{\Sigma,\ast}$ and $\forall M \in \$(\chi) : \ob \m{W} , \$ , \m{V} , \chi , M \cb \models \beta^{\Omega,\ast}$. Then we have, by definition, $\forall w \in N : \ob \m{W} , \$ , \m{V} , w \cb \models \alpha^{\Sigma}$ and $\forall z \in M : \ob \m{W} , \$ , \m{V} , z \cb \models \beta^{\Omega}$. So, for every world $x$ of every neighbourhood $L$, $\ob \m{W} , \$ , \m{V} , x \cb \models \alpha^{\Sigma}$ and $\ob \m{W} , \$ , \m{V} , x \cb \models \beta^{\Omega}$. Then we may conclude, by definition, that $\ob \m{W} , \$ , \m{V} , x \cb \models \alpha^{\Sigma} \wedge \beta^{\Omega}$ and represent it with labels as $(\alpha^{\Sigma} \wedge \beta^{\Omega})^{\ast,\circledast}$ because $x$ and $L$ are arbitrary. On the other hand, $(\alpha^{\Sigma} \wedge \beta^{\Omega})^{\ast,\circledast}$ may be read as $\forall N \in \$(\chi) : \forall w \in N : \alpha^{\Sigma} \wedge \beta^{\Omega}$, which implies, by definition, $\forall N \in \$(\chi) : \forall w \in N : \alpha^{\Sigma}$ and also $\forall N \in \$(\chi) : \forall w \in N : \beta^{\Omega}$. So, we have $\ob \m{W} , \$ , \m{V} , \chi \cb \models \alpha^{\Sigma,\ast,\circledast}$ and $\ob \m{W} , \$ , \m{V} , \chi \cb \models \beta^{\Omega,\ast,\circledast}$. So, we may conclude, by definition, that $\alpha^{\Sigma,\ast,\circledast} \wedge \beta^{\Omega,\ast,\circledast}$;
\item[-] The proofs of $\alpha^{\Sigma,\ast,N} \wedge \beta^{\Omega,\ast,N} \equiv (\alpha^{\Sigma} \wedge \beta^{\Omega})^{\ast,N}$, $\alpha^{\Sigma,u,\circledast} \wedge \beta^{\Omega,u,\circledast} \equiv (\alpha^{\Sigma} \wedge \beta^{\Omega})^{u,\circledast}$ and$\alpha^{\Sigma,u,N} \wedge \beta^{\Omega,u,N} \equiv (\alpha^{\Sigma} \wedge \beta^{\Omega})^{u,N}$ are analogous.
\end{itemize}

\noindent (induction) If $\alpha^{\Sigma} \wedge \beta^{\Omega} \in \bs{F}_{w}$ and $s(\Delta) = n + 1$, then $\oc\Delta \in \bs{L}_{n}$ and $\alpha^{\Sigma,\overline{\Delta}} \wedge \beta^{\Omega,\overline{\Delta}}$ may be written as $\alpha^{\Sigma,\phi,\overline{\Delta'}} \wedge \beta^{\Omega,\phi,\overline{\Delta'}}$, where $s(\Delta') = n$. Then, by the induction hypothesis, $\alpha^{\Sigma,\phi,\overline{\Delta'}} \wedge \beta^{\Omega,\phi,\overline{\Delta'}} = (\alpha^{\Sigma,\phi} \wedge \beta^{\Omega,\phi})^{\overline{\Delta'}}$. From the base assertions, $(\alpha^{\Sigma,\phi} \wedge \beta^{\Omega,\phi})^{\overline{\Delta'}} = ((\alpha^{\Sigma} \wedge \beta^{\Omega})^{\phi})^{\overline{\Delta'}} = (\alpha^{\Sigma} \wedge \beta^{\Omega})^{\phi,\overline{\Delta'}} = (\alpha^{\Sigma} \wedge \beta^{\Omega})^{\overline{\Delta}}$;\\
\noindent (induction) If $\alpha^{\Sigma} \wedge \beta^{\Omega} \in \bs{F}_{n}$ and $s(\Delta) = n + 2$, then $\oc\Delta \in \bs{L}_{w}$ and $\alpha^{\Sigma,\overline{\Delta}} \wedge \beta^{\Omega,\overline{\Delta}}$ may be written as $\alpha^{\Sigma,\phi,\Theta,\overline{\Delta'}} \wedge \beta^{\Omega,\phi,\Theta,\overline{\Delta'}}$, where $s(\Delta') = n$. Then, by induction hypothesis, $\alpha^{\Sigma,\phi,\Theta,\overline{\Delta'}} \wedge \beta^{\Omega,\phi,\Theta,\overline{\Delta'}} = (\alpha^{\Sigma,\phi,\Theta} \wedge \beta^{\Omega,\phi,\Theta})^{\overline{\Delta'}}$. By base, $(\alpha^{\Sigma,\phi,\Theta} \wedge \beta^{\Omega,\phi,\Theta})^{\overline{\Delta'}} = ((\alpha^{\Sigma} \wedge \beta^{\Omega})^{\phi,\Theta})^{\overline{\Delta'}} = (\alpha^{\Sigma} \wedge \beta^{\Omega})^{\phi,\Theta,\overline{\Delta'}} = (\alpha^{\Sigma} \wedge \beta^{\Omega})^{\overline{\Delta}}$.\end{proof}

\begin{lemma} \label{lemmaRa}
Given $\Delta$ without existential quantifiers, if $(\alpha^{\Sigma} \ra \beta^{\Omega})^{\overline{\Delta}}$ is wff, then it implies $\alpha^{\Sigma,\overline{\Delta}} \ra \beta^{\Omega,\overline{\Delta}}$.
\end{lemma}

\begin{proof} We proceed by induction on the size of $\Delta$:\\
\noindent If $\Delta$ is empty, then the implication is true;\\
\noindent (base) If $\Delta$ contains only one label, it must be a neighbourhood label:
\begin{itemize}
\item[-] $(\alpha^{\Sigma} \ra \beta^{\Omega})^{\circledast}$ means, by definition, that $\forall N \in \$(\chi) : \ob \m{W} , \$ , \m{V} , \chi , N \cb \models \alpha^{\Sigma} \ra \beta^{\Omega}$. Then we know that $\forall N \in \$(\chi) : \ob \m{W} , \$ , \m{V} , \chi , N \cb \not\models \alpha^{\Sigma} \mbox{ or } \ob \m{W} , \$ , \m{V} , \chi , N \cb \models \beta^{\Omega}$. So, if we have $\forall N \in \$(\chi) : \ob \m{W} , \$ , \m{V} , \chi , N \cb \models \alpha^{\Sigma} $, we must have $\forall N \in \$(\chi) : \ob \m{W} , \$ , \m{V} , \chi , N \cb \models \beta^{\Omega}$. In other words, $\alpha^{\Sigma,\circledast} \ra \beta^{\Omega,\circledast}$;
\item[-] $(\alpha^{\Sigma} \ra \beta^{\Omega})^{N}$ means, by definition, that $\sigma(N) \in \$(\chi)$ and $\ob \m{W} , \$ , \m{V} , \chi , \sigma(N) \cb \models \alpha^{\Sigma} \ra \beta^{\Omega}$. Then we know that $\sigma(N) \in \$(\chi)$ and ($\ob \m{W} , \$ , \m{V} , \chi , \sigma(N) \cb \not\models \alpha^{\Sigma} \mbox{ or } \ob \m{W} , \$ , \m{V} , \chi , \sigma(N) \cb \models \beta^{\Omega}$). So, if we have $\ob \m{W} , \$ , \m{V} , \chi , \sigma(N) \cb \models \alpha^{\Sigma} $, we must have $\ob \m{W} , \$ , \m{V} , \chi , \sigma(N) \cb \models \beta^{\Omega}$. In other words, $\alpha^{\Sigma,N} \ra \beta^{\Omega,N}$.
\end{itemize}

\noindent (base) If $\Delta$ contains two labels, it may be $\{\circledast,\ast\}$, $\{N,\ast\}$, $\{\circledast,u\}$ or $\{N,u\}$. But we just need to look at the distributivity for the $\ast$ label and for world variables, because we have already seen the distributivity of the $\ra$ connective for the label $\circledast$ and for any neighbourhood variable.
\begin{itemize}
\item[-] $(\alpha^{\Sigma} \ra \beta^{\Omega})^{\ast,\circledast}$ means, by definition, that $\forall N \in \$(\chi) : \forall w \in N : \ob \m{W} , \$ , \m{V} , w \cb \models \alpha^{\Sigma} \ra \beta^{\Omega}$. Then we know that $\forall N \in \$(\chi) : \forall w \in N : \ob \m{W} , \$ , \m{V} , w \cb \not\models \alpha^{\Sigma} \mbox{ or } \ob \m{W} , \$ , \m{V} , w \cb \models \beta^{\Omega}$. So, if we have $\forall N \in \$(\chi) : \forall w \in N : \ob \m{W} , \$ , \m{V} , w \cb \models \alpha^{\Sigma} $, we must have $\forall N \in \$(\chi) : \forall w \in N : \ob \m{W} , \$ , \m{V} , w \cb \models \beta^{\Omega}$. In other words, $\alpha^{\Sigma,\ast,\circledast} \ra \beta^{\Omega,\ast,\circledast}$;
\item[-] The proofs of $(\alpha^{\Sigma} \ra \beta^{\Omega})^{\ast,N}$, $(\alpha^{\Sigma} \ra \beta^{\Omega})^{u,\circledast}$ and $(\alpha^{\Sigma} \ra \beta^{\Omega})^{u,N}$ are analogous.
\end{itemize}

\noindent (induction) If $\alpha^{\Sigma} \ra \beta^{\Omega} \in \bs{F}_{w}$ and $s(\Delta) = n + 1$, then $\oc\Delta \in \bs{L}_{n}$ and $\alpha^{\Sigma,\overline{\Delta}} \ra \beta^{\Omega,\overline{\Delta}}$ may be written as $\alpha^{\Sigma,\phi,\overline{\Delta'}} \ra \beta^{\Omega,\phi,\overline{\Delta'}}$, where $s(\Delta') = n$. Then, by the induction hypothesis, $\alpha^{\Sigma,\phi,\overline{\Delta'}} \ra \beta^{\Omega,\phi,\overline{\Delta'}} = (\alpha^{\Sigma,\phi} \ra \beta^{\Omega,\phi})^{\overline{\Delta'}}$. From the base assertions, $(\alpha^{\Sigma,\phi} \ra \beta^{\Omega,\phi})^{\overline{\Delta'}} = ((\alpha^{\Sigma} \ra \beta^{\Omega})^{\phi})^{\overline{\Delta'}} = (\alpha^{\Sigma} \ra \beta^{\Omega})^{\phi,\overline{\Delta'}} = (\alpha^{\Sigma} \ra \beta^{\Omega})^{\overline{\Delta}}$;\\
\noindent (induction) If $\alpha^{\Sigma} \ra \beta^{\Omega} \in \bs{F}_{n}$ and $s(\Delta) = n + 2$, then $\oc\Delta \in \bs{L}_{w}$ and $\alpha^{\Sigma,\overline{\Delta}} \ra \beta^{\Omega,\overline{\Delta}}$ may be written as $\alpha^{\Sigma,\phi,\Theta,\overline{\Delta'}} \ra \beta^{\Omega,\phi,\Theta,\overline{\Delta'}}$, where $s(\Delta') = n$. Then, by the induction hypothesis, $\alpha^{\Sigma,\phi,\Theta,\overline{\Delta'}} \ra \beta^{\Omega,\phi,\Theta,\overline{\Delta'}} = (\alpha^{\Sigma,\phi,\Theta} \ra \beta^{\Omega,\phi,\Theta})^{\overline{\Delta'}}$. By the base, $(\alpha^{\Sigma,\phi,\Theta} \ra \beta^{\Omega,\phi,\Theta})^{\overline{\Delta'}} = ((\alpha^{\Sigma} \ra \beta^{\Omega})^{\phi,\Theta})^{\overline{\Delta'}} = (\alpha^{\Sigma} \ra \beta^{\Omega})^{\phi,\Theta,\overline{\Delta'}} = (\alpha^{\Sigma} \ra \beta^{\Omega})^{\overline{\Delta}}$.\end{proof}

Now we prove one of the main lemmas, in which, from the resolution of the hypothesis, follow the resolution of the conclusion. We express this property by saying that PUC-ND preserves resolution. 

\begin{lemma} \label{resolution}
PUC-ND without the rules $5, 7, 11, 18, 20, 27, 28$ and $29$ preserves resolution.
\end{lemma}

\begin{proof} Consider $\m{M} = \ob \m{W} , \$ , \m{V} , \chi \cb$.
\begin{enumerate}
\item If $\m{M} \models^{\Delta} \alpha^{\Sigma} \wedge \beta^{\Omega}$, then $\m{M} \models (\alpha^{\Sigma} \wedge \beta^{\Omega})^{\overline{\Delta}}$, and, by lemma \ref{lemmaWedge}, $\m{M} \models \alpha^{\Sigma,\overline{\Delta}} \wedge \beta^{\Omega,\overline{\Delta}}$, which means, by definition, $\m{M} \models \alpha^{\Sigma,\overline{\Delta}}$ and $\m{M} \models \beta^{\Omega,\overline{\Delta}}$. So, we have $\m{M} \models^{\Delta} \alpha^{\Sigma}$;
\item Follow the same argument for rule 1;
\item If $\m{M} \models^{\Delta} \alpha^{\Sigma}$ and $\m{M} \models^{\Delta} \beta^{\Omega}$, then $\m{M} \models \alpha^{\Sigma,\overline{\Delta}}$ and $\m{M} \models \beta^{\Omega,\overline{\Delta}}$, then, by definition, $\m{M} \models \alpha^{\Sigma,\overline{\Delta}} \wedge \beta^{\Omega,\overline{\Delta}}$, then, by lemma \ref{lemmaWedge}, $\m{M} \models (\alpha^{\Sigma} \wedge \beta^{\Omega})^{\overline{\Delta}}$, then, by definition, $\m{M} \models^{\Delta} \alpha^{\Sigma} \wedge \beta^{\Omega}$;
\item If $\m{M} \models^{\Delta} \alpha^{\Sigma}$, then $\m{M} \models \alpha^{\Sigma,\overline{\Delta}}$, and, by definition, $\m{M} \models \alpha^{\Sigma,\overline{\Delta}} \vee \beta^{\Omega,\overline{\Delta}}$, then, by lemma \ref{lemmaVee}, $\m{M} \models (\alpha^{\Sigma} \vee \beta^{\Omega})^{\overline{\Delta}}$, and, by definition, $\m{M} \models^{\Delta} \alpha^{\Sigma} \vee \beta^{\Omega}$;

\setcounter{enumi}{5}
\item Follow the same argument for rule 4;

\setcounter{enumi}{7}
\item By definition, there is no template $\m{T}$, such that $\m{T} \models \bot_{w}$. So, by definition, for every $\alpha^{\Sigma} \in \bs{F}_{w}$, $\bot_{w} \models \alpha^{\Sigma}$ and, by lemma \ref{lemmaDirectConsequence}, $\m{M} \models^{\Delta} \alpha^{\Sigma}$. The same argument holds for $\bot_{n}$ considering formulas in $\bs{F}_{n}$;
\item If $\Delta = \{\circledcirc\}$, then $\m{M} \models^{\Delta} \bot_{w}$ means $\ob \m{W} , \$ , \m{V} , \chi \cb \models \bot_{w}^{\circledcirc}$. This means that $\exists N \in \$(\chi) : \ob \m{W} , \$ , \m{V} , \chi , N \cb \models \bot_{w}$, but, by definition, $\nexists N \in \$(\chi) : \ob \m{W} , \$ , \m{V} , \chi , N \cb \models \bot_{w}$, so $\ob \m{W} , \$ , \m{V} , \chi \cb \models \neg (\bot_{w}^{\circledcirc})$. Then, by the rule 3, $\ob \m{W} , \$ , \m{V} , \chi \cb \models \bot_{n}$ and, by definition, $\m{M} \models \bot_{n}$. The case $\Delta = \{N\}$ is similar. If $\Delta = \{\circledcirc,\bullet\}$, then $\m{M} \models^{\Delta} \bot_{n}$ means $\ob \m{W} , \$ , \m{V} , \chi \cb \models \bot_{n}^{\bullet,\circledcirc}$. But this means that $\exists N \in \$(\chi) : \ob \m{W} , \$ , \m{V} , \chi , N \cb \models \bot_{n}^{\bullet}$ and $\exists w \in N : \ob \m{W} , \$ , \m{V} , w \cb \models \bot_{n}$. But, by definition, $\nexists w \in N : \ob \m{W} , \$ , \m{V} , w \cb \models \bot_{n}$, so $\ob \m{W} , \$ , \m{V} , \chi , N \cb \models \neg (\bot_{n}^{\bullet})$. Using rule 3, we conclude that $\ob \m{W} , \$ , \m{V} , \chi , N \cb \models \bot_{w}$ and, by a previous case, $\ob \m{W} , \$ , \m{V} , \chi \cb \models \bot_{n}$. The other cases where $s(\Delta) = 2$ are similar. If $\Delta = \{\circledcirc,\bullet,\circledcirc\}$, then $\m{M} \models^{\Delta} \bot_{w}$ means $\ob \m{W} , \$ , \m{V} , \chi \cb \models \bot_{w}^{\circledcirc,\bullet,\circledcirc}$. But this means that $\exists N \in \$(\chi) : \ob \m{W} , \$ , \m{V} , \chi , N \cb \models \bot_{w}^{\circledcirc,\bullet}$ and $\exists w \in N : \ob \m{W} , \$ , \m{V} , w \cb \models \bot_{w}^{\circledcirc}$. But, by a previous case, it means that $\exists w \in N : \ob \m{W} , \$ , \m{V} , w \cb \models \bot_{n}$ and $\ob \m{W} , \$ , \m{V} , \chi , N \cb \models \bot_{n}^{\bullet}$. But, by definition, $\nexists w \in N : \ob \m{W} , \$ , \m{V} , w \cb \models \bot_{n}$ and $\ob \m{W} , \$ , \m{V} , \chi , N \cb \models \neg(\bot_{n}^{\bullet})$. So, using rule 3, $\ob \m{W} , \$ , \m{V} , \chi , N \cb \models \bot_{w}$. Then $\exists N \in \$(\chi) : \ob \m{W} , \$ , \m{V} , \chi , N \cb \models \bot_{w}$ and $\ob \m{W} , \$ , \m{V} , \chi \cb \models \bot_{w}^{\circledcirc}$. By a previous case, we conclude that $\ob \m{W} , \$ , \m{V} , \chi \cb \models \bot_{n}$. The other cases where $s(\Delta) = 3$ are similar. If $\Delta = \{\circledcirc,\bullet,\circledcirc,\bullet\}$, then $\m{M} \models^{\Delta} \bot_{n}$ means $\ob \m{W} , \$ , \m{V} , \chi \cb \models \bot_{n}^{\bullet,\circledcirc,\bullet,\circledcirc}$. But this means that $\exists N \in \$(\chi) : \ob \m{W} , \$ , \m{V} , \chi , N \cb \models \bot_{n}^{\bullet,\circledcirc,\bullet}$ and $\exists w \in N : \ob \m{W} , \$ , \m{V} , w \cb \models \bot_{n}^{\bullet,\circledcirc}$ and, by the above arguments, $\ob \m{W} , \$ , \m{V} , w \cb \models \bot_{n}$. But, by definition,  $\nexists w \in N : \ob \m{W} , \$ , \m{V} , w \cb \models \bot_{n}$, so $\ob \m{W} , \$ , \m{V} , \chi , N \cb \models \neg (\bot_{n}^{\bullet,\circledcirc,\bullet})$ because of the implication of $\bot_{n}$ from $\bot_{n}^{\bullet,\circledcirc}$. Using rule 3, we conclude that $\ob \m{W} , \$ , \m{V} , \chi , N \cb \models \bot_{w}$ and $\ob \m{W} , \$ , \m{V} , \chi \cb \models \bot_{n}$ by a previous argument. The other cases are similar and the general case is treated by induction on the size of $\Delta$ following the previous arguments;
\item If $\m{M} \models^{\Delta} \alpha^{\Sigma}$, then $\m{M} \models^{\Delta} \alpha^{\Sigma}$;

\setcounter{enumi}{11}
\item If $\m{M} \models^{\Delta} \alpha^{\Sigma} \ra \beta^{\Omega}$, then $\m{M} \models (\alpha^{\Sigma} \ra \beta^{\Omega})^{\overline{\Delta}}$, then, by lemma \ref{lemmaRa}, $\m{M} \models \alpha^{\Sigma,\overline{\Delta}} \ra \beta^{\Omega,\overline{\Delta}}$. Then, by definition, $\m{M} \models \neg (\alpha^{\Sigma,\overline{\Delta}})$ or $\m{M} \models \beta^{\Omega,\overline{\Delta}}$. But we know from $\m{M} \models^{\Delta} \alpha^{\Sigma}$ that $\m{M} \models \alpha^{\Sigma,\overline{\Delta}}$. So, we can conclude $\m{M} \models^{\Delta} \beta^{\Omega}$;
\item If $\m{M} \models^{\Delta} \alpha^{\Sigma,\phi}$, then $\m{M} \models \alpha^{\Sigma,\phi,\overline{\Delta}}$. But, $\{\phi,\overline{\Delta}\} \equiv \overline{\{{\Delta,\phi}\}}$, then, by definition, $\m{M} \models^{\Delta,\phi} \alpha^{\Sigma}$;
\item If $\m{M} \models^{\Delta,\phi} \alpha^{\Sigma}$, then $\m{M} \models \alpha^{\Sigma,\overline{\{\Delta,\phi\}}}$. But, $\overline{\{{\Delta,\phi}\}} \equiv \{\phi,\overline{\Delta}\}$, and, by definition, $\m{M} \models^{\Delta} \alpha^{\Sigma,\phi}$;
\item If $\m{M} \models^{\Delta,u} \alpha^{\Sigma}$, then, by the rule 14, $\m{M} \models^{\Delta} \alpha^{\Sigma,u}$. By the fact that $\alpha^{\Sigma,u} \in \bs{F}_{w}$, the fitting relation and lemma \ref{countingLemma}, we know that $s(\Delta)$ is odd. If we take some template $\m{T} = \ob \m{W} , \$ , \m{V} , z , N \cb$, such that $\m{M} \multimap_{s(\Delta)} \m{T}$ and $\m{T} \models \alpha^{\Sigma,u}$, we can conclude that $N \in \$(z)$, $\sigma(u) \in N$ and $\ob \m{W} , \$ , \m{V} , \sigma(u) \cb \models \alpha^{\Sigma}$. The restrictions of the rule assures us that the variable $u$ is arbitrary and we may conclude that $\forall w \in N : \ob \m{W} , \$ , \m{V} , w \cb \models \alpha^{\Sigma}$. So, $\m{T} \models \alpha^{\Sigma,\ast}$ and, by definition, $\alpha^{\Sigma,u} \models_{\m{M}:s(\Delta)} \alpha^{\Sigma,\ast}$, which means, by lemma \ref{lemmaConsequence}, that $\m{M} \models^{\Delta} \alpha^{\Sigma,\ast}$ and, by rule 13, $\m{M} \models^{\Delta,\ast} \alpha^{\Sigma}$;
\item If $\m{M} \models^{\Delta,\ast} \alpha^{\Sigma}$, then, by the rule 14, $\m{M} \models^{\Delta} \alpha^{\Sigma,\ast}$. By the fact that $\alpha^{\Sigma,\ast} \in \bs{F}_{w}$, the fitting relation and lemma \ref{countingLemma}, we know that $s(\Delta)$ is odd. If we take some template $\m{T} = \ob \m{W} , \$ , \m{V} , z \cb$, such that $\m{M} \multimap_{s(\Delta)} \m{T}$ and $\m{T} \models \alpha^{\Sigma,\ast}$, then $N \in \$(z)$ and $\forall w \in N : \ob \m{W} , \$ , \m{V} , w \cb \models \alpha^{\Sigma}$. If we take a variable $u$ to denote a world of $N$ obeying the restrictions of the rule, then we may conclude that $u \in N$ and $\ob \m{W} , \$ , \m{V} , u \cb \models \alpha^{\Sigma}$. So, $\m{T} \models \alpha^{\Sigma,u}$ and, by definition, $\alpha^{\Sigma,\ast} \models_{\m{M}:s(\Delta)} \alpha^{\Sigma,u}$, which means, by lemma \ref{lemmaConsequence}, that $\m{M} \models^{\Delta} \alpha^{\Sigma,u}$ and, by rule 13, $\m{M} \models^{\Delta,u} \alpha^{\Sigma}$;

\item If $\m{M} \models^{\Delta,u} \alpha^{\Sigma}$, then, by the rule 14, $\m{M} \models^{\Delta} \alpha^{\Sigma,u}$. By the fact that $\alpha^{\Sigma,u} \in \bs{F}_{w}$, the fitting relation and lemma \ref{countingLemma}, we know that $s(\Delta)$ is odd. If we take some template $\m{T} = \ob \m{W} , \$ , \m{V} , z \cb$, such that $\m{M} \multimap_{s(\Delta)} \m{T}$ and $\m{T} \models \alpha^{\Sigma,u}$, then $N \in \$(z)$, $\sigma(u) \in N$ and $\ob \m{W} , \$ , \m{V} , \sigma(u) \cb \models \alpha^{\Sigma}$. Since we denote some world with the variable $u$, we know that there is some world in $N$ such that the formula $\alpha^{\Sigma}$ holds. Then we may conclude that $\exists w \in N : \ob \m{W} , \$ , \m{V} , w \cb \models \alpha^{\Sigma}$. So, $\m{T} \models \alpha^{\Sigma,\bullet}$ and, by definition, $\alpha^{\Sigma,u} \models_{\m{M}:s(\Delta)} \alpha^{\Sigma,\bullet}$, which means, by lemma \ref{lemmaConsequence}, that $\m{M} \models^{\Delta} \alpha^{\Sigma,\bullet}$ and, by rule 13, $\m{M} \models^{\Delta,\bullet} \alpha^{\Sigma}$;

\setcounter{enumi}{18}
\item If $\m{M} \models^{\Delta,N} \alpha^{\Sigma}$ and $\m{M} \models^{\Delta,\circledcirc} \beta^{\Omega}$, then, by the rule 14, $\m{M} \models^{\Delta} \alpha^{\Sigma,N}$ and $\m{M} \models^{\Delta} \beta^{\Omega,\circledcirc}$ and, by rule 3, $\m{M} \models^{\Delta} \alpha^{\Sigma,N} \wedge \beta^{\Omega,\circledcirc}$. By the fact that $\alpha^{\Sigma,N} \wedge \beta^{\Omega,\circledcirc} \in \bs{F}_{n}$, the fitting relation and lemma \ref{countingLemma}, we know that $s(\Delta)$ is even. If we take some model $\m{H} = \ob \m{W} , \$ , \m{V} , z \cb$, such that $\m{M} \multimap_{s(\Delta)} \m{H}$ and $\m{H} \models \alpha^{\Sigma,N} \wedge \beta^{\Omega,\circledcirc}$, then from $\beta^{\Omega,\circledcirc}$ we know that $\$(z) \neq \emptyset$, $\sigma(N) \in \$(z)$ and $\ob \m{W} , \$ , \m{V} , z , \sigma(N) \cb \models \alpha^{\Sigma}$. Since we denote some neighbourhood with the variable $N$, we know that there is some neighbourhood in $\$(z)$, such that the formula $\alpha^{\Sigma}$ holds. Then $\exists M \in \$(z) : \ob \m{W} , \$ , \m{V} , z , M \cb \models \alpha^{\Sigma}$ and $\m{H} \models \alpha^{\Sigma,\circledcirc}$. So, by definition, $\alpha^{\Sigma,N} \wedge \beta^{\Omega,\circledcirc} \models_{\m{M}:s(\Delta)} \alpha^{\Sigma,\circledcirc}$, which means, by lemma \ref{lemmaConsequence}, that $\m{M} \models^{\Delta} \alpha^{\Sigma,\circledcirc}$ and, by rule 13, $\m{M} \models^{\Delta,\circledcirc} \alpha^{\Sigma}$;

\setcounter{enumi}{20}
\item If $\m{M} \models^{\Delta,N} \alpha^{\Sigma}$, then, by the rule 14, $\m{M} \models^{\Delta} \alpha^{\Sigma,N}$. By the fact that $\alpha^{\Sigma,N} \in \bs{F}_{n}$, the fitting relation and lemma \ref{countingLemma}, we know that $s(\Delta)$ is even. If we take some model $\m{H} = \ob \m{W} , \$ , \m{V} , z \cb$, such that $\m{M} \multimap_{s(\Delta)} \m{H}$ and $\m{H} \models \alpha^{\Sigma,N}$, then $\sigma(N) \in \$(z)$ and $\ob \m{W} , \$ , \m{V} , z , \sigma(N) \cb \models \alpha^{\Sigma}$. From the restrictions of the rule, we know that $N$ is arbitrary, so, $\forall M \in \$(z) : \ob \m{W} , \$ , \m{V} , z , M \cb \models \alpha^{\Sigma}$, which means that $\m{H} \models \alpha^{\Sigma,\circledast}$. So, by definition, $\alpha^{\Sigma,N} \models_{\m{M}:s(\Delta)} \alpha^{\Sigma,\circledast}$, which means, by lemma \ref{lemmaConsequence}, that $\m{M} \models^{\Delta} \alpha^{\Sigma,\circledast}$ and, by rule 13, $\m{M} \models^{\Delta,\circledast} \alpha^{\Sigma}$;
\item If $\m{M} \models^{\Delta,\circledast} \alpha^{\Sigma}$ and $\m{M} \models^{\Delta,N} \beta^{\Omega}$, then, by the rule 14, $\m{M} \models^{\Delta} \alpha^{\Sigma,\circledast}$ and $\m{M} \models^{\Delta} \beta^{\Omega,N}$. So, by rule 3, $\m{M} \models^{\Delta} \alpha^{\Sigma,\circledast} \wedge \beta^{\Omega,N}$. By the fact that $\alpha^{\Sigma,\circledast} \wedge \beta^{\Omega,N} \in \bs{F}_{n}$, the fitting relation and lemma \ref{countingLemma}, we know that $s(\Delta)$ is even. If we take some model $\m{H} = \ob \m{W} , \$ , \m{V} , z \cb$, such that $\m{M} \multimap_{s(\Delta)} \m{H}$ and $\m{H} \models \alpha^{\Sigma,\circledast} \wedge \beta^{\Omega,N}$, then $\m{H} \models \alpha^{\Sigma,\circledast}$ and $\m{H} \models \beta^{\Omega,N}$. By definition, $\sigma(N) \in \$(z)$ and $\ob \m{W} , \$ , \m{V} , z , \sigma(N) \cb \models \beta^{\Omega}$ and $\forall M \in \$(z) : \ob \m{W} , \$ , \m{V} , z , M \cb \models \alpha^{\Sigma}$. So, $\sigma(N) \in \$(z)$ and, by the universal quantification, $\ob \m{W} , \$ , \m{V} , z , \sigma(N) \cb \models \alpha^{\Sigma}$. This means that $\m{H} \models \alpha^{\Sigma,N}$ and, by definition, $\alpha^{\Sigma,\circledast} \wedge \beta^{\Omega,N} \models_{\m{M}:s(\Delta)} \alpha^{\Sigma,N}$, which means, by lemma \ref{lemmaConsequence}, that $\m{M} \models^{\Delta} \alpha^{\Sigma,N}$ and, by rule 13, $\m{M} \models^{\Delta,N} \alpha^{\Sigma}$;
\item If $\m{M} \models^{\Delta,N} \alpha^{\Sigma,\bullet}$ and $\m{M} \models^{\Delta,M} \shneg N$, then, by the rule 14, $\m{M} \models^{\Delta} \alpha^{\Sigma,\bullet,N}$ and $\m{M} \models^{\Delta} (\shneg N)^{M}$. By the rule 3, $\m{M} \models^{\Delta} \alpha^{\Sigma,\bullet,N} \wedge (\shneg N)^{M}$. By the fact that $\alpha^{\Sigma,\bullet,N} \wedge (\shneg N)^{M} \in \bs{F}_{n}$, the fitting relation and lemma \ref{countingLemma}, we know that $s(\Delta)$ is even. If we take some model $\m{H} = \ob \m{W} , \$ , \m{V} , z \cb$, such that $\m{M} \multimap_{s(\Delta)} \m{H}$ and $\m{H} \models \alpha^{\Sigma,\bullet,N} \wedge (\shneg N)^{M}$, then $\sigma(N) \in \$(z)$ and $\exists w \in \sigma(N) : \ob \m{W} , \$ , \m{V} , w \cb \models \alpha^{\Sigma}$. From $(\shneg N)^{M}$, we know that $\sigma(M) \in \$(z)$ and $\sigma(N) \subset \sigma(M)$, then $\exists w \in \sigma(M) : \ob \m{W} , \$ , \m{V} , w \cb \models \alpha^{\Sigma}$. We conclude that $\m{H} \models \alpha^{\Sigma,\bullet,M}$ and, by definition, $\alpha^{\Sigma,\bullet,N} \wedge (\shneg N)^{M} \models_{\m{M}:s(\Delta)} \alpha^{\Sigma,\bullet,M}$, which means, by lemma \ref{lemmaConsequence}, that $\m{M} \models^{\Delta} \alpha^{\Sigma,\bullet,M}$ and, by rule 13, $\m{M} \models^{\Delta,M} \alpha^{\Sigma,\bullet}$;
\item If $\m{M} \models^{\Delta,N} \alpha^{\Sigma,\ast}$ and $\m{M} \models^{\Delta,M} \shpos N$, then, by the rule 14, $\m{M} \models^{\Delta} \alpha^{\Sigma,\ast,N}$ and $\m{M} \models^{\Delta} (\shpos N)^{M}$. By the rule 3, $\m{M} \models^{\Delta} \alpha^{\Sigma,\ast,N} \wedge (\shpos N)^{M}$. By the fact that $\alpha^{\Sigma,\ast,N} \wedge (\shpos N)^{M} \in \bs{F}_{n}$, the fitting relation and lemma \ref{countingLemma}, we know that $s(\Delta)$ is even. If we take some model $\m{H} = \ob \m{W} , \$ , \m{V} , z \cb$, such that $\m{M} \multimap_{s(\Delta)} \m{H}$ and $\m{H} \models \alpha^{\Sigma,\ast,N} \wedge (\shpos N)^{M}$, then $\sigma(N) \in \$(z)$ and $\forall w \in \sigma(N) : \ob \m{W} , \$ , \m{V} , w \cb \models \alpha^{\Sigma}$. From $(\shpos N)^{M}$, we know that $\sigma(M) \in \$(z)$ and $\sigma(M) \subset \sigma(N)$, then $\forall w \in \sigma(M) : \ob \m{W} , \$ , \m{V} , w \cb \models \alpha^{\Sigma}$. We conclude that $\m{H} \models \alpha^{\Sigma,\ast,M}$ and, by definition, $\alpha^{\Sigma,\ast,N} \wedge (\shpos N)^{M} \models_{\m{M}:s(\Delta)} \alpha^{\Sigma,\ast,M}$, which means, by lemma \ref{lemmaConsequence}, that $\m{M} \models^{\Delta} \alpha^{\Sigma,\ast,M}$ and, by rule 13, $\m{M} \models^{\Delta,M} \alpha^{\Sigma,\ast}$;
\item If $\m{M} \models^{\Delta,N} \shneg M$ and $\m{M} \models^{\Delta,M} \shneg P$, then, by the rule 14, $\m{M} \models^{\Delta} (\shneg M)^{N}$ and $\m{M} \models^{\Delta} (\shneg P)^{M}$. By the rule 3, $\m{M} \models^{\Delta} (\shneg M)^{N} \wedge (\shneg P)^{M}$. By the fact that $(\shneg M)^{N} \wedge (\shneg P)^{M} \in \bs{F}_{n}$, the fitting relation and lemma \ref{countingLemma}, we know that $s(\Delta)$ is even. If we take some model $\m{H} = \ob \m{W} , \$ , \m{V} , z \cb$, such that $\m{M} \multimap_{s(\Delta)} \m{H}$ and $\m{H} \models (\shneg M)^{N} \wedge (\shneg P)^{M}$, then $\sigma(N) \in \$(z)$ and $\sigma(M) \subset \sigma(N)$. From $(\shneg P)^{M}$, we know that $\sigma(M) \in \$(z)$ and $\sigma(P) \subset \sigma(M)$, then $\sigma(P) \subset \sigma(N)$. We conclude that $\m{H} \models (\shneg P)^{N}$ and, by definition, $(\shneg M)^{N} \wedge (\shneg P)^{M} \models_{\m{M}:s(\Delta)} (\shneg P)^{N}$, which means, by lemma \ref{lemmaConsequence}, that $\m{M} \models^{\Delta} (\shneg P)^{N}$ and, by rule 13, $\m{M} \models^{\Delta,N} \shneg P$;
\item It follows the same argument of rule 25;

\setcounter{enumi}{29}
\item According to the satisfaction relation, every model must model $\top_{n}$ and every template must model $\top_{w}$. So, given a model $\m{M}$, if $s(\Delta)$ is even, then, for every model $\m{H}$, such that $\m{M} \multimap_{s(\Delta)} \m{H}$, $\m{H} \models \top_{n}$ and, by lemma \ref{lemmaConsequence}, $\m{M} \models^{\Delta} \top_{n}$. The argument for odd $s(\Delta)$ is analogous.\end{enumerate}\end{proof}

\begin{lemma} \label{withOrWithoutYou}
Given a context $\Delta$ with no existential label, and a wff $\alpha^{\Sigma}$ that fits on $\Delta$, then, for any model, $\m{M} \models^{\Delta} \alpha^{\Sigma} \vee \neg ( \alpha^{\Sigma} )$.
\end{lemma}

\begin{proof}
We proceed by induction on the size of $\Delta$.\\
\noindent If $\Delta$ is empty, then $\alpha^{\Sigma} \in \bs{F}_{n}$. $\alpha^{\Sigma} \vee \neg ( \alpha^{\Sigma} )$ is a tautology because of the satisfaction relation definition: given any model $\m{M}$, if $\m{M} \models \alpha^{\Sigma}$, then $\m{M} \models \alpha^{\Sigma} \vee \neg ( \alpha^{\Sigma} )$. If  $\m{M} \not\models \alpha^{\Sigma}$, then $\m{M} \models \neg ( \alpha^{\Sigma} )$ and $\m{M} \models \alpha^{\Sigma} \vee \neg ( \alpha^{\Sigma} )$.\\
\noindent (base) If $\Delta = \{\circledast\}$, then $\alpha^{\Sigma} \in \bs{F}_{w}$. $\alpha^{\Sigma} \vee \neg ( \alpha^{\Sigma} )$ is a tautology because of the satisfaction relation definition: given any template $\m{T}$, if $\m{T} \models \alpha^{\Sigma}$, then $\m{T} \models \alpha^{\Sigma} \vee \neg ( \alpha^{\Sigma} )$. If $\m{T} \not\models \alpha^{\Sigma}$, then $\m{T} \models \neg ( \alpha^{\Sigma} )$ and $\m{T} \models \alpha^{\Sigma} \vee \neg ( \alpha^{\Sigma} )$. Given any model $\m{M} = \ob \m{W} , \$ , \m{V} , \chi \cb$, then for every template $\ob \m{W} , \$ , \m{V} , \chi , N \cb \models \alpha^{\Sigma} \vee \neg ( \alpha^{\Sigma} )$ and, by definition, $\m{M} \models (\alpha^{\Sigma} \vee \neg ( \alpha^{\Sigma} ))^{\circledast}$. So, $\m{M} \models (\alpha^{\Sigma} \vee \neg ( \alpha^{\Sigma} ))^{\overline{\Delta}}$ and, by definition, $\m{M} \models^{\Delta} \alpha^{\Sigma} \vee \neg ( \alpha^{\Sigma} )$.\\
\noindent (base) If $\Delta = \{N\}$: by the previous case, $\m{M} \models (\alpha^{\Sigma} \vee \neg ( \alpha^{\Sigma} ))^{\circledast}$ and, in particular, $\m{M} \models (\alpha^{\Sigma} \vee \neg ( \alpha^{\Sigma} ))^{N}$, for any neighbourhood variable $N$.\\
\noindent (base) If $\Delta = \{\circledast,\ast\}$, then $\alpha^{\Sigma} \in \bs{F}_{n}$. $\alpha^{\Sigma} \vee \neg ( \alpha^{\Sigma} )$ is a tautology because of the satisfaction relation definition: given any model $\m{H}$, if $\m{H} \models \alpha^{\Sigma}$, then $\m{H} \models \alpha^{\Sigma} \vee \neg ( \alpha^{\Sigma} )$. If $\m{H} \not\models \alpha^{\Sigma}$, then $\m{H} \models \neg ( \alpha^{\Sigma} )$ and $\m{H} \models \alpha^{\Sigma} \vee \neg ( \alpha^{\Sigma} )$. We apply lemma \ref{taut} to conclude that $(\alpha^{\Sigma} \vee \neg ( \alpha^{\Sigma} ))^{\ast,\circledast}$ is also a tautology. So, for any model $\m{M} \models (\alpha^{\Sigma} \vee \neg ( \alpha^{\Sigma} ))^{\ast,\circledast}$ and, by definition, $\m{M} \models^{\Delta} \alpha^{\Sigma} \vee \neg ( \alpha^{\Sigma} )$.\\
\noindent (base) If $\Delta = \{\circledast,u\}$: by the previous case $(\alpha^{\Sigma} \vee \neg ( \alpha^{\Sigma} ))^{\ast,\circledast}$ is a tautology. So, in particular, $\m{M} \models (\alpha^{\Sigma} \vee \neg ( \alpha^{\Sigma} ))^{u,\circledast}$ for any world variable and, by definition, $\m{M} \models^{\Delta} \alpha^{\Sigma} \vee \neg ( \alpha^{\Sigma} )$.\\
\noindent (base) If $\Delta = \{N,\ast\}$ and $\Delta = \{N,u\}$ are analogous to the previous case.\\
\noindent (induction) If $\Delta = \{\phi,\Delta'\}$: by lemma \ref{lemmaVee}, $(\alpha^{\Sigma} \vee \neg ( \alpha^{\Sigma} ))^{\phi,\Delta'} \equiv (\alpha^{\Sigma,\phi} \vee (\neg ( \alpha^{\Sigma} )^{\phi} )^{\Delta'}$. By the induction hypothesis, $\m{M} \models^{\Delta'} \alpha^{\Sigma,\phi} \vee (\neg ( \alpha^{\Sigma} ))^{\phi}$. By lemma \ref{lemmaVee} again, $\m{M} \models^{\Delta'} (\alpha^{\Sigma} \vee (\neg ( \alpha^{\Sigma} ))^{\phi}$ and, by definition, $\m{M} \models^{\Delta',\phi} \alpha^{\Sigma} \vee \neg ( \alpha^{\Sigma} )$.\end{proof}

\begin{lemma} \label{completeResolutionPUC-ND}
PUC-ND preserves resolution.
\end{lemma}

\begin{proof} We present the proof for each remaining rule of the PUC-ND inside an induction. Base argument:
\begin{enumerate}
\setcounter{enumi}{4}
\item If $\m{M} \models^{\Delta} \alpha^{\Sigma} \vee \beta^{\Omega}$, then $\m{M} \models (\alpha^{\Sigma} \vee \beta^{\Omega})^{\overline{\Delta}}$, then, by lemma \ref{lemmaVee}, $\m{M} \models \alpha^{\Sigma,\overline{\Delta}} \vee \beta^{\Omega,\overline{\Delta}}$, then, by definition, $\m{M} \models \alpha^{\Sigma,\overline{\Delta}}$ or $\m{M} \models \beta^{\Omega,\overline{\Delta}}$. This means, by definition, that $\m{M} \models^{\Delta} \alpha^{\Sigma}$ or $\m{M} \models^{\Delta} \beta^{\Omega}$. So, if $\Pi_{1}$ and $\Pi_{2}$ only contains the rules from lemma \ref{resolution}, $\m{M} \models^{\Theta} \gamma^{\Lambda}$ in both cases, because of the preservation of the resolution relation. And, for that conclusion, the hypothesis are no longer necessary and may be discharged;

\setcounter{enumi}{6}
\item We know from classical logic that $\m{M} \models \alpha^{\Sigma,\overline{\Delta}} \vee \neg (\alpha^{\Sigma,\overline{\Delta}})$, which means that $\m{M} \models \alpha^{\Sigma,\overline{\Delta}}$ or $\m{M} \models \neg (\alpha^{\Sigma,\overline{\Delta}})$. In the first case, we know that $\m{M} \models^{\Delta} \alpha^{\Sigma}$. In the second case, we know that $\m{M} \models^{\Delta} \neg (\alpha^{\Sigma})$. If the subderivation $\Pi$ only contains the rules from lemma \ref{resolution}, we can conclude that $\m{M} \models^{\Delta} \bot$. But, from rule 7, this means that $\m{M} \models^{\Delta} \alpha^{\Sigma}$. So, in either case, we can conclude $\m{M} \models^{\Delta} \alpha^{\Sigma}$ and we are able to discharge the hypothesis;

\setcounter{enumi}{10}
\item From lemma \ref{withOrWithoutYou}, we know that $\m{M} \models^{\Delta} \alpha^{\Sigma} \vee \neg ( \alpha^{\Sigma} )$, so $\m{M} \models^{\Delta} \alpha^{\Sigma}$ or $\m{M} \models^{\Delta} \neg ( \alpha^{\Sigma} )$. In the first case, if $\Pi$ only contains the rules of lemma \ref{resolution}, then the derivation gives us $\m{M} \models^{\Delta} \beta^{\Omega}$. If $\beta^{\Omega} \in \bs{F}_{n}$, then, by the fitting relation and lemma \ref{countingLemma}, we know that $s(\Delta)$ is even. If we take some model $\m{H} = \ob \m{W} , \$ , \m{V} , z \cb$, such that $\m{M} \multimap_{s(\Delta)} \m{H}$ and $\m{H} \models \beta^{\Omega}$, then, by definition, $\m{H} \models \alpha^{\Sigma} \ra \beta^{\Omega}$. So, by definition, $\beta^{\Omega} \models_{\m{M}:s(\Delta)} \alpha^{\Sigma} \ra \beta^{\Omega}$, which means, by lemma \ref{lemmaConsequence}, that $\m{M} \models^{\Delta} \alpha^{\Sigma} \ra \beta^{\Omega}$. If $\beta^{\Omega} \in \bs{F}_{w}$, then, by the fitting relation and lemma \ref{countingLemma}, we know that $s(\Delta)$ is odd. If we take some template $\m{T} = \ob \m{W} , \$ , \m{V} , z , L\cb$, such that $\m{M} \multimap_{s(\Delta)} \m{T}$ and $\m{T} \models \beta^{\Omega}$, then, by definition, $\m{T} \models \alpha^{\Sigma} \ra \beta^{\Omega}$. So, by definition, $\beta^{\Omega} \models_{\m{M}:s(\Delta)} \alpha^{\Sigma} \ra \beta^{\Omega}$, which means, by lemma \ref{lemmaConsequence}, that $\m{M} \models^{\Delta} \alpha^{\Sigma} \ra \beta^{\Omega}$. In the case where $\m{M} \models^{\Delta} \neg (\alpha^{\Sigma})$, if $\neg (\alpha^{\Sigma}) \in \bs{F}_{n}$, then, by the fitting relation and lemma \ref{countingLemma}, we know that $s(\Delta)$ is even. If we take some model $\m{H} = \ob \m{W} , \$ , \m{V} , z \cb$, such that $\m{M} \multimap_{s(\Delta)} \m{H}$ and $\m{H} \models \neg (\alpha^{\Sigma})$, then, by definition, $\m{H} \models \alpha^{\Sigma} \ra \beta^{\Omega}$. So, by definition, $\neg (\alpha^{\Sigma}) \models_{\m{M}:s(\Delta)} \alpha^{\Sigma} \ra \beta^{\Omega}$, which means, by lemma \ref{lemmaConsequence}, that $\m{M} \models^{\Delta} \alpha^{\Sigma} \ra \beta^{\Omega}$. If $\neg (\alpha^{\Sigma}) \in \bs{F}_{w}$, then, by the fitting relation and lemma \ref{countingLemma}, we know that $s(\Delta)$ is odd. If we take some template $\m{T} = \ob \m{W} , \$ , \m{V} , z , L\cb$, such that $\m{M} \multimap_{s(\Delta)} \m{T}$ and $\m{T} \models \neg (\alpha^{\Sigma})$, then, by definition, $\m{T} \models \alpha^{\Sigma} \ra \beta^{\Omega}$. So, by definition, $\neg (\alpha^{\Sigma}) \models_{\m{M}:s(\Delta)} \alpha^{\Sigma} \ra \beta^{\Omega}$, which means, by lemma \ref{lemmaConsequence}, that $\m{M} \models^{\Delta} \alpha^{\Sigma} \ra \beta^{\Omega}$. So the hypothesis is unnecessary and may be discharged;

\setcounter{enumi}{17}
\item If $\m{M} \models^{\Delta,\bullet} \alpha^{\Sigma}$, then, by the rule 14, $\m{M} \models^{\Delta} \alpha^{\Sigma,\bullet}$. By the fact that $\alpha^{\Sigma,\bullet} \in \bs{F}_{w}$, the fitting relation and lemma \ref{countingLemma}, we know that $s(\Delta)$ is odd. If we take some template $\m{T} = \ob \m{W} , \$ , \m{V} , z , N \cb$, such that $\m{M} \multimap_{s(\Delta)} \m{T}$ and $\m{T} \models \alpha^{\Sigma,\bullet}$, then, $N \in \$(z)$ and $\exists w \in N : \ob \m{W} , \$ , \m{V} , w \cb \models \alpha^{\Sigma}$. Since the variable $u$ occurs nowhere else in the derivation, $u$ can be taken as a denotation of the given existential and we conclude that $\ob \m{W} , \$ , \m{V} , \sigma(u) \cb \models \alpha^{\Sigma}$, what means that $\m{T} \models \alpha^{\Sigma,u}$. So, by definition, $\alpha^{\Sigma,\bullet} \models_{\m{M}:s(\Delta)} \alpha^{\Sigma,u}$, which means, by lemma \ref{lemmaConsequence}, that $\m{M} \models^{\Delta} \alpha^{\Sigma,u}$. We conclude, using the rule 13, that $\m{M} \models^{\Delta,u} \alpha^{\Sigma}$. If $\Pi$ only contains rules of the lemma \ref{resolution}, then we can conclude $\m{M} \models^{\Theta} \beta^{\Omega}$. Then we can discharge the hypothesis because we know that any denotation of the existential may provide the same conclusion;

\setcounter{enumi}{19}
\item If $\m{M} \models^{\Delta,\circledcirc} \alpha^{\Sigma}$, then, by the rule 14, $\m{M} \models^{\Delta} \alpha^{\Sigma,\circledcirc}$. By the fact that $\alpha^{\Sigma,\circledcirc} \in \bs{F}_{n}$, the fitting relation and lemma \ref{countingLemma}, we know that $s(\Delta)$ is even. If we take some model $\m{H} = \ob \m{W} , \$ , \m{V} , z \cb$, such that $\m{M} \multimap_{s(\Delta)} \m{H}$ and $\m{H} \models \alpha^{\Sigma,\circledcirc}$, then $\exists M \in \$(z) : \ob \m{W} , \$ , \m{V} , z , M \cb \models \alpha^{\Sigma}$. Since the variable $N$ occurs nowhere else in the derivation, $N$ can be taken as a denotation of the given existential and we conclude that $\ob \m{W} , \$ , \m{V} , z , \sigma(N) \cb \models \alpha^{\Sigma}$, what means that $\m{H} \models \alpha^{\Sigma,N}$. So, by definition, $\alpha^{\Sigma,\circledcirc} \models_{\m{M}:s(\Delta)} \alpha^{\Sigma,N}$, which means, by lemma \ref{lemmaConsequence}, that $\m{M} \models^{\Delta} \alpha^{\Sigma,N}$. We conclude, using the rule 13, that $\m{M} \models^{\Delta,N} \alpha^{\Sigma}$. If $\Pi$ only contains rules of the lemma \ref{resolution}, then we can conclude $\m{M} \models^{\Theta} \beta^{\Omega}$. Then we can discharge the hypothesis because we know that any denotation of the existential may provide the same conclusion;

\setcounter{enumi}{26}
\item From rule 14, the fitting relation, and lemma \ref{countingLemma}, we know that $s(\Delta)$ is even. If we take some model $\m{H} = \ob \m{W} , \$ , \m{V} , z \cb$, such that $\m{M} \multimap_{s(\Delta)} \m{H}$, we know that the neighbourhoods of $\$(z)$ are in total order for the inclusion relation. Given any two neighbourhood variables $M$ and $N$, we know that $\sigma(M) \in \$(z)$, $\sigma(N) \in \$(z)$ and either $\sigma(M) \subset \sigma(N)$ or $\sigma(N) \subset \sigma(M)$. This can be expressed by $\m{H} \models (\shneg N)^{M} \vee (\shneg M)^{N}$. By definition, $\m{H} \models (\shneg N)^{M}$ or $\m{H} \models (\shneg M)^{N}$, then, by definition, $\m{M} \models^{\Delta} (\shneg N)^{M}$ or $\m{M} \models^{\Delta} (\shneg M)^{N}$ and, using rule 13, $\m{M} \models^{\Delta,M} \shneg N$ or $\m{M} \models^{\Delta,N} \shneg M$. If the subderivations $\Pi_{1}$ and $\Pi_{2}$ only contains the rules of lemma 8, then $\m{M} \models^{\Theta} \alpha^{\Sigma}$ and the hypothesis may be discharged.
\item Follow the same argument for rule 28.
\item From rule 14, the fitting relation, and lemma \ref{countingLemma}, we know that $s(\Delta)$ is even. If we take some model $\m{H} = \ob \m{W} , \$ , \m{V} , z \cb$, such that $\m{M} \multimap_{s(\Delta)} \m{H}$, we know that the neighbourhoods of $\$(z)$ are in total order for the inclusion relation. Given a neighbourhood variable $M$, we know that, for every neighbourhood variable $N$, either $\sigma(M) \subset \sigma(N)$ or $\sigma(N) \subset \sigma(M)$. This can be expressed by $\m{H} \models (\shneg N)^{M} \vee (\shpos N)^{M}$. By definition, $\m{H} \models (\shneg N)^{M}$ or $\m{H} \models (\shpos N)^{M}$, then, by definition, $\m{M} \models^{\Delta} (\shneg N)^{M}$ or $\m{M} \models^{\Delta} (\shpos N)^{M}$ and, using rule 13, $\m{M} \models^{\Delta,M} \shneg N$ or $\m{M} \models^{\Delta,M} \shpos N$. If the subderivations $\Pi_{1}$ and $\Pi_{2}$ only contains the rules of lemma 8, then $\m{M} \models^{\Theta} \alpha^{\Sigma}$ and the hypothesis may be discharged.
\end{enumerate}
\noindent Inductive case: for every rule, we suppose that the subderivations ($\Pi$) were only composed by rules of the lemma \ref{resolution}. If some derivation may contains all rules of the PUC-ND, then there must be an application of the rules of the present lemma that contains only the rules of the lemma \ref{resolution}, because the derivation is finite and the subderivations have a positive number of application of rules. Those cases are covered by the Base argument and, for that reason, they preserve the resolution relation. The next step is to consider all application of the rules of the present lemma that may have one application of the rules $5, 7, 11, 19, 20, 28, 29$ or $30$. Then, step by step, we cover all possible nested application of the rules of the present lemma.\end{proof}

\begin{definition} \label{derivability}
Given the formulas $\alpha^{\Sigma}$ and $\beta^{\Omega}$, the relation $\alpha^{\Sigma} \vdash^{\Delta}_{\Theta} \beta^{\Omega}$ of \textit{derivability} is defined iff there is a derivation that concludes $\beta^{\Omega}$ in the context $\Theta$ and that may only have $\alpha^{\Sigma}$ in the context $\Delta$ as open hypothesis. If $\Gamma \subset \bs{F}_{n}$ or $\Gamma \subset \bs{F}_{w}$, the relation $\Gamma \vdash^{\Delta}_{\Theta} \alpha^{\Sigma}$ of derivability is defined iff there is a derivation that concludes $\alpha^{\Sigma}$ in the context $\Theta$ and that only has as open hypothesis the formulas of $\Gamma$ in the context $\Delta$.
\end{definition}

\begin{definition}
$\alpha^{\Sigma}$ is a \textit{theorem} iff $\vdash \alpha^{\Sigma}$.
\end{definition}

\begin{theorem} \label{soundness}
$\Gamma \vdash \alpha^{\Sigma}$ implies $\Gamma \models \alpha^{\Sigma}$ (Soundness).
\end{theorem}

\begin{proof}
The fitting restriction of the rules of PUC-ND ensures that $\alpha^{\Sigma} \in \bs{F}_{n}$ because it appears in the empty context. The same conclusion follows for every formula of $\Gamma$. The derivability assures that there is a derivation that concludes $\alpha^{\Sigma}$ and takes as open hypothesis a subset of $\Gamma$, which we call $\Gamma'$. If we take a model $\m{M}$ that satisfies every formula of $\Gamma$, then it also satisfies every formula of $\Gamma'$. So, $\m{M} \models \gamma^{\Theta}$, for every $\gamma^{\Theta} \in \Gamma'$. But this means, by definition, that, for every wff of $\Gamma'$, the resolution relation holds with the empty context. Then, from lemma \ref{completeResolutionPUC-ND}, we know that $\m{M} \models \alpha^{\Sigma}$. So, every model, that satisfies every formula of $\Gamma$, satisfies $\alpha^{\Sigma}$ and, by definition, $\Gamma \models \alpha^{\Sigma}$.\end{proof}

In order to prove the converse implication, we use maximal consistent sets to prove completeness for the fragment $\{\wedge, \ra, \bullet, \circledcirc, \circledast\}$ of the language. The label $\circledcirc$ is not definable from $\circledast$ and vice-versa because the chosen logic for neighbourhoods is a free logic \cite{Lambert}. The reader can see the propositional classic logic case of this way of proving completeness in [17]. But for the completeness proof we must restrict the formulas to \textit{sentences} due to occurrences of variables.

\begin{definition}
Given $\alpha^{\Sigma} \in \bs{F}_{n}$, if $\alpha^{\Sigma}$ has no variables in the attributes of its subformulas nor any subformula of the shape $\shneg N$ or $\shpos N$, then $\alpha^{\Sigma} \in \bs{S}_{n}$. By analogy, we can construct $\bs{S}_{w}$ from $\bs{F}_{w}$.
\end{definition}

\begin{definition}
Given $\Gamma \subset \bs{S}_{n}$ ($\Gamma \subset \bs{S}_{w}$), we say that $\Gamma$ is \textit{n-inconsistent} (\textit{w-inconsistent}) if $\Gamma \vdash \bot_{n}$ ($\Gamma \vdash^{N}_{N} \bot_{w}$, where $N$ is a neighbourhood variable that does not occur in $\Gamma$) and \textit{n-consistent} (\textit{w-consistent}) if $\Gamma \not\vdash \bot_{n}$ ($\Gamma \not\vdash^{N}_{N} \bot_{w}$).
\end{definition}

\begin{lemma}
Given $\Gamma \subset \bs{S}_{n}$ ($\Gamma \subset \bs{S}_{w}$), the following three conditions are equivalents:
\begin{enumerate}[1.]
\item $\Gamma$ is n-inconsistent;
\item $\Gamma \vdash \phi^{\Theta}$, for any formula $\phi^{\Theta}$ that fits into the empty context;
\item There is at least a formula $\phi^{\Theta}$, such that $\Gamma \vdash \phi^{\Theta}$ and $\Gamma \vdash \phi^{\Theta} \ra \bot_{n}$
\end{enumerate}
\end{lemma}

\begin{proof} $1 \Rightarrow 2)$ If $\Gamma \vdash \bot_{n}$, then there is a derivation $\m{D}$ with conclusion $\bot_{n}$ and hypothesis in $\Gamma$. To $\m{D}$ we can add one inference using the rule 8 of PUC-ND to conclude any formula that fits into the empty context. $2 \Rightarrow 3)$ Trivial; $3 \Rightarrow 1)$ If $\Gamma \vdash \phi^{\Theta}$ and $\Gamma \vdash \phi^{\Theta} \ra \bot_{n}$, then there is a derivation for each formula with the hypothesis in $\Gamma$. Combining the derivations, we conclude $\bot_{n}$ using rule 12 of the PUC-ND. There is no problem with existential quantifiers in the context because we conclude the formulas in the empty context. So, $\Gamma \vdash \bot_{n}$. The same holds for $\Gamma \subset \bs{S}_{w}$.\end{proof}

\begin{lemma} \label{modelForSet}
Given $\Gamma \subset \bs{S}_{n}$ ($\Gamma \subset \bs{S}_{w}$), if there is a model (template) that satisfies every formula of $\Gamma$, then $\Gamma$ is n-consistent (w-consistent).
\end{lemma}

\begin{proof}
If $\Gamma \vdash \bot_{n}$, then, by theorem \ref{soundness}, $\Gamma \models \bot_{n}$. If there is model that satisfies every formula of $\Gamma$, then it also satisfies $\bot_{n}$ by the definition of logical consequence. But there is no model that satisfies $\bot_{n}$ because of the definition of the truth evaluation function. The same holds for $\Gamma \subset \bs{S}_{w}$.\end{proof}

\begin{lemma} \label{consistency}
Given $\Gamma \subset \bs{S}_{n}$: 1. If $\Gamma \cup \{\phi^{\Theta} \ra \bot_{n}\} \vdash \bot_{n}$, then $\Gamma \vdash \phi^{\Theta}$; 2. If $\Gamma \cup \{\phi^{\Theta}\} \vdash \bot_{n}$, then $\Gamma \vdash \phi^{\Theta} \ra \bot_{n}$. Likewise for $\Gamma \subset \bs{S}_{w}$.
\end{lemma}

\begin{proof}
The first (second) assumption implies that there is a derivation $\m{D}$ ($\m{D}'$) with hypothesis in $\Gamma \cup \{\phi^{\Theta} \ra \bot_{n}\}$ ($\Gamma \cup \{\phi^{\Theta}\}$) and conclusion $\bot_{n}$. Since $\neg (\phi^{\Theta}) \equiv \phi^{\Theta} \ra \bot_{n}$, we can apply the rule $\bot$-classical ($\ra$-introduction) and eliminate all occurrences of $\phi^{\Theta} \ra \bot_{n}$ ($\phi^{\Theta}$) as hypothesis, then we obtain a derivation with hypothesis in $\Gamma$ and conclusion $\phi^{\Theta}$ ($\phi^{\Theta} \ra \bot_{n}$). The same argument holds for $\Gamma \subset \bs{S}_{w}$.\end{proof}

\begin{lemma} \label{counting}
$\bs{S}_{n}$ and $\bs{S}_{w}$ are denumerable.
\end{lemma}

\begin{proof}
Every $\alpha^{\Sigma} \in \bs{S}_{n}$ contains a finite number of proposition symbols and logical operators. So, any lexical order provide a bijection from $\bs{S}_{n}$ to the natural numbers. The same argument works for $\bs{S}_{w}$.\end{proof}

\begin{definition}
$\Gamma \subset \bs{S}_{n}$ ($\Gamma \subset \bs{S}_{w}$) is \textit{maximally n-consistent} (\textit{maximally w-consistent}) iff $\Gamma$ is n-consistent (w-consistent) and it cannot be a proper subset of any other n-consistent (w-consistent) set.
\end{definition}

\begin{lemma} \label{subMax}
Every n-consistent (w-consistent) set is subset of a maximally n-consistent (w-consistent) set.
\end{lemma}
\begin{proof}
According to the lemma \ref{counting}, we may have a list $\varphi_{0}, \varphi_{1}, \ldots$ of all wff of $\bs{S}_{n}$. We build a non-decreasing sequence of sets $\Gamma_{i}$ such that the union is maximally n-consistent.\\
\noindent $\Gamma_{0} = \Gamma$;\\
\noindent $\Gamma_{k+1} = \Gamma_{k} \cup \{\varphi_{k}\}$ if n-consistent, $\Gamma_{k} $ otherwise;\\
\noindent $\hat{\Gamma} = \bigcup\{\Gamma_{k} \; | \; k \geq 0 \}$.\\
(a) $\Gamma_{k}$ is n-consistent for all $k$: by induction; (b) $\hat{\Gamma}$ is n-consistent: suppose that $\hat{\Gamma} \vdash \bot_{n}$, then for every derivation $\m{D}$ of $\bot_{n}$ with hypothesis in $\hat{\Gamma}$ we have a finite set of hypothesis. By definition, every wff is included in $\hat{\Gamma}$ via a set $\Gamma_{k}$. Then, because the sequence of construction of $\hat{\Gamma}$ is non-decreasing, there is a number $m$, such that $\Gamma_{m}$ contains all hypothesis of $\m{D}$. But $\Gamma_{m}$ is n-consistent and, therefore, cannot derive $\bot_{n}$. The same holds for w-consistent sets.\end{proof}

\begin{lemma} \label{closedDeriv}
If $\Gamma$ is maximally n-consistent (w-consistent) set, then $\Gamma$ is closed under derivability.
\end{lemma}
\begin{proof}
Suppose that $\Gamma \vdash \varphi^{\Theta}$ and $\varphi^{\Theta} \not\in \Gamma$. Then $\Gamma \cup \{\varphi^{\Theta}\}$ must be n-inconsistent by the definition of maximally n-consistent set. By lemma \ref{consistency}, $\Gamma \vdash \varphi^{\Theta} \ra \bot_{n}$, so $\Gamma$ is n-inconsistent. The same argument holds for w-consistent sets.\end{proof}

\begin{lemma} \label{dual}
If $\Gamma$ is maximally n-consistent (w-consistent), then:
\begin{itemize}
\item[(a)] For all $\varphi^{\Theta} \in \bs{S}_{n}$ ($\in \bs{S}_{w}$), either $\varphi^{\Theta} \in \Gamma$ or $\varphi^{\Theta} \ra \bot_{n} \in \Gamma$ ($\varphi^{\Theta} \ra \bot_{w}$);
\item[(b)] For all $\varphi^{\Theta},\psi^{\Upsilon} \in \bs{S}_{n}$ ($\in \bs{S}_{w}$), $\varphi^{\Theta} \ra \psi^{\Upsilon} \in \Gamma$ iff $\varphi^{\Theta} \in \Gamma$ implies $\psi^{\Upsilon} \in \Gamma$. 
\end{itemize}
\end{lemma}
\begin{proof}
(a) Both $\varphi^{\Theta}$ and $\varphi^{\Theta} \ra \bot_{n}$ cannot belong to $\Gamma$. If $\Gamma \cup \varphi^{\Theta}$ is n-consistent, then, by the definition of maximally n-consistent set, $\varphi^{\Theta} \in \Gamma$. If it is n-inconsistent, then by lemmas \ref{consistency} and \ref{closedDeriv}, $\varphi^{\Theta} \ra \bot_{n} \in \Gamma$.
(b) If $\varphi^{\Theta} \ra \psi^{\Upsilon} \in \Gamma$ and $\varphi^{\Theta} \in \Gamma$, then $ \Gamma \vdash \psi^{\Upsilon}$ by $\ra$-elimination and, by lemma \ref{closedDeriv}, $\psi^{\Upsilon} \in \Gamma$. In other way, supposing that $\varphi^{\Theta} \in \Gamma$ implies $\psi^{\Upsilon} \in \Gamma$, if $\varphi^{\Theta} \in \Gamma$, then obviously $\Gamma \vdash \psi^{\Upsilon}$ and $\Gamma \vdash \varphi^{\Theta} \ra \psi^{\Upsilon}$ by $\ra$-introduction. If $\varphi^{\Theta} \not\in \Gamma$, then, by the (a) conclusion, $\varphi^{\Theta} \ra \bot_{n} \in \Gamma$. The conclusion $\varphi^{\Theta} \ra \psi^{\Upsilon} \in \Gamma$ comes from a simple derivation with $\varphi^{\Theta}$ as a discharged hypothesis of a $\ra$-introduction that follows an application of the intuitionistic absurd. The same argument holds for w-consistent sets.\end{proof}

\begin{corollary} \label{corIff}
If $\Gamma$ is maximally n-consistent (w-consistent), then $\varphi^{\Theta} \in \Gamma$ iff $\varphi^{\Theta} \ra \bot_{n} \not\in \Gamma$.
\end{corollary}

\begin{definition}
Given the maximally n-consistent set $\Gamma \subset \bs{S}_{n}$ and the maximally w-consistent set $\Lambda \subset \bs{S}_{w}$, we say that $\Gamma$ \textit{accepts} $\Lambda$ ($\Gamma \propto \Lambda$) if $\alpha^{\Sigma} \in \Lambda$ implies $\alpha^{\Sigma,\circledcirc} \in \Gamma$. If $\alpha^{\Sigma} \in \Gamma$ implies $\alpha^{\Sigma,\bullet} \in \Lambda$, then $\Lambda \propto \Gamma$.
\end{definition}

\begin{definition}
Given maximally w-consistent sets $\Gamma$ and $\Lambda$, we say that $\Gamma$ \textit{subordinates} $\Lambda$ ($\Lambda \sqsubset \Gamma$) iff $\alpha^{\Sigma,\bullet} \in \Lambda$ implies $\alpha^{\Sigma,\bullet} \in \Gamma$ and $\alpha^{\Sigma,\ast} \in \Gamma$ implies $\alpha^{\Sigma,\ast} \in \Lambda$.
\end{definition}

\begin{lemma} \label{consistentModel}
If $\Gamma$ is n-consistent, then there is a model $\m{M}$, such that $\m{M} \models \alpha^{\Sigma}$, for every $\alpha^{\Sigma} \in \Gamma$.
\end{lemma}

\begin{proof}
By lemma \ref{subMax}, $\Gamma$ is contained in a maximally n-consistent set $\hat{\Gamma}$. We consider every maximally n-consistent set $\Psi$ as a representation of one world, denoted by $\chi_{\Psi}$. Every maximally w-consistent set will be seen as a set of worlds that may be a neighbourhood. We take the set of maximally n-consistent sets as $\m{W}$. We take $\propto$ as the nested neighbourhood function $\$$ and $\sqsubset$ as the total order among neighbourhoods. To build the truth evaluation function $\m{V}$, we require, for every maximally n-consistent set $\Psi$ and for every $\alpha$ atomic: (a) $\chi_{\Psi} \in \m{V}(\alpha)$ if $\alpha \in \Psi$; (b) $\chi_{\Psi} \not\in \m{V}(\alpha)$ if $\alpha \not\in \Psi$. If we take $\m{M} = \ob \m{W} , \$ , \m{V} , \chi_{\hat{\Gamma}} \cb$, then, for every $\alpha^{\Sigma} \in \hat{\Gamma}$, $\m{M} \models \alpha^{\Sigma}$. We proceed by induction on the structure of $\alpha^{\Sigma}$:\\
\noindent (Base) If $\alpha^{\Sigma}$ is atomic, $\m{M} \models \alpha^{\Sigma}$ iff $\alpha^{\Sigma} \in \hat{\Gamma}$, by the definition of $\m{V}$;
\begin{itemize}
\item $\alpha^{\Sigma} = \beta^{\Omega} \wedge \gamma^{\Theta}$. $\m{M} \models \alpha^{\Sigma}$ iff $\m{M} \models \beta^{\Omega}$ and $\m{M} \models \gamma^{\Theta}$ iff (induction hypothesis) $\beta^{\Omega} \in \hat{\Gamma}$ and $\gamma^{\Theta} \in \hat{\Gamma}$. We conclude that $\alpha^{\Sigma} \in \hat{\Gamma}$ by lemma \ref{closedDeriv}. Conversely $\alpha^{\Sigma} \in \hat{\Gamma}$ iff $\beta^{\Omega} \in \hat{\Gamma}$ and $\gamma^{\Theta} \in \hat{\Gamma}$ by lemma \ref{closedDeriv} and the rest follows by the induction hypothesis;
\item $\alpha^{\Sigma} = \beta^{\Omega} \ra \gamma^{\Theta}$. $\m{M} \not\models \alpha^{\Sigma}$ iff $\m{M} \models \beta^{\Omega}$ and $\m{M} \not\models \gamma^{\Theta}$ iff (induction hypothesis) $\beta^{\Omega} \in \hat{\Gamma}$ and $\gamma^{\Theta} \not\in \hat{\Gamma}$ iff $\beta^{\Omega} \ra \gamma^{\Theta} \not\in \hat{\Gamma}$ by lemma \ref{dual};
\item $\alpha^{\Sigma} = \beta^{\Omega,\circledast}$. If there is no maximally w-consistent set $\Upsilon$, such that $\hat{\Gamma} \propto \Upsilon$, then $\$(\chi)$ is empty and for every $\beta^{\Omega} \in \bs{F}_{w}$, $\m{M} \models \beta^{\Omega,\circledast}$. This case occurs iff there is no wff of the form $\sigma^{\Phi,\circledcirc}$ in $\hat{\Gamma}$. If there is some maximally w-consistent set accepted by $\hat{\Gamma}$, then $\m{M} \models \beta^{\Omega,\circledast}$ iff, for every maximally w-consistent set $\Upsilon$, such that $\hat{\Gamma} \propto \Upsilon$, $\beta^{\Omega} \in \Upsilon$ iff $(\beta^{\Omega} \ra \bot_{w})^{\circledcirc} \ra \bot_{n} \in \hat{\Gamma}$ which is verified by the other cases;
\item $\alpha^{\Sigma} = \beta^{\Omega,\circledcirc}$. We build a set $\Upsilon \subset \bs{F}_{w}$, starting by $\beta^{\Omega} \in \Upsilon$. We take a sequence $\varphi_{i}$ of all wff with the shape of $(\beta^{\Omega} \wedge \gamma^{\Theta})^\circledcirc$ in $\hat{\Gamma}$. If, for $\varphi_{i} = (\beta^{\Omega} \wedge \gamma^{\Theta})^\circledcirc$, $\Upsilon \cup \{ \gamma^{\Theta} \}$ is w-consistent, then $\gamma^{\Theta} \in \Upsilon$. To demonstrate that $\Upsilon$ is maximally w-consistent, we suppose that there is a wff $\sigma^{\Phi} \in \bs{F}_{w}$, such that $\sigma^{\Phi} \not\in \Upsilon$ and $\Upsilon \cup \{ \sigma^{\Phi} \}$ is w-consistent. Then $(\beta^{\Omega} \wedge \sigma^{\Phi})^\circledcirc \not\in \hat{\Gamma}$ by the definition of $\Upsilon$ and, by lemma \ref{dual}, $(\beta^{\Omega} \wedge \sigma^{\Phi})^\circledcirc \ra \bot_{n} \in \hat{\Gamma}$. But from $\beta^{\Omega,\circledcirc} \in \hat{\Gamma}$ and $(\beta^{\Omega} \wedge \sigma^{\Phi})^\circledcirc \ra \bot_{n} \in \hat{\Gamma}$ we know that $(\beta^{\Omega} \wedge ( \sigma^{\Phi} \ra \bot_{w} ) )^\circledcirc \in \hat{\Gamma}$, using lemma \ref{closedDeriv} and the following derivation:
\noindent\begin{center}
\AxiomC{$\beta^{\Omega,\circledcirc}$}
\UnaryInfC{$\beta^{\Omega,\circledcirc}$} \RightLabel{$\circledcirc$}
\UnaryInfC{$\beta^{\Omega}$}
\AxiomC{$\beta^{\Omega,\circledcirc}$}
\UnaryInfC{$\beta^{\Omega,\circledcirc}$} \RightLabel{$\circledcirc$}
\UnaryInfC{$\beta^{\Omega}$}
\AxiomC{$^{1}$[$\beta^{\Omega}$]} \RightLabel{$N$}
\UnaryInfC{$\beta^{\Omega}$}
\AxiomC{$\Pi$} \RightLabel{$N$}
\UnaryInfC{$\sigma^{\Phi} \ra \bot_{w}$} \RightLabel{$N$}
\BinaryInfC{$\beta^{\Omega} \wedge ( \sigma^{\Phi} \ra \bot_{w} )$} \RightLabel{$\circledcirc$}
\BinaryInfC{$\beta^{\Omega} \wedge ( \sigma^{\Phi} \ra \bot_{w} )$}
\UnaryInfC{$(\beta^{\Omega} \wedge ( \sigma^{\Phi} \ra \bot_{w} ) )^\circledcirc$} \LeftLabel{$1$}
\BinaryInfC{$(\beta^{\Omega} \wedge ( \sigma^{\Phi} \ra \bot_{w} ) )^\circledcirc$}
\alwaysNoLine
\UnaryInfC{$\phantom{.}$} \DisplayProof

\AxiomC{$\beta^{\Omega,\circledcirc}$}
\UnaryInfC{$\beta^{\Omega,\circledcirc}$}
\AxiomC{$\beta^{\Omega}$} \RightLabel{$N$}
\UnaryInfC{$\beta^{\Omega}$}
\AxiomC{$^{2}$[$\sigma^{\Phi}$]} \RightLabel{$N$}
\UnaryInfC{$\sigma^{\Phi}$} \RightLabel{$N$}
\BinaryInfC{$\beta^{\Omega} \wedge \sigma^{\Phi}$} \RightLabel{$N$}
\UnaryInfC{$\beta^{\Omega} \wedge \sigma^{\Phi}$} \RightLabel{$\circledcirc$}
\BinaryInfC{$\beta^{\Omega} \wedge \sigma^{\Phi}$}
\UnaryInfC{$(\beta^{\Omega} \wedge \sigma^{\Phi})^\circledcirc$}
\AxiomC{$(\beta^{\Omega} \wedge \sigma^{\Phi})^\circledcirc \ra \bot_{n}$}
\UnaryInfC{$(\beta^{\Omega} \wedge \sigma^{\Phi})^\circledcirc \ra \bot_{n}$}
\BinaryInfC{$\bot_{n}$} 
\UnaryInfC{$\bot_{w}^{N}$} \RightLabel{$N$}
\UnaryInfC{$\bot_{w}$} \LeftLabel{$\bs{\Pi}$ \hspace*{2cm} $2$} \RightLabel{$N$}
\UnaryInfC{$\sigma^{\Phi} \ra \bot_{w}$} \DisplayProof\end{center}

\noindent So, by definition, $\sigma^{\Phi} \ra \bot_{w} \in \Upsilon$ and $\Upsilon \cup \{\sigma^{\Phi}\}$ cannot be w-consistent. We conclude that $\Upsilon$ is maximally w-consistent and $\hat{\Gamma} \propto \Upsilon$. $\Upsilon$ represents a neighbourhood $N_{\Upsilon} \in \$(\chi_{\hat{\Gamma}})$. To prove that $\m{M} \models \beta^{\Omega,\circledcirc}$, we need to prove that $\m{T} = \ob \m{W} , \$ , \m{V} , \chi_{\hat{\Gamma}} , N_{\Upsilon} \cb \models \beta^{\Omega}$. We proceed by induction on the structure of $\beta^{\Omega}$:
\begin{itemize}
\item $\beta^{\Omega} = \varphi^{\Lambda} \wedge \gamma^{\Theta}$. $\m{T} \models \beta^{\Omega}$ iff $\m{T} \models \varphi^{\Lambda}$ and $\m{T} \models \gamma^{\Theta}$ iff (induction hypothesis) $\varphi^{\Lambda} \in \Upsilon$ and $\gamma^{\Theta} \in \Upsilon$. We conclude that $\beta^{\Omega} \in \Upsilon$ by lemma \ref{closedDeriv}. Conversely $\beta^{\Omega} \in \hat{\Gamma}$ iff $\varphi^{\Lambda} \in \Upsilon$ and $\gamma^{\Theta} \in \Upsilon$ by lemma \ref{closedDeriv} and the rest follows by the induction hypothesis;

\item $\beta^{\Omega} = \varphi^{\Lambda} \ra \gamma^{\Theta}$. $\m{T} \not\models \varphi^{\Lambda}$ iff $\m{T} \models \varphi^{\Lambda}$ and $\m{T} \not\models \gamma^{\Theta}$ iff (induction hypothesis) $\varphi^{\Lambda} \in \hat{\Gamma}$ and $\gamma^{\Theta} \not\in \Upsilon$ iff $\varphi^{\Lambda} \ra \gamma^{\Theta} \not\in \Upsilon$ by lemma \ref{dual};

\item $\beta^{\Omega} = \varphi^{\Lambda,\bullet}$. We build a set $\Psi$, starting by $\varphi^{\Lambda} \in \Psi$. We take a sequence $\varphi_{i}$ in $\Upsilon$ that have the form $(\varphi^{\Lambda} \wedge \gamma^{\Theta})^\bullet$. If, for $\varphi_{i} = (\varphi^{\Lambda} \wedge \gamma^{\Theta})^\bullet$, $\Psi \cup \{ \gamma^{\Theta} \}$ is n-consistent, then $\gamma^{\Theta} \in \Upsilon$. To demonstrate that $\Psi$ is maximally n-consistent, we suppose that there is a wff $\sigma^{\Phi}$, such that $\sigma^{\Phi} \not\in \Psi$ and $\Psi \cup \{ \sigma^{\Phi} \}$ is n-consistent. Then $(\varphi^{\Lambda} \wedge \sigma^{\Phi})^\bullet \not\in \hat{\Gamma}$ by the definition of $\Psi$ and, by lemma \ref{dual}, $(\varphi^{\Lambda} \wedge \sigma^{\Phi})^\bullet \ra \bot_{w} \in \Upsilon$. But from $\varphi^{\Lambda,\bullet} \in \Upsilon$ and $(\varphi^{\Lambda} \wedge \sigma^{\Phi})^\bullet \ra \bot_{w} \in \Upsilon$ we know that $(\varphi^{\Lambda} \wedge ( \sigma^{\Phi} \ra \bot_{n} ) )^\bullet \in \Upsilon$, using lemma \ref{closedDeriv} and the following derivation:
\noindent\begin{center}
\AxiomC{$\beta^{\Omega,\bullet}$}
\UnaryInfC{$\beta^{\Omega,\bullet}$} \RightLabel{$\bullet$}
\UnaryInfC{$\beta^{\Omega}$}
\AxiomC{$^{1}$[$\beta^{\Omega}$]} \RightLabel{$u$}
\UnaryInfC{$\beta^{\Omega}$}
\AxiomC{$\Pi$} \RightLabel{$u$}
\UnaryInfC{$\sigma^{\Phi} \ra \bot_{n}$} \RightLabel{$u$}
\BinaryInfC{$\beta^{\Omega} \wedge ( \sigma^{\Phi} \ra \bot_{n} )$} \RightLabel{$\bullet$}
\UnaryInfC{$\beta^{\Omega} \wedge ( \sigma^{\Phi} \ra \bot_{n} )$}
\UnaryInfC{$(\beta^{\Omega} \wedge ( \sigma^{\Phi} \ra \bot_{n} ) )^\bullet$} \LeftLabel{$1$}
\BinaryInfC{$(\beta^{\Omega} \wedge ( \sigma^{\Phi} \ra \bot_{n} ) )^\bullet$}
\alwaysNoLine
\UnaryInfC{$\phantom{.}$} \DisplayProof

\AxiomC{$\beta^{\Omega,\bullet}$}
\UnaryInfC{$\beta^{\Omega,\bullet}$}
\AxiomC{$^{1}$[$\beta^{\Omega}$]} \RightLabel{$u$}
\UnaryInfC{$\beta^{\Omega}$}
\AxiomC{$^{2}$[$\sigma^{\Phi}$]} \RightLabel{$u$}
\UnaryInfC{$\sigma^{\Phi}$} \RightLabel{$u$}
\BinaryInfC{$\beta^{\Omega} \wedge \sigma^{\Phi}$} \RightLabel{$u$}
\UnaryInfC{$\beta^{\Omega} \wedge \sigma^{\Phi}$} \RightLabel{$\bullet$}
\BinaryInfC{$\beta^{\Omega} \wedge \sigma^{\Phi}$}
\UnaryInfC{$(\beta^{\Omega} \wedge \sigma^{\Phi})^\bullet$}
\AxiomC{$(\beta^{\Omega} \wedge \sigma^{\Phi})^\bullet \ra \bot_{w}$}
\UnaryInfC{$(\beta^{\Omega} \wedge \sigma^{\Phi})^\bullet \ra \bot_{w}$}
\BinaryInfC{$\bot_{w}$} 
\UnaryInfC{$\bot_{n}^{u}$} \RightLabel{$u$}
\UnaryInfC{$\bot_{n}$} \LeftLabel{$\bs{\Pi}$ \hspace*{2cm} $2$} \RightLabel{$u$}
\UnaryInfC{$\sigma^{\Phi} \ra \bot_{n}$} \DisplayProof\end{center}

\noindent So, by definition, $\sigma^{\Phi} \ra \bot_{n} \in \Psi$ and $\Psi \cup \{\sigma^{\Phi}\}$ can not be n-consistent. We conclude that $\Psi$ is maximally n-consistent and $\Upsilon \propto \Psi$. $\Psi$ represents a world $\chi_{\Psi} \in N_{\Upsilon}$. To prove that $\m{T} \models \varphi^{\Lambda,\bullet}$, we need to prove that $\ob \m{W} , \$ , \m{V} , \chi_{\Psi} \cb \models \varphi^{\Lambda}$ using the previous cases.
\end{itemize}\end{itemize}\end{proof}

\begin{corollary} \label{trick}
$\Gamma \not\vdash \alpha^{\Sigma}$ iff there is a model $\m{M}$, such that $\m{M} \models \phi^{\Theta}$, for every $\phi^{\Theta} \in \Gamma$, and $\m{M} \not\models \alpha^{\Sigma}$.
\end{corollary}
\begin{proof}
$\Gamma \not\vdash \alpha^{\Sigma}$ iff $\Gamma \cup \{\alpha^{\Sigma} \ra \bot_{n}\}$ is n-consistent by lemma \ref{consistency} and the definition of n-consistent set. By lemmas \ref{modelForSet} and \ref{consistentModel}, $\Gamma \cup \{\alpha^{\Sigma} \ra \bot_{n}\}$ is n-consistent iff there is a model $\m{M}$, such that $\m{M} \models \phi^{\Theta}$, for every $\phi^{\Theta} \in \Gamma \cup \{\alpha^{\Sigma} \ra \bot_{n}\}$. It means that $\m{M}$ satisfies every formula of $\Gamma$ and $\m{M} \not\models \alpha^{\Sigma}$.\end{proof}

\begin{theorem}
$\Gamma \models \alpha^{\Sigma}$ implies $\Gamma \vdash \alpha^{\Sigma}$ (Completeness).
\end{theorem}

\begin{proof}
$\Gamma \not\vdash \alpha^{\Sigma}$ implies $\Gamma \not\models \alpha^{\Sigma}$, by the corollary \ref{trick} and the definition of logical consequence.\end{proof}

\section{Normalization, Decidability, Complexity}

We investigate here the normalization of PUC-ND. For the normalization proof, we want to present first the approach similar to the classical propositional normalization. This case happens for maximum formulas in derivations with fixed contexts, since the contexts are not defined for propositional logic.

To do so, we investigate a fragment of the presented language, in order to use the Prawitz \cite{Prawitz} strategy for propositional logic normalization, in which he restricted the applications of the classical absurd to atomic formulas. In the chosen fragment $\m{L}_{-}$ we only omit the operator $\vee$, which may be recovered by the definition $\alpha \vee \beta \equiv \neg \alpha \ra \beta$. After that result, we present the reductions for the remaining rules.

In every case we follow the van Dalen algorithm for normalizing a derivation, starting form a subderivation that concludes a maximum formula with maximum rank, what means a maximum formula that has no maximum formula above it with more connectives in the subderivation.

\begin{lemma} \label{lemmaNormProp0}
Every derivation that is composed only by the rules 1 to 8 and 10 to 12 is normalizable.
\end{lemma}

\begin{proof}
These rules may be seen as a natural deduction system for the classical propositional logic, since the context is fixed and the formulas with labels are treated like atomic formulas. We follow the strategy of Prawitz \cite{Prawitz}. We give here the reductions for the propositional logical operators, in the case of fixed context and labels:
\begin{itemize}
\item $\wedge$-reductions:\\[5pt]
\noindent\begin{tabular}{ccccccc}
\AxiomC{$\Pi_{1}$} \RightLabel{$\Delta$}
\UnaryInfC{$\alpha$} \RightLabel{$\Delta$}
\AxiomC{$\Pi_{2}$} \RightLabel{$\Delta$}
\UnaryInfC{$\beta$} \RightLabel{$\Delta$}
\BinaryInfC{$\alpha \wedge \beta$} \RightLabel{$\Delta$}
\UnaryInfC{$\alpha$} \RightLabel{$\Delta$}
\UnaryInfC{$\Pi_{3}$} \DisplayProof
& $\rhd$ &
\AxiomC{$\Pi_{1}$} \RightLabel{$\Delta$}
\UnaryInfC{$\alpha$} \RightLabel{$\Delta$}
\UnaryInfC{$\Pi_{3}$} \DisplayProof & &
\AxiomC{$\Pi_{1}$} \RightLabel{$\Delta$}
\UnaryInfC{$\alpha$} \RightLabel{$\Delta$}
\AxiomC{$\Pi_{2}$} \RightLabel{$\Delta$}
\UnaryInfC{$\beta$} \RightLabel{$\Delta$}
\BinaryInfC{$\alpha \wedge \beta$} \RightLabel{$\Delta$}
\UnaryInfC{$\beta$} \RightLabel{$\Delta$}
\UnaryInfC{$\Pi_{3}$} \DisplayProof
& $\rhd$ &
\AxiomC{$\Pi_{2}$} \RightLabel{$\Delta$}
\UnaryInfC{$\beta$} \RightLabel{$\Delta$}
\UnaryInfC{$\Pi_{3}$} \DisplayProof \\
\end{tabular}

\item $\ra$-reduction:\\[5pt]
\noindent\begin{tabular}{ccc}
\AxiomC{$\Pi_{1}$} \RightLabel{$\Delta$}
\UnaryInfC{$\alpha$} \RightLabel{$\Delta$}
\AxiomC{[$\alpha$]} \RightLabel{$\Delta$}
\UnaryInfC{$\Pi_{2}$} \RightLabel{$\Delta$}
\UnaryInfC{$\beta$} \RightLabel{$\Delta$}
\UnaryInfC{$\alpha \ra \beta$} \RightLabel{$\Delta$}
\BinaryInfC{$\beta$} \RightLabel{$\Delta$}
\UnaryInfC{$\Pi_{3}$} \RightLabel{$\Delta$} \DisplayProof
& $\rhd$ &
\AxiomC{$\Pi_{1}$} \RightLabel{$\Delta$}
\UnaryInfC{$\alpha$} \RightLabel{$\Delta$}
\UnaryInfC{$\Pi_{2}$} \RightLabel{$\Delta$}
\UnaryInfC{$\beta$} \RightLabel{$\Delta$}
\UnaryInfC{$\Pi_{3}$} \RightLabel{$\Delta$} \DisplayProof
\end{tabular}
\end{itemize}

The application of the classical absurd may be restricted to atomic formulas only. We change the following derivation according to the principal logical operator of $\gamma$. We only present the change procedure for $\wedge$, see \cite{Prawitz} for further details.\\\hfill\\
\begin{tabular}{ccc}
\AxiomC{[$\neg \gamma$]} \RightLabel{$\Delta$}
\UnaryInfC{$\Pi_{1}$} \RightLabel{$\Delta$}
\UnaryInfC{$\bot$} \RightLabel{$\Delta$}
\UnaryInfC{$\gamma$} \RightLabel{$\Delta$}
\UnaryInfC{$\Pi_{2}$} \RightLabel{$\Delta$} \DisplayProof & \quad \quad &
\AxiomC{$^{1}$[$\alpha \wedge \beta$]} \RightLabel{$\Delta$} \RightLabel{$\Delta$}
\UnaryInfC{$\alpha$} \RightLabel{$\Delta$} \RightLabel{$\Delta$}
\AxiomC{$^{2}$[$\neg \alpha$]} \RightLabel{$\Delta$} \RightLabel{$\Delta$}
\BinaryInfC{$\bot$} \RightLabel{$\Delta$} \LeftLabel{\small{1}} \RightLabel{$\Delta$}
\UnaryInfC{[$\neg (\alpha \wedge \beta)$]} \RightLabel{$\Delta$}
\UnaryInfC{$\Pi_{1}$} \RightLabel{$\Delta$} \RightLabel{$\Delta$}
\UnaryInfC{$\bot$} \RightLabel{$\Delta$} \LeftLabel{\small{2}} \RightLabel{$\Delta$}
\UnaryInfC{$\alpha$} \RightLabel{$\Delta$} \RightLabel{$\Delta$}
\AxiomC{$^{3}$[$\alpha \wedge \beta$]} \RightLabel{$\Delta$} \RightLabel{$\Delta$}
\UnaryInfC{$\alpha$} \RightLabel{$\Delta$} \RightLabel{$\Delta$}
\AxiomC{$^{4}$[$\neg \beta$]} \RightLabel{$\Delta$} \RightLabel{$\Delta$}
\BinaryInfC{$\bot$} \RightLabel{$\Delta$} \LeftLabel{\small{3}} \RightLabel{$\Delta$}
\UnaryInfC{[$\neg (\alpha \wedge \beta)$]} \RightLabel{$\Delta$}
\UnaryInfC{$\Pi_{1}$} \RightLabel{$\Delta$} \RightLabel{$\Delta$}
\UnaryInfC{$\bot$} \RightLabel{$\Delta$} \LeftLabel{\small{4}} \RightLabel{$\Delta$}
\UnaryInfC{$\beta$} \RightLabel{$\Delta$} \RightLabel{$\Delta$}
\BinaryInfC{$\alpha \wedge \beta$} \RightLabel{$\Delta$}
\UnaryInfC{$\Pi_{2}$} \RightLabel{$\Delta$} \DisplayProof
\end{tabular}\end{proof}

\begin{lemma}
Given a derivation $\Pi$, if we exchange every occurence of a world variable $u$ in $\Pi$ by a world variable $w$ that does occurs in $\Pi$, then the resulting derivation, which we represent by $\Pi(u\mid w)$, is also a derivation.
\end{lemma}

\begin{proof}
By induction.
\end{proof}

\begin{theorem} \label{normalization}
Every derivation is normalizable.
\end{theorem}

\begin{proof}
We present the argument for the introduction of the remaining rules. The introduction of the rule 9 cannot produce maximum formulae, but it may produce detours, considering the rules 7 and 8, if the considered subderivation ($\Pi_{2}$ below) do not discharge any hypothesis of the upper subderivation ($\Pi_{1}$ below). But such detours may be substituted by one application of the rule 8 as shown below:\begin{center}\begin{tabular}{lcr}
\AxiomC{$\Pi_{1}$} \RightLabel{$\Delta$}
\UnaryInfC{$\bot$} \LeftLabel{rule 9:}
\UnaryInfC{$\bot_{n}$}
\UnaryInfC{$\Pi_{2}$} \RightLabel{$\Delta$}
\UnaryInfC{$\bot$} \LeftLabel{rule 8:} \RightLabel{$\Delta$}
\UnaryInfC{$\beta^{\Omega}$} \RightLabel{$\Delta$}
\UnaryInfC{$\Pi_{3}$} \DisplayProof & $\rhd$ &

\AxiomC{$\Pi_{2}$} \RightLabel{$\Delta$}
\UnaryInfC{$\Pi_{2}$} \RightLabel{$\Delta$}
\UnaryInfC{$\bot$} \LeftLabel{rule 8:} \RightLabel{$\Delta$}
\UnaryInfC{$\beta^{\Omega}$} \RightLabel{$\Delta$}
\UnaryInfC{$\Pi_{3}$} \DisplayProof\\&&\\
\AxiomC{[$\neg (\beta^{\Omega})$]} \RightLabel{$\Delta$}
\UnaryInfC{$\neg (\beta^{\Omega})$}
\AxiomC{$\Pi_{1}$} \RightLabel{$\Delta$}
\UnaryInfC{$\bot$} \LeftLabel{rule 9:}
\UnaryInfC{$\bot_{n}$}
\BinaryInfC{$\Pi_{2}$} \RightLabel{$\Delta$}
\UnaryInfC{$\bot$} \LeftLabel{rule 7:} \RightLabel{$\Delta$}
\UnaryInfC{$\beta^{\Omega}$} \RightLabel{$\Delta$}
\UnaryInfC{$\Pi_{3}$} \DisplayProof & $\rhd$ &
\AxiomC{$\Pi_{1}$} \RightLabel{$\Delta$}
\UnaryInfC{$\bot$} \LeftLabel{rule 8:} \RightLabel{$\Delta$}
\UnaryInfC{$\beta^{\Omega}$} \RightLabel{$\Delta$}
\UnaryInfC{$\Pi_{3}$} \DisplayProof
\end{tabular}\end{center}

The rules 13 and 14 produce a detour only if the conclusion of one is taken as an hypothesis of the other rule for the same context and, as above, the considered subderivation do not discharge any hypothesis of the upper subderivation. In this case, if we eliminate such detour, as below, we may produce a new maximum formula of the case of lemma \ref{lemmaNormProp0}. We cannot produce new detours by doing that elimination because, if there is any detour surrounding the formula $\alpha^{\Sigma}$, it must exist before the elimination. If we start from the up and left most detour, we eliminate the detours until we produce a derivation that contains only maximum formulas of the case of lemma \ref{lemmaNormProp0}. The same argument works for the rules 15 and 16 and to the rules 21 and 22.\begin{center}\begin{tabular}{ccccccc}
\AxiomC{$\Pi_{1}$} \RightLabel{$\Delta$}
\UnaryInfC{$\alpha^{\Sigma,\phi}$} \LeftLabel{rule 13:} \RightLabel{$\Delta,\phi$}
\UnaryInfC{$\alpha^{\Sigma}$}
\UnaryInfC{$\Pi_{2}$} \RightLabel{$\Delta,\phi$}
\UnaryInfC{$\alpha^{\Sigma}$} \LeftLabel{rule 14:} \RightLabel{$\Delta$}
\UnaryInfC{$\alpha^{\Sigma,\phi}$} \RightLabel{$\Delta$}
\UnaryInfC{$\Pi_{3}$} \DisplayProof & $\rhd$ &
\AxiomC{$\Pi_{1}$} \RightLabel{$\Delta$}
\UnaryInfC{$\alpha^{\Sigma,\phi}$} \RightLabel{$\Delta$}
\UnaryInfC{$\Pi_{3}$} \DisplayProof & \quad \quad &
\AxiomC{$\Pi_{1}$} \RightLabel{$\Delta,\phi$}
\UnaryInfC{$\alpha^{\Sigma}$} \LeftLabel{rule 14:} \RightLabel{$\Delta$}
\UnaryInfC{$\alpha^{\Sigma,\phi}$} \RightLabel{$\Delta$}
\UnaryInfC{$\Pi_{2}$} \RightLabel{$\Delta$}
\UnaryInfC{$\alpha^{\Sigma,\phi}$} \LeftLabel{rule 13:} \RightLabel{$\Delta,\phi$}
\UnaryInfC{$\alpha^{\Sigma}$} 
\UnaryInfC{$\Pi_{3}$} \DisplayProof
 & $\rhd$ &
\AxiomC{$\Pi_{1}$} \RightLabel{$\Delta,\phi$}
\UnaryInfC{$\alpha^{\Sigma}$} \RightLabel{$\Delta,\phi$}
\UnaryInfC{$\Pi_{3}$} \DisplayProof 
\end{tabular}\end{center}\begin{center}\begin{tabular}{ccc}
\AxiomC{$\Pi_{1}$} \RightLabel{$\Delta,N$}
\UnaryInfC{$\alpha^{\Sigma}$} \LeftLabel{rule 21:} \RightLabel{$\Delta,\circledast$}
\UnaryInfC{$\alpha^{\Sigma}$}
\AxiomC{$\phantom{-}$} \RightLabel{$\Delta,N$}
\UnaryInfC{$\beta^{\Omega}$} \LeftLabel{rule 22:} \RightLabel{$\Delta,N$}
\BinaryInfC{$\alpha^{\Sigma}$} 
\alwaysNoLine
\UnaryInfC{$\phantom{.}$} \DisplayProof & $\rhd$ &
\AxiomC{$\Pi_{1}$} \RightLabel{$\Delta,N$}
\UnaryInfC{$\alpha^{\Sigma}$} \DisplayProof\\&&\\
\AxiomC{$\Pi_{1}$} \RightLabel{$\Delta,\circledast$}
\UnaryInfC{$\alpha^{\Sigma}$}
\AxiomC{$\phantom{-}$} \RightLabel{$\Delta,N$}
\UnaryInfC{$\beta^{\Omega}$} \LeftLabel{rule 22:} \RightLabel{$\Delta,N$}
\BinaryInfC{$\alpha^{\Sigma}$} \LeftLabel{rule 21:} \RightLabel{$\Delta,\circledast$}
\UnaryInfC{$\alpha^{\Sigma}$} \DisplayProof & $\rhd$ &
\AxiomC{$\Pi_{1}$} \RightLabel{$\Delta,\circledast$}
\UnaryInfC{$\alpha^{\Sigma}$} \DisplayProof
\end{tabular}\end{center}

The introduction of the rules 17 and 19 preserves normalization. These rules produce a detour only if the conclusion of one is taken as an hypothesis of the other rule for the same context. In this case, if we eliminate such detour, as below, we may produce a new maximum formula of the case of lemma \ref{lemmaNormProp0}. We cannot produce new detours by doing that elimination because, if there is any detour surrounding the formula $\alpha^{\Sigma}$, it must exist before the elimination. If we start from the up and left most detour, we eliminate the detours until we produce a derivation that contains only maximum formulas of the case of lemma \ref{lemmaNormProp0}. We used the representation $(u , v \mid w,u)$ for the substitution of all occurrences of the variable $u$ by the variable $w$, that do not occur in $\Pi_{2}$, $\Theta$ or $\beta^{\Omega}$, and the subsequent substitution of all occurrences of the variable $v$ by the variable $u$. The same argument works for the rules 18 and 20.\begin{small}\begin{center}\begin{tabular}{lcl}
\AxiomC{$\Pi_{1}$} \RightLabel{$\Delta,N,u$}
\UnaryInfC{$\alpha^{\Sigma}$} \LeftLabel{rule 17:} \RightLabel{$\Delta,N,\bullet$}
\UnaryInfC{$\alpha^{\Sigma}$}
\AxiomC{[$\alpha^{\Sigma}$]} \RightLabel{$\Delta,N,v$}
\UnaryInfC{$\Pi_{2}$} \RightLabel{$\Theta$}
\UnaryInfC{$\beta^{\Omega}$} \LeftLabel{rule 19:} \RightLabel{$\Theta$}
\BinaryInfC{$\beta^{\Omega}$} \DisplayProof & $\rhd$ &
\AxiomC{$\Pi_{1}$} \RightLabel{$\Delta,N,u$}
\UnaryInfC{$\alpha^{\Sigma}$} \RightLabel{$\Delta,N,u$}
\UnaryInfC{$\Pi_{2} (u , v \mid w,u)$} \RightLabel{$\Theta (u , v \mid w,u)$}
\UnaryInfC{$\beta^{\Omega} (u , v \mid w,u)$} \DisplayProof\\&&\\
\AxiomC{$\Pi_{1}$} \RightLabel{$\Delta,N,\bullet$}
\UnaryInfC{$\alpha^{\Sigma}$}
\AxiomC{[$\alpha^{\Sigma}$]} \RightLabel{$\Delta,N,u$}
\UnaryInfC{$\alpha^{\Sigma}$} \LeftLabel{rule 17:} \RightLabel{$\Delta,N,\bullet$}
\UnaryInfC{$\alpha^{\Sigma}$}
\UnaryInfC{$\Pi_{2}$} \RightLabel{$\Theta$}
\UnaryInfC{$\beta^{\Omega}$} \LeftLabel{rule 19:} \RightLabel{$\Theta$}
\BinaryInfC{$\beta^{\Omega}$} \DisplayProof
 & $\rhd$ &
\AxiomC{$\Pi_{1}$} \RightLabel{$\Delta,N,\bullet$}
\UnaryInfC{$\alpha^{\Sigma}$}
\UnaryInfC{$\Pi_{2}$} \RightLabel{$\Theta$}
\UnaryInfC{$\beta^{\Omega}$} \DisplayProof 
\end{tabular}\end{center}\end{small}

The introduction of the rules 23 to 26 may produce no maximum formula but they produce unnecessary detours. We repeat the above arguments to eliminate them. The reduction for rule 24 is similar to the reduction for rule 23 and the reductions for rule 26 are similar to the reductions for rule 25. For rules 25 and 26 the reductions depend on the size of the cycles built to recover the same formula in the same context. We present only the case for a cycle of size 3. The rules 27 to 30 produce no maximum formula nor any unnecessary detour.\\[5pt]
\begin{tabular}{lcr}
\AxiomC{$\Pi_{1}$} \RightLabel{$\Delta,N$}
\UnaryInfC{$\alpha^{\Sigma,\bullet}$}
\AxiomC{$\Pi_{2}$} \RightLabel{$\Delta,M$}
\UnaryInfC{$\shneg N$} \LeftLabel{rule 23:} \RightLabel{$\Delta,M$}
\BinaryInfC{$\alpha^{\Sigma,\bullet}$} \RightLabel{$\Delta,M$}
\AxiomC{$\Pi_{3}$} \RightLabel{$\Delta,N$}
\UnaryInfC{$\shneg M$} \LeftLabel{rule 23:} \RightLabel{$\Delta,N$}
\BinaryInfC{$\alpha^{\Sigma,\bullet}$} \RightLabel{$\Delta,N$}
\UnaryInfC{$\Pi_{4}$}
\alwaysNoLine \UnaryInfC{$\phantom{-}$} \DisplayProof & $\rhd$ &
\AxiomC{$\Pi_{1}$} \RightLabel{$\Delta,N$}
\UnaryInfC{$\alpha^{\Sigma,\bullet}$} \RightLabel{$\Delta,N$}
\UnaryInfC{$\Pi_{4}$} \DisplayProof\\&&\\
\AxiomC{$\Pi_{1}$} \RightLabel{$\Delta,N$}
\UnaryInfC{$\shneg M$}
\AxiomC{$\Pi_{2}$} \RightLabel{$\Delta,M$}
\UnaryInfC{$\shneg P$} \LeftLabel{rule 25:} \RightLabel{$\Delta,N$}
\BinaryInfC{$\shneg P$} \RightLabel{$\Delta,N$}
\AxiomC{$\Pi_{3}$} \RightLabel{$\Delta,P$}
\UnaryInfC{$\shneg Q$} \LeftLabel{rule 25:} \RightLabel{$\Delta,N$}
\BinaryInfC{$\shneg Q$} \RightLabel{$\Delta,N$}
\AxiomC{$\Pi_{4}$} \RightLabel{$\Delta,Q$}
\UnaryInfC{$\shneg M$} \LeftLabel{rule 25:} \RightLabel{$\Delta,N$}
\BinaryInfC{$\shneg M$} \RightLabel{$\Delta,N$}
\UnaryInfC{$\Pi_{5}$}  \DisplayProof & $\rhd$ &
\AxiomC{$\Pi_{1}$} \RightLabel{$\Delta,N$}
\UnaryInfC{$\shneg M$} \RightLabel{$\Delta,N$}
\UnaryInfC{$\Pi_{5}$} \DisplayProof
\end{tabular}\end{proof}

\begin{definition}
Given a wff $\alpha^{\Sigma}$, the label rank $\aleph(\alpha^{\Sigma})$ is the depth of label nesting:
\begin{enumerate}
\item $\aleph(\alpha^{\Sigma}) = \aleph(\alpha) + s(\Sigma)/2$;
\item If $\alpha^{\Sigma} = \beta^{\Omega} \vee \gamma^{\Theta}$, then $\aleph(\alpha^{\Sigma}) = \max(\aleph(\beta^{\Omega}),\aleph(\gamma^{\Theta}))$;
\item If $\alpha^{\Sigma} = \beta^{\Omega} \wedge \gamma^{\Theta}$, then $\aleph(\alpha^{\Sigma}) = \max(\aleph(\beta^{\Omega}),\aleph(\gamma^{\Theta}))$;
\item If $\alpha^{\Sigma} = \beta^{\Omega} \ra \gamma^{\Theta}$, then $\aleph(\alpha^{\Sigma}) = \max(\aleph(\beta^{\Omega}),\aleph(\gamma^{\Theta}))$;
\item If $\alpha^{\Sigma} = \neg \beta^{\Omega}$, then $\aleph(\alpha^{\Sigma}) = \aleph(\beta^{\Omega})$;
\end{enumerate}
\end{definition}

\noindent Remark: by definition, the rank for a wff in $\bs{F}_{n}$ must be a natural number.

\begin{lemma} \label{depthModel}
Given a model $\m{M} = \ob \m{W} , \$ , \m{V} , \chi \cb$ and a $\alpha^{\Sigma} \in \bs{F}_{n}$, if $\aleph(\alpha^{\Sigma}) = k$, then we only need to verify the worlds of $\bigtriangleup^{\$}_{\vec{k}}(\chi)$ to know if $\m{M} \models \alpha^{\Sigma}$ holds.
\end{lemma}

\begin{proof}
If $\aleph(\alpha^{\Sigma}) = 0$, then $\alpha^{\Sigma}$ is a propositional formula. In this case, we need only to verify that the formula holds at $\bigtriangleup^{\$}_{\vec{0}}(\chi) = \{\chi\}$. If $\aleph(\alpha^{\Sigma}) = k + 1$, then it must have a subformula of the form $(\beta^{\Omega})^{\phi}$, where $\phi$ is a neighbourhood label. In the worst case, we need to verify all neighbourhoods of $\$(\chi)$ to assure that the property described by $\beta^{\Omega}$ holds in all of them. $\beta^{\Omega}$ must have a subformula of the form $(\gamma^{\Theta})^{\psi}$, where $\psi$ is a world label. In the worst case, we need to verify all worlds of $\$(\chi)$ to ensure that the property described by $\gamma^{\Theta}$ holds in all of them. But $\aleph(\gamma^{\Theta}) = k$ and, by the induction hypothesis, we need only to verify in the worlds of $\bigtriangleup^{\$}_{\vec{k}}(w)$, for every $w \in \bigtriangleup^{\$}_{1}(\chi)$. So we need, at the worst case, to verify the worlds of $\bigtriangleup^{\$}_{\vec{k+1}}(w)$.
\end{proof}

\begin{lemma} \label{finiteVerification}
If $\m{M} = \ob \m{W} , \$ , \m{V} , \chi \cb \models \alpha^{\Sigma}$, then there is a finite model $\m{M}' = \ob \m{W}' , \$' , \m{V}' , \chi' \cb$, such that $\m{M}' \models \alpha^{\Sigma}$.
\end{lemma}

\begin{proof}
In the proof of lemma \ref{consistentModel}, we verified the pertinence of the formulas in maximally n-consistent sets and maximally w-consistent sets based on the structure of the given formula to stablish the satisfying relation. Each existential label required the existence of one neighbourhood or world for the verification of the validity of a given subformula. The universal label for neighbourhood required no neighbourhood at all. It only added properties to the neighbourhoods that exist in a given system of neighbourhoods. The procedure is a demonstration that, for any wff in $\bs{F}_{n}$, we only need to gather a finite set of neighbourhoods and worlds.
\end{proof}

\begin{theorem} \label{decidability}
PUC-Logic is decidable.
\end{theorem}

\begin{proof}
If $\not\vdash \alpha^{\Sigma}$, then it must be possible to find a template that satisfies the negation of the formula. By the lemma above, there is a finite template that satisfies this negation.\end{proof}

\begin{definition}
Every label occurrence $\phi$ inside a formula $\alpha^{\Sigma}$ is an index of a subformula $\beta^{\Omega,\phi}$. Every label occurrence $\phi$ has a relative label depth defined by $\flat(\phi) = \aleph(\alpha^{\Sigma}) - \aleph(\beta^{\Omega,\phi})$.
\end{definition}

\begin{lemma} \label{complexity}
Given $\alpha^{\Sigma} \in \bs{F}_{n}$, there is a finite model $\m{M} = \ob \m{W} , \$ , \m{V} , \chi \cb$, such that $\m{M} \models \alpha^{\Sigma}$ with the following properties: (a) $\m{W} = \bigtriangleup^{\$}_{\vec{k}}(\chi)$, where $k = \aleph(\alpha^{\Sigma})$; (b) For every world $w \in \bigtriangleup^{\$}_{n}(\chi)$, $\$(w)$ has at most the same number of neighbourhoods as labels $\phi$, such that $\flat(\phi)=n$; (c) Every neighbourhood $N \in \$(w)$ has at most the same number of worlds as the labels $\phi$, such that $\flat(\phi)=n+1/2$, plus the number of labels $\varphi$, such that $\flat(\varphi)=n$.
\end{lemma}

\begin{proof}
(a) From lemmas \ref{finiteVerification} and \ref{depthModel}; (b) Every neighbourhood existential label $\phi$, such that $\flat(\phi) = 0$ contribute, by the procedure of lemma \ref{consistentModel}, to one neighbourhood to $\$(\chi)$ for the model $\m{M} = \ob \m{W} , \$ , \m{V} , \chi \cb$. The neighbourhood universal requires no additional neighbourhood to $\$(\chi)$ according to the explanation of lemma \ref{finiteVerification}. In the worst case, all neighbourhood labels $\phi$, such that $\flat(\phi) = 0$, are existential. The labels $\phi$, such that $\flat(\phi) = n$, $n \geq 0$, $n \in \mathbb{N}$ contributes to the systems of neighbourhoods of the worlds of $\bigtriangleup^{\$}_{n}(\chi)$. In the worst case, all of this labels contributes to system of neighbourhoods of a single world; (c) The same argument works for number of worlds in a neighbourhood except that the number of worlds in a neighbourhood is bigger than the number worlds in every neighbourhoods it contains. In the worst case, the smallest neighbourhood contains the same number of worlds as the number of labels $\phi$, such that $\flat(\phi)=n+1/2$. In this case, we must add at least one world to each neighbourhood that contains the smallest neighbourhood in the considered system of neighbourhoods. But the number of neighbourhoods is limited by the number of labels $\flat(\phi) = n$, $n \in \mathbb{N}$. So, the biggest neighbourhood reaches the asserted limit and the number of worlds of the model is linear in the number of labels.
\end{proof}

\begin{theorem} \label{satisfabilityPUC}
The problem of satisfiability is $\bs{NP}$-complete for PUC-Logic.
\end{theorem}

\begin{proof}
A wff without labels is a propositional formula, then, by \cite{Cook}, the complexity of the satisfiability problem for PUC-Logic must be a least $\bs{NP}$-complete. Given a wff with labels, by lemma \ref{complexity}, we know that there is a directed graph, in the manner of lemma \ref{lemmaConsequence}, that depends on the satisfiability of the endpoints. Those endpoints are always propositional formulas. So, the complexity of the problem of satisfiability is the sum of complexities of the problems for each endpoint. It means that the biggest subformula dictates the complexity because the model of lemma \ref{complexity} has at most a linear number of worlds and the satisfiability problem is $\bs{NP}$-complete. So, the worst case is the wff without labels.
\end{proof}

\section{Counterfactual logics}

In \cite{Lewis}, Lewis presents many logics for counterfactual reasoning, organized according to some given conditions imposed on the nested neighbourhood function. The most basic logic is $\bs{V}$, which has no condition imposed on $\$$. Lewis presented the axioms and inference rules of $\bs{V}$ using his comparative possibility operator ($\cp$).

\begin{definition}
$\alpha^{\Sigma} \cp \beta^{\Omega} \equiv (\beta^{\Omega,\bullet} \ra \alpha^{\Sigma,\bullet})^{\circledast}$
\end{definition}

Here we prove that the the axioms of the $\bs{V}$-logic are theorems and that the inference rules are derived rules in PUC-Logic. This is proof that the PUC-Logic is complete for the $\bs{V}$-logic based on the completeness proof of completeness given by Lewis\cite{Lewis}.

\begin{itemize}
\item TRANS axiom: $((\alpha \cp \beta) \wedge (\beta \cp \gamma)) \ra (\alpha \cp \gamma)$;
\item CONNEX axiom: $(\alpha \cp \beta) \vee (\beta \cp \alpha)$;
\item Comparative Possibility Rule (CPR): If $\vdash \alpha \ra (\beta_{1} \vee \ldots \vee \beta_{n})$, then $\vdash (\beta_{1} \cp \alpha) \vee \ldots \vee (\beta_{n} \cp \alpha)$, for any $n \geq 1$.
\end{itemize}

We present a proof of the CPR rule for $n = 2$. We omit the attribute representation of the wff denoted by $\alpha$, $\beta$ and $\gamma$ to simplify the reading of the derivations. We use lemma \ref{transfer} below for the theorem $\alpha \ra (\beta \vee \gamma)$ and a derivation $\Xi$ of it.\begin{center}\AxiomC{$^{2}$[$\gamma^ {\bullet}$]} \RightLabel{$\circledast$}
\UnaryInfC{$\gamma^ {\bullet}$}
\AxiomC{$\phantom{.}$}\alwaysNoLine
\UnaryInfC{$^{1}$[$(\beta^{\bullet} \ra \alpha^{\bullet})^{\circledast}\wedge(\gamma^{\bullet} \ra \beta^{\bullet})^{\circledast}$]} \alwaysSingleLine
\UnaryInfC{$(\beta^{\bullet} \ra \alpha^{\bullet})^{\circledast}\wedge(\gamma^{\bullet} \ra \beta^{\bullet})^{\circledast}$}
\UnaryInfC{$(\gamma^{\bullet} \ra \beta^{\bullet})^{\circledast}$} \RightLabel{$\circledast$}
\UnaryInfC{$\gamma^{\bullet} \ra \beta^{\bullet}$} \RightLabel{$\circledast$}
\BinaryInfC{$\beta^{\bullet}$}
\AxiomC{$^{1}$[$(\beta^{\bullet} \ra \alpha^{\bullet})^{\circledast}\wedge(\gamma^{\bullet} \ra \beta^{\bullet})^{\circledast}$]}
\UnaryInfC{$(\beta^{\bullet} \ra \alpha^{\bullet})^{\circledast}\wedge(\gamma^{\bullet} \ra \beta^{\bullet})^{\circledast}$}
\UnaryInfC{$(\beta^{\bullet} \ra \alpha^{\bullet})^{\circledast}$} \RightLabel{$\circledast$}
\UnaryInfC{$\beta^{\bullet} \ra \alpha^{\bullet}$} \RightLabel{$\circledast$}
\BinaryInfC{$\alpha^{\bullet}$} \LeftLabel{\textbf{TRANS} \hspace*{4cm} \scriptsize{2}} \RightLabel{$\circledast$}
\UnaryInfC{$\gamma^{\bullet} \ra \alpha^{\bullet}$} \RightLabel{$\circledast$}
\UnaryInfC{$(\gamma^{\bullet} \ra \alpha^{\bullet})^{\circledast}$} \LeftLabel{\scriptsize{1}}
\UnaryInfC{$((\beta^{\bullet} \ra \alpha^{\bullet})^{\circledast}\wedge(\gamma^{\bullet} \ra \beta^{\bullet})^{\circledast}) \ra (\gamma^{\bullet} \ra \alpha^{\bullet})^{\circledast}$} \DisplayProof\end{center}

\begin{landscape}\begin{center}\AxiomC{$^{1}$[$\neg ((\beta^{\bullet} \ra \alpha^{\bullet})^{\circledast}\vee(\alpha^{\bullet} \ra \beta^{\bullet})^{\circledast})$]}
\UnaryInfC{$\neg ((\beta^{\bullet} \ra \alpha^{\bullet})^{\circledast}\vee(\alpha^{\bullet} \ra \beta^{\bullet})^{\circledast})$}
\AxiomC{$^{1}$[$\neg ((\beta^{\bullet} \ra \alpha^{\bullet})^{\circledast}\vee(\alpha^{\bullet} \ra \beta^{\bullet})^{\circledast})$]}
\UnaryInfC{$\neg ((\beta^{\bullet} \ra \alpha^{\bullet})^{\circledast}\vee(\alpha^{\bullet} \ra \beta^{\bullet})^{\circledast})$}
\AxiomC{$^{2}$[$\beta^{\bullet}$]} \RightLabel{$\circledast$}
\UnaryInfC{$\beta^{\bullet}$} \RightLabel{$\circledast$}
\UnaryInfC{$\alpha^{\bullet} \ra \beta^{\bullet}$}
\UnaryInfC{$(\alpha^{\bullet} \ra \beta^{\bullet})^{\circledast}$}
\UnaryInfC{$(\beta^{\bullet} \ra \alpha^{\bullet})^{\circledast}\vee(\alpha^{\bullet} \ra \beta^{\bullet})^{\circledast}$}
\BinaryInfC{$\bot_{n}$} 
\UnaryInfC{$\alpha^{\bullet,\circledast}$} \RightLabel{$\circledast$}
\UnaryInfC{$\alpha^{\bullet}$} \LeftLabel{\scriptsize{2}} \RightLabel{$\circledast$}
\UnaryInfC{$\beta^{\bullet} \ra \alpha^{\bullet}$}
\UnaryInfC{$(\beta^{\bullet} \ra \alpha^{\bullet})^{\circledast}$}
\UnaryInfC{$(\beta^{\bullet} \ra \alpha^{\bullet})^{\circledast}\vee(\alpha^{\bullet} \ra \beta^{\bullet})^{\circledast}$}
\BinaryInfC{$\bot_{n}$} \LeftLabel{\textbf{CONNEX} \hspace*{2cm} \scriptsize{1}}
\UnaryInfC{$(\beta^{\bullet} \ra \alpha^{\bullet})^{\circledast}\vee(\alpha^{\bullet} \ra \beta^{\bullet})^{\circledast}$}
\alwaysNoLine
\UnaryInfC{$\phantom{.}$}
\UnaryInfC{$\phantom{.}$}\DisplayProof

\AxiomC{$^{1}$[$\neg ((\alpha^{\bullet} \ra \beta^{\bullet})^{\circledast} \vee(\alpha^{\bullet} \ra \gamma^{\bullet})^{\circledast})$]}
\UnaryInfC{$\neg ((\alpha^{\bullet} \ra \beta^{\bullet})^{\circledast} \vee(\alpha^{\bullet} \ra \gamma^{\bullet})^{\circledast})$}
\AxiomC{$^{1}$[$\neg ((\alpha^{\bullet} \ra \beta^{\bullet})^{\circledast} \vee(\alpha^{\bullet} \ra \gamma^{\bullet})^{\circledast})$]}
\UnaryInfC{$\neg ((\alpha^{\bullet} \ra \beta^{\bullet})^{\circledast} \vee(\alpha^{\bullet} \ra \gamma^{\bullet})^{\circledast})$}
\AxiomC{$^{2}$[$\alpha^{\bullet}$]} \RightLabel{$N$}
\UnaryInfC{$\alpha^{\bullet}$}
\UnaryInfC{$\Sigma$}
\UnaryInfC{$(\alpha^{\bullet} \ra \beta^{\bullet})^{\circledast} \vee(\alpha^{\bullet} \ra \gamma^{\bullet})^{\circledast}$}
\BinaryInfC{$\bot_{n}$}
\UnaryInfC{$\beta^{\bullet,N}$} \RightLabel{$N$}
\UnaryInfC{$\beta^{\bullet}$} \LeftLabel{2} \RightLabel{$N$}
\UnaryInfC{$\alpha^{\bullet} \ra \beta^{\bullet}$} \RightLabel{$\circledast$}
\UnaryInfC{$\alpha^{\bullet} \ra \beta^{\bullet}$}
\UnaryInfC{$(\alpha^{\bullet} \ra \beta^{\bullet})^{\circledast}$}
\UnaryInfC{$(\alpha^{\bullet} \ra \beta^{\bullet})^{\circledast} \vee(\alpha^{\bullet} \ra \gamma^{\bullet})^{\circledast}$}
\BinaryInfC{$\bot_{n}$} \LeftLabel{\textbf{CPR} \hspace*{2cm}1}
\UnaryInfC{$(\alpha^{\bullet} \ra \beta^{\bullet})^{\circledast} \vee(\alpha^{\bullet} \ra \gamma^{\bullet})^{\circledast}$}\DisplayProof

\AxiomC{$\alpha^{\bullet}$} \RightLabel{$N$}
\UnaryInfC{$\alpha^{\bullet}$} \RightLabel{$N,\bullet$}
\AxiomC{$\Xi$} \RightLabel{$N,u$}
\UnaryInfC{$\alpha \ra (\beta \vee \gamma)$}
\BinaryInfC{$\Pi$} \RightLabel{$\circledast$}
\UnaryInfC{$\beta^{\bullet} \vee \gamma^{\bullet} $}
\AxiomC{$^{3}$[$\beta^{\bullet}$]} \RightLabel{$\circledast$}
\UnaryInfC{$\beta^{\bullet}$} \RightLabel{$\circledast$}
\UnaryInfC{$\alpha^{\bullet} \ra \beta^{\bullet}$}
\UnaryInfC{$(\alpha^{\bullet} \ra \beta^{\bullet})^{\circledast}$}
\UnaryInfC{$(\alpha^{\bullet} \ra \beta^{\bullet})^{\circledast} \vee(\alpha^{\bullet} \ra \gamma^{\bullet})^{\circledast}$}
\AxiomC{$^{3}$[$\gamma^{\bullet}$]} \RightLabel{$\circledast$}
\UnaryInfC{$\gamma^{\bullet}$} \RightLabel{$\circledast$}
\UnaryInfC{$\alpha^{\bullet} \ra \gamma^{\bullet} $}
\UnaryInfC{$(\alpha^{\bullet} \ra \gamma^{\bullet})^{\circledast}$}
\UnaryInfC{$(\alpha^{\bullet} \ra \beta^{\bullet})^{\circledast} \vee(\alpha^{\bullet} \ra \gamma^{\bullet})^{\circledast}$} \LeftLabel{$\bs{\Sigma}$ \hspace*{0.5cm} 3}
\TrinaryInfC{$(\alpha^{\bullet} \ra \beta^{\bullet})^{\circledast} \vee(\alpha^{\bullet} \ra \gamma^{\bullet})^{\circledast}$}
\alwaysNoLine
\UnaryInfC{$\phantom{.}$}
\UnaryInfC{$\phantom{.}$}
\UnaryInfC{$\phantom{.}$}\DisplayProof

\AxiomC{$\alpha^{\bullet}$} \RightLabel{$N$}
\UnaryInfC{$\alpha^{\bullet}$} \RightLabel{$N,\bullet$}
\UnaryInfC{$\alpha$}
\AxiomC{$^{4}$[$\alpha$]} \RightLabel{$N,u$}
\UnaryInfC{$\alpha$}
\AxiomC{$\Xi$} \RightLabel{$N,u$}
\UnaryInfC{$\alpha \ra (\beta \vee \gamma)$}
\BinaryInfC{$\beta \vee \gamma$} \RightLabel{$N,\bullet$}
\UnaryInfC{$\beta \vee \gamma$} \RightLabel{$\circledast,\bullet$}
\UnaryInfC{$\beta \vee \gamma$} \LeftLabel{4} \RightLabel{$\circledast,\bullet$}
\BinaryInfC{$\beta \vee \gamma$} \RightLabel{$\circledast,\bullet$}
\UnaryInfC{$\beta \vee \gamma$}
\AxiomC{[$\beta$]} \RightLabel{$\circledast,\bullet$}
\UnaryInfC{$\beta$} \RightLabel{$\circledast$}
\UnaryInfC{$\beta^{\bullet}$} \RightLabel{$\circledast$}
\UnaryInfC{$\beta^{\bullet} \vee \gamma^{\bullet}$}
\AxiomC{[$\gamma$]} \RightLabel{$\circledast,\bullet$}
\UnaryInfC{$\gamma$} \RightLabel{$\circledast$}
\UnaryInfC{$\gamma^{\bullet}$} \RightLabel{$\circledast$}
\UnaryInfC{$\beta^{\bullet} \vee \gamma^{\bullet}$} \RightLabel{$\circledast$}
\TrinaryInfC{$\beta^{\bullet} \vee \gamma^{\bullet}$} \LeftLabel{$\bs{\Pi}$ \hspace*{0.5cm}} \RightLabel{$\circledast$}
\UnaryInfC{$\beta^{\bullet} \vee \gamma^{\bullet} $} \DisplayProof\end{center}\end{landscape}

\begin{lemma} \label{transfer}
Given a theorem $\alpha^{\Sigma}$, there is a proof of $\alpha^{\Sigma}$ in the context $\{N,u\}$, in which the variables $N$ and $u$ do not occur in the proof.
\end{lemma}

\begin{proof}
$\alpha^{\Sigma}$ is a theorem, then, by definition, there is a proof $\Pi$ without open hypothesis that concludes the theorem in the empty context. During the proof $\Pi$, the smallest context is the empty context. So, if we can choose variables that do not occur in $\Pi$ and add the stack of labels $\{N,u\}$ at the rightmost position of each context of each rule. We end up with a proof of the theorem in the context $\{N,u\}$. This is possible because there is no restriction that could be applied over the new variables.\end{proof}

We now present some ideas related to the different counterfactual logics Lewis defined, based on conditions imposed to the function $\$$:
\begin{itemize}
\item Normality (N): $\$$ is normal iff $\forall w \in \m{W} : \$(w) \neq \emptyset$;
\item Total reflexivity (T): $\$$ is totally reflexive iff $\forall w \in \m{W} : w \in \bigcup\$(w)$;
\item Weak centering (W): $\$$ is weakly centered iff $\forall w \in \m{W} : \$(w) \neq \emptyset \mbox{ and }\forall N \in \bigcup\$(w) : w \in N$ ;
\item Centering (C): $\$$ is centered iff $\forall w \in \m{W} : \{w\} \in \$(w)$.
\end{itemize}

To each condition, corresponds a logic, respectively $\bs{VN}$, $\bs{VT}$, $\bs{VW}$ and $\bs{VC}$-logics. For each logic, the PUC-ND may change the set of rules to acquire the corresponding expressivity provided by the conditions. We present some ideas to make those changes:
\begin{itemize}
\item[$\bs{VN}$] Rule 9 looses restriction (a). Rule 19 and 22 loose second premiss.\\Introduction of the rule: \AxiomC{$\phantom{-}$} \RightLabel{$\Delta,\circledast$}
\UnaryInfC{$\alpha^{\Sigma}$} \RightLabel{$\Delta,N$}
\UnaryInfC{$\alpha^{\Sigma}$} \DisplayProof\\Restriction: (a) $\alpha^{\Sigma}$ must fit into the contexts;

\item[$\bs{VT}$] We repeat the system for VN.\\Introduction of the rule: \AxiomC{$\phantom{-}$} \RightLabel{$\Delta,\circledast,\ast$}
\UnaryInfC{$\alpha^{\Sigma}$} \RightLabel{$\Delta$}
\UnaryInfC{$\alpha^{\Sigma}$} \DisplayProof\\Restriction: (a) $\alpha^{\Sigma}$ must fit into the contexts;

\item[$\bs{VW}$] We repeat the system for VT.\\Introduction of the rule: \AxiomC{$\phantom{-}$} \RightLabel{$\Delta,\circledcirc,\ast$}
\UnaryInfC{$\alpha^{\Sigma}$} \RightLabel{$\Delta$}
\UnaryInfC{$\alpha^{\Sigma}$} \DisplayProof\\Restriction: (a) $\alpha^{\Sigma}$ must fit into the contexts;

\item[$\bs{VC}$] We repeat the system for VW.\\Introduction of the rule: \AxiomC{$\phantom{-}$} \RightLabel{$\Delta,\circledast,\bullet$}
\UnaryInfC{$\alpha^{\Sigma}$} \RightLabel{$\Delta$}
\UnaryInfC{$\alpha^{\Sigma}$} \DisplayProof\\Restriction: (a) $\alpha^{\Sigma}$ must fit into the contexts.
\end{itemize}

\section{Related Works}

As far as we know, there is only one natural deduction system for the counterfactuals, which is given by Bonevac \cite{Bonevac}. But his system is designed to deal with the $\bs{VW}$-logic, since it contains the rule of counterfactual exploitation ($\cf$E), which encapsulates the weak centering condition. His approach to define rules for the counterfactual operators provides a better intuition of the counterfactual logic. His systems is expressive enough to deal with modalities and strict conditionals. The labelling of world shifts using formulas makes it easier to capture the counterfactual mechanics.\\

We also found the work of Sano \cite{Sano} which pointed out the advantages of using the hybrid formalism for the counterfactual logic. He presented some axioms and rules for the $\bs{V_{\m{HC}(@)}}$-logic that extends the $\bs{V}$-logic of Lewis.\\

We also found a sequent calculus for the $\bs{V}$-logic that is given by \cite{Jelia2012}. But this system also demands modalities in the syntax. As far as we know, our deduction system is the only one dealing with Lewis systems in a general form, that is, without using modalities in the syntax.

\section*{Conclusions}

From the definitions of Lewis \cite{Lewis} for the counterfactual logic, we define our natural deduction system, which is proven to be sound and complete for the $\bs{V}$-logic.\\

The use of two types of labels (neighbourhood and world labels) gave us the ability to manage different types of quantifications. The quantifications are largely used by the counterfactual operators definitions according to Lewis. That approach makes it possible to build the rules for the counterfactual operators as derived rules of the system.\\

Another advantage of that approach is that our natural deduction system is built without the use of modalities or strict conditionals, making it easier to take benefits from the well known propositional results such as normalization.


\section*{References}
\begin{thebibliography}{1}
\newcommand{\enquote}[1]{``#1''}

\bibitem{Lewis}
Lewis, D. K., \enquote{Counterfactuals}, Blackwell Publishing, 2008.

\bibitem{LewisPapers}
Lewis, D. K., \enquote{Papers in ethics and social philosophy}, Cambridge University Press, 2000.

\bibitem{Goodman}
Goodman, N., \enquote{Fact, Fiction, and Forecast}, 4th Edition, Harvard University Press, 1983.

\bibitem{Bell}
Bell, J.~L., \enquote{Toposes and Local Set Theories}, Dover Publications,
  2008.
  
\bibitem{Knuth}
  Knuth, D. E., \enquote{Semantics of context-free languages}, Mathematical Systems
Theory 2 (1968).

\bibitem{Goldblatt}
Goldblatt, R., \enquote{Topoi: The categorical analysis of logic}, Dover, 2006.

\bibitem{Goldblatt2}
Goldblatt, R., \enquote{Logics of time and computation}, CSLI lecture notes, 1992.

\bibitem{Prawitz}
Prawitz, D., \enquote{Natural Deduction: a proof-theoretical study}, Dover, 2006.

\bibitem{Naufel}
do~Amaral, F.~N. and E.~H. Haeusler, \enquote{Using the internal logic of a topos
  to model search spaces for problems}, Logic Journal of IGPL  (2007).

\bibitem{Hermann}
Menezes, P.~B. and E.~H. Haeusler, \enquote{Teoria das Categorias para Ci\^{e}ncia da Computa\c{c}\~{a}o}, Editora Sagra Luzatto, 2006.

\bibitem{Ramsey}
Ramsey, F.~P., \enquote{Philosophical papers}, Cambridge University Press,
  1990.

\bibitem{Gent} 
Gent, I. P., \enquote{A Sequent- or Tableau-style System for Lewis's Counterfactual Logic VC}, Notre Dame Journal of Formal Logic, vol. 33, no. 3, pp. 369-382, 1992.

\bibitem{Bonevac}
Bonevac, D., \enquote{Deduction: Introductory Symbolic Logic}, Blackwell, 2003.

\bibitem{Sano}
Sano, K., \enquote{Hybrid counterfactual logics}, Journal of Logic, Language and Information, volume 18, No. 4, pp 515-539, 2009.

\bibitem{Escobar}
L\'{o}pez-Escobar, E.G.K., \enquote{Implicational Logics in Natural Deduction Systems}, Journal of Symbolic Logic, Vol. 47, No. 1, pp. 184-186, 1982

\bibitem{Fernandes}
Fernandes, R.Q.A., Haeusler, E.H., Pereira, L.C.P.D., \enquote{A Natural Deduction System for Counterfactual Logic}, in XVI Encontro Brasileiro de L\'{o}gica, Petr\'{o}polis, 2011.

\bibitem{LSFA09}
Fernandes, R.Q.A., Haeusler, E.H., \enquote{A Topos-Theoretic Approach to Counterfactual Logic}, in Fourth Workshop on Logical and Semantic Frameworks, Bras\'{i}lia, 2009. Pre-proceedings, 2009.

\bibitem{Hansson}
Hansson, B., \enquote{An Analysis of some Deontic Logics}, No\^{u}s, Vol. 3, No. 4, pp. 373-398, 1969.

\bibitem{vanDalen}
van Dalen, D., \enquote{Logic and Structure}, Springer, 2008.

\bibitem{Libkin}
Libkin, L., \enquote{Elements of Finite Model Theory}, Springer, 2010.

\bibitem{Troelstra}
Troelstra, A. S., Schwichtenberg, H., \enquote{Basic Proof Theory}, Cambridge University Press, 2000.

\bibitem{Lambert}
Lambert, K., \enquote{Free Logic: selected essays}, Cambridge University Press, 2004.

\bibitem{Cook}
Cook, S. A., \enquote{The complexity of theorem proving procedures}, In 3rd Annual ACM
Symposium on the Theory of Computation, pages 151-158, 1971.

\bibitem{Statman}
Statman, R., \enquote{Intuitioinistic propositional logic is polinomial-space complete}, Journal of Theoretical Computer Science, vol. 9, no. 1, pp. 67-72, 1979.

\bibitem{Jelia2012}
Lellmann, B., Pattinson, D., \enquote{Sequent Systems for Lewis' Conditional Logics}, In 13th European Conference on Logics in Artificial Intelligence, 2012.

\end{thebibliography}
\end{document}